\documentclass[12pt, a4paper]{report}

\usepackage[a4paper]{geometry}

\usepackage{pgfplots}
\pgfplotsset{compat=1.14,
    x tick style={color=black},
    y tick style={color=black}
}

\usepackage{url}
\usepackage{listings}

\renewcommand{\labelenumi}{(\alph{enumi})}
\renewcommand\theenumi\labelenumi

\usepackage[english]{babel}

\newcommand{\iemph}[1]{\emph{#1}\index{#1}}

\usepackage[utf8]{inputenc}
\usepackage{xspace}
\usepackage{amsmath,amsthm,amssymb,mathtools}
\usepackage{lmodern}

\usepackage[algo2e,ruled,vlined,linesnumbered,algochapter]{algorithm2e}
\newcommand{\assign}{\leftarrow}

\usepackage{xcolor}
\usepackage{tikz}
\usepackage{graphicx}

\usepackage{hyperref}

\allowdisplaybreaks[2]
\newtheorem{theorem}{Theorem}[section]
\newtheorem{lemma}[theorem]{Lemma}
\newtheorem{corollary}[theorem]{Corollary}
\newtheorem{definition}[theorem]{Definition}
\numberwithin{equation}{section}

\newcommand{\oea}{$(1 + 1)$~EA\xspace}
\newcommand{\oeamu}{$(1 + 1)$~EA$_\mu$\xspace}

\newcommand{\oplea}{$(1+\lambda)$~EA\xspace}
\newcommand{\mpoea}{$(\mu+1)$~EA\xspace}
\newcommand{\mplea}{$(\mu+\lambda)$~EA\xspace}

\newcommand{\opllga}{$(1+(\lambda,\lambda))$~GA\xspace}

\newcommand{\om}{\textsc{OneMax}\xspace}
\newcommand{\onemax}{\om}
\newcommand{\lo}{\textsc{LeadingOnes}\xspace}
\newcommand{\leadingones}{\lo}

\newcommand{\R}{\ensuremath{\mathbb{R}}}

\newcommand{\N}{\ensuremath{\mathbb{N}}} 
\newcommand{\Z}{\ensuremath{\mathbb{Z}}}

\newcommand{\calA}{\ensuremath{\mathcal{A}}} 
	
\DeclareMathOperator{\Bin}{Bin}
\DeclareMathOperator{\Geom}{Geom}
\DeclareMathOperator{\arcsinh}{arcsinh}
\newcommand{\pmin}{p_{\mathrm{min}}}

\newcommand{\Var}{\mathrm{Var}\xspace} 
\newcommand{\Cov}{\mathrm{Cov}\xspace} 
\newcommand{\eps}{\varepsilon}

\begin{document}
{\sloppy

\title{Probabilistic Tools for the Analysis of Randomized Optimization Heuristics\thanks{This is the author-created version of the book chapter~\cite{Doerr20bookchapter}. From arxiv version~3 on, the arxiv and book versions are identical in content, only the type-setting is different. In the fourth version of the arxiv version, the numbering of theorems, lemmas, equations etc.\ has been adjusted to the one of the book version by adding the chapter number~1. E.g., equation (10.35) became (1.10.35). In the fifth version, a mistake in the numbering (thanks to an anonymous GECCO 2021 reviewer for pointing it out!) has been corrected so that now the numberings of the arxiv and book version are really identical. In the sixth version, very minor typos have been corrected. The arxiv version is compiled with the alpha-bibstyle, which the author finds more practical to use.}}

\author{Benjamin Doerr\\ Laboratoire d'Informatique (LIX)\\ CNRS\\ \'Ecole Polytechnique\\ Institut Polytechnique de Paris\\ Palaiseau\\ France}

\maketitle

\begin{abstract}
  This chapter collects several probabilistic tools that proved to be useful in the analysis of randomized search heuristics. This includes classic material like Markov, Chebyshev and Chernoff inequalities, but also lesser known topics like stochastic domination and coupling or Chernoff bounds for geometrically distributed random variables and for negatively correlated random variables. 
  Most of the results presented here have appeared previously, some, however, only in recent conference publications. While the focus is on collecting tools for the analysis of randomized search heuristics, many of these may be useful as well in the analysis of classic randomized algorithms or discrete random structures.
\end{abstract}

%

\setcounter{chapter}{1}
\section{Introduction}

Unlike in the field of classic randomized algorithms for discrete optimization problems, where theory has always supported (and, in fact, often led) the development and understanding of new algorithms, the theoretical analysis of nature-inspired search heuristics is much younger than the use of these heuristics. The use of nature-inspired heuristics can easily be traced back to the 1960s, their rigorous analysis with proven performance guarantees only started in the late 1990s. Propelled by impressive results, most notably from the German computer scientist Ingo Wegener (*1950--$\dagger$2008) and his students, theoretical works became quickly accepted in the nature-inspired algorithms field and now form an integral part of it. They help to understand these algorithms, guide the choice of their parameters, and even (as in the classic algorithms field) suggest new promising algorithms. It is safe to say that Wegener's vision that nature-inspired heuristics are nothing more than a particular class of randomized algorithms, which therefore should be analyzed with the same rigor as other randomized algorithms, has come true. 

After around 20 years of theoretical analysis of nature-inspired algorithms, however, we have to note that the methods used here are different from those in the analysis of classic randomized algorithms. This is most visible for particular methods like the fitness level method or drift analysis, but applies even to the elementary probabilistic tools employed throughout the field. 

The aim of this chapter is to collect those elementary tools which often have been used in the past 20 years. This includes classic material like expectations, variances, the coupon collector process, Markov's inequality, Chebyshev's inequality and Chernoff-Hoeffding bounds for sums of independent random variables, but also topics that are used rarely outside the analysis of nature-inspired heuristics like stochastic domination, Chernoff-Hoeffding bounds for sums of independent geometrically distributed random variables, and Chernoff-Hoeffding bounds for sums of random variables which are not fully independent. For many results, we also sketch a typical application or refer to applications in the literature. 

The large majority of the results and applications presented in this work have appeared previously, some in textbooks, some in recent conference publications. The following results, while not necessarily very deep, to the best of our knowledge are original.
\begin{itemize}
\item The result that all known Chernoff bounds, when applied to binary random variables, hold as well for negatively correlated random variables. More precisely, for bounds on the upper tail, we only need $1$-negative correlation and for bounds on the lower tail, we only need $0$-negative correlation (Section~\ref{secprobnegcor}).
\item The insight that all commonly known Chernoff bounds can be deduced from only two bounds (Section~\ref{secprobrelation}).
\item A version of the method of bounded differences which only requires that the $t$-th random variable has a bounded influence on the expected result stemming from variable $t+1$ to $n$. This appears to be an interesting compromise between the classic method of bounded differences, which is hard to use for iterative algorithms, and martingale methods, which require the familiarity with martingales (Theorem~\ref{tprobboundedexp}).
\item Via an elementary two-stage rounding trick, we give simple proofs for the facts that (i)~a sum $X$ of independent binary random variables with $\Var[X] \ge 1$ exceeds its expectation with constant probability by at least $\Omega(\sqrt{\Var[X]})$ and (ii)~it attains a particular value at most with probability $2 / \sqrt{\Var[X]}$ (Lemmas~\ref{lprobsqrtn} and~\ref{lprobhit}). Both results have been proven before with deeper methods, e.g., an approximation via the normal distributions.
\end{itemize}

This chapter aims to serve both as introduction for newcomers to the field and as reference book for regular users of these methods. With both addressees in mind, we did not shy away from stating also elementary reformulations of the results or formulating explicitly statements that only rely on elementary mathematics like
\begin{itemize}
\item how to choose the deviation parameter $\delta$ in the strong multiplicative Chernoff bound so that the tail probability $(e / \delta)^\delta$ is below a desired value (Lemma~\ref{lprobsuperexp}), and
\item how to translate a tail bound into an expectation (Corollary~\ref{corprobtaile}).
\end{itemize}
We hope that this saves all users of this chapter some time, which is better spent on understanding the challenging random processes that arise in the analysis of nature-inspired heuristics.

\section{Notation}

All notation in this chapter is standard and should need not much additional explanation. We use $\N := \{1, 2, \dots\}$ to denote the positive integers. We write $\N_0 := \N \cup \{0\}$. For intervals of integers, we write $[a..b] := \{x \in \Z \mid a \le x \le b\}$. We use the standard definition $0^0 := 1$ (and not $0^0 = 0$).

\section{Elementary Probability Theory}\label{secprobprereq}

We shall assume that the reader has some basic understanding of the concepts of \emph{probability spaces}, \emph{events} and \emph{random variables}. As usual in probability theory and very convenient in analysis of algorithms, we shall almost never explicitly state the probability space we are working in. Hence an intuitive understanding of the notion of a random variable should be enough to follow this exposition. 

While many results presented in the following naturally extend to continuous probability spaces, in the interest of simplicity and accessibility for a discrete optimization audience, we shall assume that all random variables in this book will be \emph{discrete}, that is, they take at most a countable number of values. As a simple example, consider the random experiment of independently rolling two distinguishable dice. Let $X_1$ denote the outcome of the first roll, that is, the number between $1$ and $6$ which the first die displays. Likewise, let $X_2$ denote the outcome of the second roll. These are already two random variables. We formalize the statement that with probability $\frac 16$ the first die shows a one by saying $\Pr[X_1 = 1] = \frac 16$. Also, the probability that both dice show the same number is $\Pr[X_1 = X_2] = \frac 16$. The \emph{complementary event} that they show different numbers, naturally has a probability of $\Pr[X_1 \neq X_2] = 1 - \Pr[X_1 = X_2] = \frac 56$. 

We can add random variables (defined over the same probability space), e.g., $X := X_1 + X_2$ is the sum of the numbers shown by the two dice, and we can multiply a random variable by a number, e.g., $X := 2 X_1$ is twice the number shown by the first die. 

The most common type of random variable we shall encounter in this book is an extremely simple one called \emph{binary random variable}\index{random variable!binary} or \emph{Bernoulli random variable}\index{random variable!Bernoulli}. It takes the values $0$ and $1$ only. In consequence, the probability distribution of a binary random variable $X$ is fully described by its probability $\Pr[X = 1]$ of being one, since $\Pr[X=0] = 1 - \Pr[X=1]$. 

Binary random variables often show up as \emph{indicator random variables}\index{indicator random variable} for random events. For example, if the random experiment is a simple roll of a die, we may define a random variable $X$ by setting $X = 1$, if the die shows a 6, and $X=0$ otherwise. We say that $X$ is the indicator random variable for the event ``die shows a 6.'' 

Indicator random variables are useful for counting. If we roll a die $n$ times and $X_1, \ldots, X_n$ are the indicator random variables for the events that the corresponding roll showed a 6 (considered as a \emph{success}), then $\sum_{i = 1}^n X_i$ is a random variable describing the number of times we saw a 6 in these $n$ rolls. In general, a random variable $X$ that is the sum of $n$ independent binary random variables being one all with equal probability $p$, is called a \emph{binomial random variable}\index{random variable!binomial} (with success probability $p$). We denote this distribution by $\Bin(n,p)$ and write $X \sim \Bin(n,p)$ to denote that $X$ has this distribution. We have \[\Pr[X = k] = \binom{n}{k}p^k (1-p)^{n-k}\] for all $k \in [0..n]$. See Section~\ref{secprobstirling} for the definition of the binomial coefficient.

A different question is how long we have to wait until we roll a 6. Assume that we have an infinite sequence of die rolls and $X_1, X_2, \ldots$ are the indicator random variables for the event that the corresponding roll showed a 6 (\emph{success}). Then we are interested in the random variable $Y = \min\{k \in \N \mid X_k = 1\}$. Again for the general case of all $X_i$ being one independently with probability $p >0$, this random variable $Y$ is called \emph{geometric random variable}\index{random variable!geometric} (with success probability $p$). We denote this distribution by $\Geom(p)$ and write $Y \sim \Geom(p)$ to indicate that $Y$ is geometrically distributed (with parameter $p$). We have \[\Pr[Y = k] = (1-p)^{k-1} p\] for all $k \in \N$. We note that an equally established definition is to count only the failures, that is, to regard the random variable $Y -1$. So some care is necessary when comparing results from different sources.

\section{Useful Inequalities}

Before starting our presentation of probabilistic tools useful in the analysis of randomized search heuristics, let us brief{}ly mention a few inequalities that are often needed to estimate probabilities arising naturally in this area.

\subsection{Switching Between Exponential and Polynomial Terms}

When dealing with events occurring with small probability $\eps > 0$, we often encounter expressions like $(1-\eps)^n$. Such a mix of a polynomial term ($1-\eps$) with an exponentiation is often hard to work with. It is therefore very convenient that $1-\eps \approx e^{-\eps}$, so that the above expression becomes approximately the purely exponential term $e^{-\eps n}$. In this section, we collect a few estimates of this flavor. With the exception of the second inequality in~\eqref{eqprobweierminus}, a sharper version of a Weierstrass product inequality, all are well-known and can be derived via elementary arguments.

\begin{lemma}\label{lprobelower}
  For all $x \in \R$, 
  \begin{equation*}
    1+x \le e^x.
  \end{equation*}
\end{lemma}

We give a canonic proof as an example for a proof method that is often useful for such estimates.

\begin{proof}
  Define a function $f: \R \to \R$ by $f(x) = e^x - 1 - x$ for all $x \in \R$. Since $f'(x) = e^x - 1$, we have $f'(x) = 0$ if and only if $x=0$. Since $f''(x) = e^x > 0$ for all $x$, we see that $x=0$ is the unique minimum of $f$. Since $f(0) = 0$, we have $f(x) \ge 0$ for all $x$, which is equivalent to the claim of the lemma.
\end{proof} 

Applying Lemma~\ref{lprobelower} to $-x$ and taking reciprocals, we immediately derive the first of the following two upper bounds for the exponential function. The second bound again follows from elementary calculus. Obviously, the first estimate is better for $x < 0$, the second one is better for $x > 0$. 

\begin{lemma}\label{lprobeupper}
\begin{enumerate}
\item  For all $x < 1$, 
  \begin{equation}
    e^x \le \frac{1}{1-x} = 1 + \frac{x}{1-x} = 1 + x + \frac{x^2}{1-x}.\label{eqprobeupper1}
  \end{equation}
  In particular, for $0 \le x \le 1$, we have $e^{-x} \le 1 - \frac x2$.
\item For all $x < 1.79$, 
  \begin{equation}
    e^x \le 1 + x + x^2.
  \end{equation}
\end{enumerate}
\end{lemma}

\begin{figure}
\begin{tikzpicture}
\begin{axis}[
    axis lines = left,
    xlabel = $x$,
    ylabel = {$f(x)$},
    legend style={
      cells={anchor=east},
      legend pos=outer north east,},
    ymin=-0.25,
    ymax=3.5
]

\addplot [densely dashed,
    domain=-1.25:1.25, 
    samples=100, 
    color=black,
    ]
    {1+x+x^2};
\addlegendentry[right]{$f(x) = 1+x+x^2$}
 
%
%
\addplot [densely dotted,
    domain=-1.25:0.9, 
    samples=100, 
    color=black,
    ]
    {1/(1-x)};
\addlegendentry[right]{$f(x) = 1+x+ \frac{x^2}{1-x}$}
 
\addplot [
    domain=-1.25:1.25, 
    samples=100, 
    color=black,
]
{exp(x)};
\addlegendentry[right]{$f(x) = e^x$}

\addplot [densely dotted,
    domain=-1.25:1.25, 
    samples=100, 
    color=black,
    ]
    {1+x};
\addlegendentry[right]{$f(x) = 1+x$}
 
\end{axis}
\end{tikzpicture}
\caption{Plot of the estimates of Lemma~\ref{lprobelower} and~\ref{lprobeupper}.}\label{figprobe}
\end{figure}
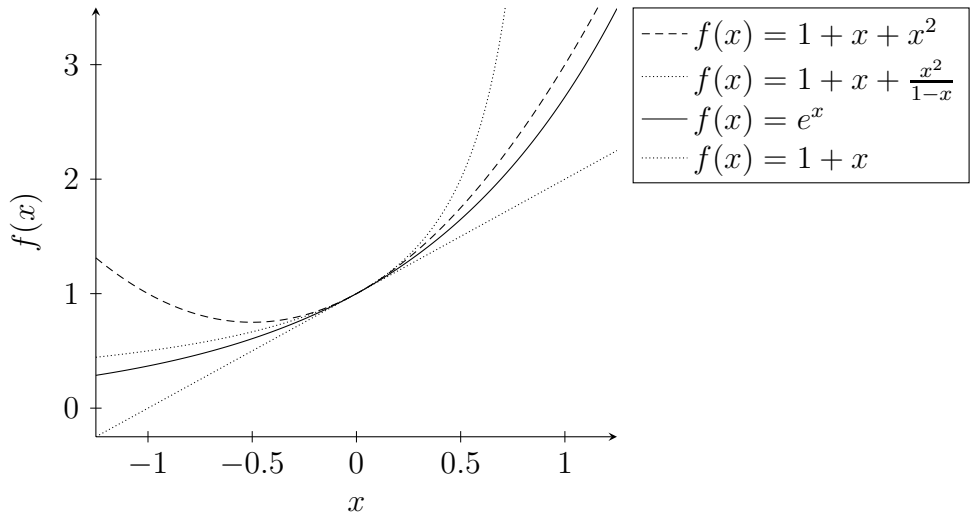

As visible also from Figure~\ref{figprobe}, these estimates are strongest for $x$ close to zero. 
Combining Lemma~\ref{lprobelower} and Lemma~\ref{lprobeupper}, the following useful estimate was obtained in~\cite[Lemma~8]{RoweS14}.

\begin{corollary}
  For all $x \in [0,1]$ and $y > 0$, $(1-x)^y \le \frac{1}{1+xy}$.
\end{corollary}

Replacing $x$ with $\frac{x}{1+x}$ in the first inequality of Lemma~\ref{lprobeupper} gives the following bounds.

\begin{corollary}\label{corprobe}
  For all $x > -1$, we have 
  \begin{equation}
  e^{\frac{x}{1+x}} \le 1+x \le e^x.\label{eqprobe1}
  \end{equation}
  For all $x, y > 0$, 
  \begin{equation}
  e^{\frac{xy}{x+y}} \le (1+\tfrac xy)^y \le e^x.
  \end{equation}
\end{corollary}

The first bound of~\eqref{eqprobe1} can, with different arguments and for a smaller range of $x$, be sharpened to the following estimate from~\cite[Lemma 8(c) of the arxiv version]{DoerrGWY17}.

\begin{lemma}
  For all $x \in [0,\frac 23]$, $e^{-x-x^2} \le 1-x$.
\end{lemma}

A reformulation of~\eqref{eqprobe1} often useful on the context of standard-bit mutation\index{standard-bit mutation} (mutating a bit-string by flipping each bit independently with a small probability like $\frac 1n$) is the following. Note that the first bound holds for all $r \ge 1$, while it is often only stated for $r \in \N$. For the (not so interesting) boundary case $r=1$, recall that we use the common convention $0^0 := 1$.
\begin{corollary}\label{corprobesbm}
  For all $r \ge 1$ and $0 \le s \le r$,
  \begin{align}
    (1 - \tfrac 1r)^r &\le \tfrac 1e \le (1 - \tfrac 1r)^{r-1},\\
    (1 - \tfrac sr)^{r} &\le e^{-s} \le (1 - \tfrac sr)^{r-s}.
  \end{align}
\end{corollary}

Occasionally, it is useful to know that $(1 - \tfrac 1r)^r$ is monotonically increasing and that $(1 - \tfrac 1r)^{r-1}$ is monotonically decreasing in $r$ (and thus both converge to $\frac 1e$).
\begin{lemma}\label{lprobemonoconv}
  For all $1 \le s \le r$, we have 
  \begin{align}
  (1 - \tfrac 1s)^s &\le (1 - \tfrac 1r)^r,\\
  (1 - \tfrac 1s)^{s-1} &\ge (1 - \tfrac 1r)^{r-1}.
  \end{align}
\end{lemma}

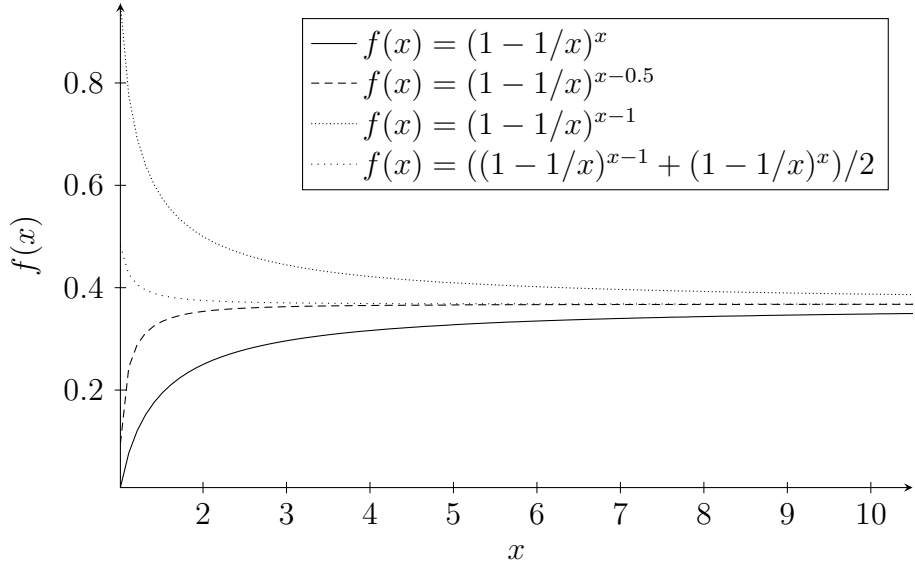
\begin{figure}
\begin{tikzpicture}
\begin{axis}[
    width=12cm,
    height=8cm,
    axis lines = left,
    xlabel = $x$,
    ylabel = {$f(x)$},
    legend style={
      cells={anchor=east},
      legend pos=north east,},
    xmin=1.01,
    xmax=10.5,
    domain=1.01:11.0
]

\addplot [solid,
    samples=100, 
    color=black,
    ]
    {(1-1/x)^x};
\addlegendentry[right]{$f(x) = (1-1/x)^x$}
 
\addplot [densely dashed,
    samples=100, 
    color=black,
    ]
    {(1-1/x)^(x-0.5)};
\addlegendentry[right]{$f(x) = (1-1/x)^{x-0.5}$}
 
%
\addplot [densely dotted,
    samples=100, 
    color=black,
    ]
    {(1-1/x)^(x-1)};
\addlegendentry[right]{$f(x) = (1-1/x)^{x-1}$}
 
\addplot [dotted,
    samples=100, 
    color=black,
    ]
    {((1-1/x)^(x-1)+(1-1/x)^(x))/2};
\addlegendentry[right]{$f(x) = ((1-1/x)^{x-1}+(1-1/x)^{x})/2$}
 
 
\end{axis}
\end{tikzpicture}
\caption{Plots related to Corollary~\ref{corprobesbm}.}\label{figprobesbm}
\end{figure}

Finally, we mention Bernoulli's inequality and a related result. Lemma~\ref{lprobbernoulli}~\ref{lprobbernoulliW} below will be proven at the end of the Section~\ref{secprobbonferroni}, both to show how probabilistic arguments can be used to prove non-probabilistic results and because we have not found a proof for the upper bound in the literature.

\begin{lemma}\label{lprobbernoulli}
\begin{enumerate}
	\item Bernoulli's inequality: Let $x \ge -1$ and $r \in \{0\} \cup [1,\infty)$. Then $(1+x)^r \ge 1+rx$.
	\item\label{lprobbernoulliW} Weierstrass product inequalities: Let $p_1, \dots, p_n \in [0,1]$. Let $P := \sum_{i=1}^n p_i$. Then 
	\begin{equation}
	1 - P \le \prod_{i=1}^n (1-p_i) \le 1 - P + \sum_{i < j} p_i p_j \le 1 - P + \tfrac 12 P^2.\label{eqprobweierminus}
	\end{equation}
	If in addition $P < 1$, then   
  \begin{equation}
	1 + P \le \prod_{i=1}^n (1+p_i) \le \frac{1}{1-P}.\label{eqprobweierplus}
	\end{equation}
	\end{enumerate}
\end{lemma}

The term Weierstrass product inequality is sometimes only attributed to the lower bounds in Lemma~\ref{lprobbernoulli}~\ref{lprobbernoulliW}. For the upper bound in~\eqref{eqprobweierminus}, the estimate 
\begin{equation}
  \prod_{i=1}^n (1-p_i) \le \frac{1}{1 + P}
\end{equation}
is well-known. It is stronger than our bound if and only if $P > 1$. Since for $P \ge 1$ the lower bound is trivial, this might be the less interesting case.

\subsection{Harmonic Number}\label{secprobharmonic}

Quite frequently in the analysis of randomized search heuristics we will encounter the \emph{harmonic number}\index{harmonic number} $H_n$. For all $n \in \N$, it is defined by $H_n = \sum_{k = 1}^n \frac 1k$. Approximating this sum via integrals, namely by \[\int_1^{n+1} \frac 1x dx \le H_n \le 1 + \int_1^n \frac 1x dx,\] we obtain the estimate
\begin{equation}
  \ln n < H_n \le 1 + \ln n\label{eqprobhn}
\end{equation}
valid for all $n \ge 1$. Sharper estimates involving the Euler-Mascheroni constant $\gamma \approx 0.5772156649$ are known, e.g.,
\begin{align*}
  H_n & = \ln n + \gamma \pm O(\tfrac 1n),\\
  H_n & = \ln n + \gamma + \tfrac 1 {2n} \pm O(\tfrac 1 {n^2}).
\end{align*}
For non-asymptotic statements, it is helpful to know that $H_n - \ln n$ is monotonically decreasing (with limit $\gamma$ obviously). In most cases, however, the simple estimate~\eqref{eqprobhn} will be sufficient.

\subsection{Binomial Coefficients and Stirling's Formula}\label{secprobstirling}

Since discrete probability is strongly related to counting, we often encounter the binomial coefficients defined by \[\binom nk := \frac{n!}{k! \, (n-k)!}\] for all $n \in \N_0$, $k \in [0..n]$. The binomial coefficient $\binom nk$ equals the number of $k$-element subsets of a given $n$-element set. For this reason, the above definition is often extended to $\binom nk :=0$ for $k > n$. 

In this section, we give several useful estimates for binomial coefficients. We start by remarking that, while very precise estimates are available, in the analysis of randomized search heuristics often crude estimates are sufficient. 

The following lemma lists some estimates which all can be proven by elementary means. To prove the second inequality of~\eqref{eqprobbinom2}, note that $e^k = \sum_{i = 0}^\infty \frac{k^i}{i!} \ge \frac{k^k}{k!}$ gives the elementary estimate 
\begin{equation}
  \bigg(\frac{k}{e}\bigg)^k \le k! \le k^k \label{eqprobfact}.
\end{equation}
To prove~\eqref{eqprobbinom3}, note that for even $n$ we have $\binom{n}{n/2} = \frac{n!}{(n/2)! (n/2)!} = \prod_{i=1}^{n/2} \frac{2i (2i-1)}{i^2} = 2^n \prod_{i=1}^{n/2} (1 - \frac{1}{2i}) \le 2^n \exp(-\frac 12 \sum_{i=1}^{n/2} \frac 1i) \le 2^n \exp(-\frac 12 \ln \frac n2) = 2^n \sqrt{\frac{2}{n}}$, see Lemma~\ref{lprobelower} and~\eqref{eqprobhn}, while for odd $n$ we have $\binom{n}{\lfloor n/2 \rfloor} = \frac 12 \binom{n+1}{(n+1)/2} \le 2^n \sqrt{\frac{2}{n+1}}$.

\begin{lemma}
  For all $n \in \N$ and $k \in [1..n]$, we have
  \begin{align}
	&\binom nk \le 2^n,\\
	\bigg(\frac nk\bigg)^k \le &\binom nk \le n^k,\\
	&\binom nk \le \frac{n^k}{k!} \le \bigg(\frac{ne}{k}\bigg)^k,\label{eqprobbinom2}\\
	&\binom nk \le \binom{n}{\lfloor n/2 \rfloor} \le 2^n \sqrt{\frac{2}{n}}.\label{eqprobbinom3}
\end{align}
\end{lemma}

Stronger estimates, giving also the well-known version 
\begin{equation}
\binom{n}{k} \le \binom{n}{\lfloor n/2 \rfloor} \le 2^n \sqrt{\frac{2}{\pi n}}\label{eqprobmiddlesharp}
\end{equation}
of~\eqref{eqprobbinom3}, can be obtained from the following estimate known as Stirling's formula. 
	
\begin{theorem}[Robbins~\cite{Robbins55}]\label{tprobstirling}
  For all $n \in \N$, \[n! = \sqrt{2\pi n} (\tfrac ne)^{n} R_n,\] where $1 < \exp(\frac{1}{12n+1}) < R_n < \exp(\frac{1}{12n}) < 1.08690405$.
\end{theorem}

\begin{corollary}\label{corprobstirling}
  For all $n \in \N$ and $k \in [1..n-1]$, \[\binom{n}{k} = \frac{1}{\sqrt{2\pi}} \, \sqrt{\frac{n}{k(n-k)}} \, \bigg(\frac{n}{k}\bigg)^k \bigg(\frac{n}{n-k}\bigg)^{n-k} R_{nk},\] where $0.88102729... = \exp(-\frac 16 + \frac{1}{25}) \le \exp(-\tfrac{1}{12k} - \tfrac{1}{12(n-k)} + \tfrac{1}{12n+1}) < R_{nk} < \exp(-\tfrac{1}{12k+1} - \tfrac{1}{12(n-k)+1} + \tfrac{1}{12n}) < 1$.
\end{corollary}

We refer to~\cite{HwangPRTC18} for an analysis of randomized search heuristics which clearly required Stirling's formula. Stirling's formula was also used in~\cite[proof of Lemma~8]{DoerrW14ranking} to compute another useful fact, namely that all binomial coefficients that are $O(\sqrt n)$ away from the middle one have the same asymptotic order of magnitude of $\Theta(2^n n^{-1/2})$. Here the upper bound is simply~\eqref{eqprobbinom3}.

\begin{corollary}
  Let $\gamma \ge 0$. Let $n \in \N$ and $\ell = \frac n2 \pm \gamma \sqrt n$. Then $\binom n\ell \ge (1-o(1))\frac{2^n}{2\sqrt{\pi n}} e^{-4\gamma^2}$.
\end{corollary}

When working with mutation rates different from classic choice of $\frac 1n$, the following estimates can be useful.
\begin{lemma}\label{lprobmaxbinom}
 Let $n \in \N$, $k \in [0..n]$, and $p \in [0,1]$. Let $X \sim \Bin(n,p)$.
\begin{enumerate}
	\item Let $Y \sim \Bin(n,\frac kn)$. Then $\Pr[X = k] \le \Pr[Y = k]$. This inequality is strict except for the trivial case $p=\frac kn$.
	\item For $k \in [1..n-1]$, $\Pr[X = k] \le \frac{1}{\sqrt{2\pi}} \, \sqrt{\frac{n}{k(n-k)}}$.
\end{enumerate}
\end{lemma}

\begin{proof}
  The first part follows from $\Pr[X = k] = \binom{n}{k} p^k (1-p)^{n-k}$ and noting that $p \mapsto p^k (1-p)^{n-k}$ has a unique maximum in the interval $[0,1]$, namely at $p=\frac kn$. The second part follows from the first and using Corollary~\ref{corprobstirling} to estimate the binomial coefficient in the expression $\Pr[Y=k] = \binom{n}{k}(\frac 1k)^k (1-\frac nk)^{n-k}$.
\end{proof}

For the special case that $np=k$, the second part of the lemma above was already shown in~\cite[Lemma~10 of the arxiv version]{SudholtW16}. For $k \in \{\lfloor np \rfloor, \lceil np \rceil\}$ but $np \neq k$, a bound larger than ours by a factor of $e$ was shown there as well.

Finally, we note that to estimate sums of binomial coefficients, large deviations bounds (to be discussed in Section~\ref{secproblargedev}) can be an elegant tool. Imagine we need an upper bound for $S = \sum_{k=a}^n \binom{n}{k}$, where $a > \frac n2$. Let $X$ be a random variable with distribution $\Bin(n,\frac 12)$. Then $\Pr[X \ge a] = 2^{-n} S$. Using the additive Chernoff bound of Theorem~\ref{tprobchernoffadditive01}, we also see $\Pr[X \ge a] = \Pr[X \ge E[X] + (a-\frac n2)] \le \exp(-\frac{2(a-\frac n2)^2}{n})$. Consequently, $S \le 2^n \exp(-\frac{2(a-\frac n2)^2}{n})$. 

The same argument can even be used to estimate single binomial coefficients, in particular, those not to close to the middle one. Note that by Lemma~\ref{lprobklar}, $S = \sum_{k=a}^n \binom{n}{k}$ and $\binom{n}{a}$ are quite close when $a$ is not too close to $\frac n2$. Hence \begin{equation}
\binom na \le 2^n \exp\left(-\frac{2(a-\frac n2)^2}{n}\right)
\end{equation}
is a good estimate in this case.

\section{Union Bound}

The \emph{union bound}\index{Union Bound}, sometimes called Boole's inequality, is a very elementary consequence of the axioms of a probability space, in particular, the $\sigma$-additivity of the probability measure. 

\begin{lemma}[Union bound]\label{lprobunionbound}\index{Union Bound}
  Let $E_1, \ldots, E_n$ be arbitrary events in some probability space. Then \[\Pr\bigg[\bigcup_{i = 1}^n E_i\bigg] \le \sum_{i = 1}^n \Pr[E_i].\]
\end{lemma}

Despite its simplicity, the union bound is a surprisingly powerful tool in the analysis of randomized algorithms. It draws its strength from the fact that does not need any additional assumptions. In particular, the events $E_i$ are not required to be independent. Here is an example of such an application of the union bound.

\subsection{Example: The \texorpdfstring{\oea}{(1+1) EA} Solving the Needle Problem} 

The \emph{needle function}\index{needle function} is the fitness function $f : \{0,1\}^n \to \Z$ defined by $f(x) = 0$ for all $x \in \{0,1\}^n \setminus \{(1,\dots,1)\}$ and $f((1,\dots,1)) = 1$. It is neither surprising nor difficult to prove that all reasonable randomized search heuristics need time exponential in $n$ to find the maximum of the needle function. To give a simple example for the use of the simplified drift theorem, it was shown in~\cite{OlivetoW11} that the classic \oea within a sufficiently small exponential time does not even get close to the optimum of the needle function (see Theorem~\ref{tprobneedle} below). We now show that the same result (and in fact a stronger one) can be shown via the union bound. 

The \oea is the simple randomized search heuristic that starts with a random search point $x \in \{0,1\}^n$. Then, in each iteration, it generates from $x$ a new search point $y$ by copying $x$ into $y$ and flipping each bit independently with probability $\frac 1n$. If the new search point (``offspring'') $y$ is at least as good as the parent $x$, that is, if $f(y) \ge f(x)$ for an objective function to be maximized, then $x$ is replaced by $y$, that is, we set $x := y$. Otherwise, $y$ is discarded. 

\begin{algorithm2e}%
	Choose $x \in \{0,1\}^n$ uniformly at random\;
  \For{$t=1,2,3,\ldots$}{
    $y \assign x$\;
    \For{$i \in [1..n]$}{with probability~$\frac 1n$ do $y_i \assign 1 - y_i$\;}
    \lIf{$f(y)\geq f(x)$}{$x \assign y$}
  }
\caption{The \oea for maximizing~$f\colon\{0,1\}^n\to\mathbb{R}$.}
\label{alg:oea}
\end{algorithm2e}

The precise result of~\cite[Theorem~5]{OlivetoW11} is the following.

\begin{theorem}\label{tprobneedle}
  For all $\eta>0$ there are $c_1, c_2 >0$ such that with probability $1 - 2^{c_1 n}$ the first $2^{c_2 n}$ search points generated in a run of the \oea on the needle function all have a Hamming distance of more than $(\frac 12 - \eta) n$ from the optimum. 
\end{theorem}

The proof of this theorem in~\cite{OlivetoW11} argues as follows. Denote by $x^{(0)}, x^{(1)}, \dots$ the search points generated in a run of the \oea. Denote by $x^*$ the optimum of the needle function. For all $i \ge 0$, let $X_i := H(x^{(i)},x^*) := |\{j \in [1..n] \mid x^{(i)}_j \neq x^*_j\}|$ be the Hamming distance of $x^{(i)}$ from the optimum. The random initial search point $x^{(0)}$ has an expected Hamming distance of $\frac n2$ from the optimum. By a simple Chernoff bound argument (Theorem~\ref{tprobchernoffadditive01}), we see that with probability $1 - \exp(-2 \eta^2 n)$, we have $X_0 = H(x^{(0)},x^*) > (\frac 12 - \eta) n$. Now a careful analysis of the random process $(X_i)_{i \ge 0}$ via a new ``simplified drift theorem'' gives the claim. 

We now show that the Chernoff bound argument plus a simple union bound are sufficient to prove the theorem. We show the following more explicit bound, which also applies to all other unbiased algorithms in the sense of Lehre and Witt~\cite{LehreW12} (roughly speaking, all algorithms which treat the bit-positions $[1..n]$ and the bit-values $\{0,1\}$ in a symmetric fashion).

\begin{theorem}\label{tprobneedle2}
  For all $\eta>0$ and $c >0$ we have that with probability at least $1 - 2^{(c - 2\ln(2)\eta^2)n}$ the first $L := 2^{c n}$ search points generated in a run of the \oea (or any other unbiased black-box optimization algorithm) on the needle function all have a Hamming distance of more than $(\frac 12 - \eta) n$ from the optimum. 
\end{theorem}

\begin{proof}
  The key observation is that as long as the \oea has not found the optimum, any search point $x$ generated by the \oea is uniformly distributed in $\{0,1\}^n$. Hence $\Pr[H(x,x^*) \le (\frac 12 - \eta)n] \le \exp(-2\eta^2 n)$ by Theorem~\ref{tprobchernoffadditive01}. By the union bound, the probability that one of the first $L := 2^{c n}$ search points generated by the \oea has a distance $H(x,x^*)$ of at most $(\frac 12 - \eta)n$, is at most $L \exp(-2 \eta^2 n) = 2^{(c - 2\ln(2)\eta^2)n}$.
  
  To be more formal, let $x^{(0)}, x^{(1)}, \dots$ be the search points generated in a run of the \oea. Let $T = \min\{t \in \N_0 \mid x^{(t)} = x^*\}$. Define a sequence $y^{(0)}, y^{(1)}, \dots$ of search points by setting $y^{(t)} := x^{(t)}$ for all $t \le T$. For all $t > T$, let $y^{(t)}$ be obtained from $y^{(t-1)}$ by flipping each bit independently with probability $\frac 1n$. With this definition, and since $x^{(t)} = x^*$ for all $t \ge T$, we have 
\begin{align*}
  \{x^{(t)} \mid t \in [0..L-1]\} &= \{x^{(t)} \mid t \in [0..\min\{T,L-1\}]\} \\
  &= \{y^{(t)} \mid t \in [0..\min\{T,L-1\}]\} \subseteq \{y^{(t)} \mid t \in [0..L-1]\}.
\end{align*}
Consequently, \[\Pr[\exists t \in [0..L-1] : H(x^{(t)},x^*) \le (\tfrac 12 - \eta)n] \le \Pr[\exists t \in [0..L-1] : H(y^{(t)},x^*) \le (\tfrac 12 - \eta)n].\] By the union bound, \[\Pr[\exists t \in [0..L-1] : H(y^{(t)},x^*) \le (\tfrac 12 - \eta)n] \le \sum_{t = 0}^{L-1} \Pr[H(y^{(t)},x^*) \le (\tfrac 12 - \eta)n].\] Note that when $y^{(t)}$ is a search point uniformly distributed in $\{0,1\}^n$, then so is $y^{(t+1)}$. Since $y^{(0)}$ is uniformly distributed, all $y^{(t)}$ are. Hence by Theorem~\ref{tprobchernoffadditive01} we have $\Pr[H(y^{(t)},x^*) \le (\frac 12 - \eta)n] \le \exp(-2 \eta^2 n)$ for all $t$ and thus \[\sum_{t = 0}^{L-1} \Pr[H(y^{(t)},x^*) \le (\tfrac 12 - \eta)n] \le L \exp(-2\eta^2 n) = 2^{(c - 2\ln(2)\eta^2)n}.\]  
  
  This proof immediately extends to all algorithms which, when optimizing the needle function, generate uniformly distributed search points until the optimum is found. These are, in particular, all unbiased algorithms in the sense of Lehre and Witt~\cite{LehreW12}.
\end{proof}

Note that the $y_t$ in the proof above are heavily correlated. For all $t$, the search points $y_t$ and $y_{t+1}$ have an expected Hamming distance of exactly one. Nevertheless, we could apply the union bound to the events ``$H(y_t,x^*) < (\frac 12 - \eta)n$'' and from this obtain a very elementary proof of Theorem~\ref{tprobneedle2}.

\subsection{Lower Bounds, Bonferroni Inequalities}\label{secprobbonferroni}

The union bound is tight, that is, holds with equality, when the events $E_i$ are disjoint. In this case, the union bound simply reverts to the $\sigma$-additivity of the probability measure. The \emph{second Bonferroni inequality} gives a lower bound for the probability of a union of events also when they are not disjoint. 

\begin{lemma}\label{lem:bonferroni2}
  Let $E_1, \ldots, E_n$ be arbitrary events in some probability space. Then \[\Pr\bigg[\bigcup_{i = 1}^n E_i\bigg] \ge \sum_{i = 1}^n \Pr[E_i] - \sum_{i=1}^{n-1}\sum_{{j=i+1}}^n \Pr[E_i \cap E_j].\]
\end{lemma}

As an illustration, let us regard the performance of \emph{blind random search}\index{blind random search} on the needle function, that is, we let $x^{(1)}, x^{(2)}, \dots$ be independent random search points from $\{0,1\}^n$ and ask ourselves what is the first hitting time $T = \min\{t \in \N \mid x^{(t)} = (1,\dots,1)\}$ of the maximum $x^* = (1,\dots,1)$ of the needle function (any other function $f : \{0,1\}^n \to \R$ with unique global optimum would do as well). This is easy to compute directly. We see that $T$ has geometric distribution with success probability $2^{-n}$, so the probability that $L$ iterations do not suffice to find the optimum is $\Pr[T > L] = (1-2^{-n})^L$. 

Let us nevertheless see what we can derive from union bound and second Bonferroni inequality. Let $E_t$ be the event $x^{(t)} = x^*$. Then the union bound gives \[\Pr[T \le L] = \Pr\bigg[\bigcup_{t=1}^L E_t\bigg] \le L 2^{-n},\] the second Bonferroni inequality yields \[\Pr[T \le L] = \Pr\bigg[\bigcup_{t=1}^L E_t\bigg] \ge L 2^{-n} - \frac{L(L-1)}{2} 2^{-2n}.\] Hence if $L = o(2^{n})$, that is, $L$ is of smaller asymptotic order than $2^n$, then $\Pr[T \le L] = (1-o(1)) L 2^{-n}$, that is, the union bound estimate is asymptotically tight. 

For reasons of completeness, we state the full set of Bonferroni inequalities. Note that the case $k=1$ is the union bound and the case $k=2$ is the lemma above.
 
\begin{lemma}\label{lem:bonferronik}
  Let $E_1, \ldots, E_n$ be arbitrary events in some probability space. For all $k \in [1..n]$, let \[S_k :=\sum _{1 \le i_1 < \cdots < i_k \leq n} \Pr[A_{i_{1}}\cap \cdots \cap A_{i_{k}}].\]
  Then for all $k \in [1..n]$ we have
  \begin{itemize}
	  \item $\displaystyle{\Pr\bigg[\bigcup_{i = 1}^n E_i\bigg] \le \sum_{j=1}^k (-1)^{j-1} S_j}$ for $k \in [1..n]$ odd,
	  \item $\displaystyle{\Pr\bigg[\bigcup_{i = 1}^n E_i\bigg] \ge \sum_{j=1}^k (-1)^{j-1} S_j}$ for $k \in [1..n]$ even.
  \end{itemize}
\end{lemma}  
  
In simple terms, the Bonferroni inequalities state that when we omit the terms for $j > k$ in the inclusion-exclusion formula \[\Pr\bigg[\bigcup_{i = 1}^n E_i\bigg] = \sum_{j=1}^n (-1)^{j-1} S_j,\] then first of the omitted terms (that is, the one for $j=k+1$) dominates the error. So if $k$ is odd and thus the first omitted term is negative, then we obtain a $\le$ inequality, and inversely for $k$ even.

We now use the Bonferroni inequalities to prove two of the inequalities given in Lemma~\ref{lprobbernoulli}~\ref{lprobbernoulliW}.

\begin{proof}[Proof of~\eqref{eqprobweierminus}]
  Consider some probability space with independent events $E_1, \dots, E_n$ having $\Pr[E_i] = p_i$. Due to the independence, 
\begin{equation}
  \prod_{i=1}^n (1-p_i) = \Pr[\forall i \in [1..n] : \neg E_i] = 1 - \Pr[\exists i \in [1..n] : E_i]. \label{eqprobxyz}
\end{equation}
  By the union bound, the right-hand side of~\eqref{eqprobxyz} is at least $1 - \sum_{i=1}^n p_i = 1 - P$. By the Bonferroni inequality for $k=2$ and again the independence, the right-hand side of~\eqref{eqprobxyz} is at most \[1 - P + \sum_{i < j} \Pr[E_i \cap E_j] = 1 - P + \sum_{i < j} p_i p_j \le 1 - P + \tfrac 12 \sum_{i=1}^n \sum_{j=1}^n p_i p_j = 1 - P + \tfrac 12 P^2.\]
\end{proof}
Note that the slack in the last inequality is only the term $ \tfrac 12 \sum_{i=1}^n p_i^2$, so there is not much reason to prefer the stronger upper bound $1 - P + \sum_{i < j} p_i p_j$ over the bound $1 - P + \tfrac 12 P^2$.

\section{Expectation and Variance}\label{secprobexp}

Expectation and variance are two key characteristic numbers of a random variable.

\subsection{Expectation}

The \emph{expectation}\index{expectation} (or mean) of a random variable $X$ taking values in some set $\Omega \subseteq \R$ is defined by $E[X] = \sum_{\omega \in \Omega} \omega \Pr[X = \omega]$, where we shall always assume that the sum exists and is finite. As a trivial example, we immediately see that if $X$ is a binary random variable, then $E[X] = \Pr[X=1]$.

For non-negative integral random variables, the expectation can also be computed by the following formula (which is valid also when $E[X]$ is not finite). 

\begin{lemma}\label{lprobnonnegexp}
  Let $X$ be a random variable taking values in the non-negative integers. Then \[E[X] = \sum_{i=1}^\infty \Pr[X \ge i].\] Is $X$ takes values in $(-\infty,0] \cup \N$, then $E[X] \le \sum_{i=1}^\infty \Pr[X \ge i]$ still holds.
\end{lemma}

This lemma, among others, allows to conveniently transform information about the tail bound of a distribution into a bound on its expectation. This was done, e.g., in~\cite[proof of Lemma~10]{DrosteJW02} for lower bounds, in~\cite[proof of Theorem~2]{DoerrSW13foga} in a classic runtime analysis, and in~\cite[proof of Theorem~5]{DoerrG13algo} in the simplified proof of the multiplicative drift theorem. 

Lemma~\ref{lprobnonnegexp} also allows to conveniently derive from information on the upper tail of a random variable an estimate for its expectation, as done in the following elementary result.
  
\begin{corollary}[Expectations from exponential tail bounds]\label{corprobtaile}
  Let $\alpha, \beta > 0$ and $T \ge 0$. Let $X$ be an integer random variable and $Y$ be a non-negative integer random variable.
  \begin{enumerate}
	\item If $\Pr[X \ge T + \lambda] \le \alpha \exp(-\tfrac{\lambda}{\beta})$ for all $\lambda \in \N$, then $E[X] \le T+\alpha\beta$.
	\item If $\Pr[Y \le T - \lambda] \le \alpha \exp(-\tfrac{\lambda}{\beta})$ for all $\lambda \in [1..T]$, then $E[Y] \ge T-\alpha\beta$.
	\item If $\Pr[X \ge (1+\eps) T] \le \alpha \exp(-\tfrac{\eps}{\beta})$ for all $\eps > 0$, then $E[X] \le (1+\alpha\beta) T$.
	\item If $\Pr[X \le (1-\eps) T] \le \alpha \exp(-\tfrac{\eps}{\beta})$ for all $\eps \in (0,1]$, then $E[X] \ge (1-\alpha\beta) T$.
  \end{enumerate}
\end{corollary}

\begin{proof}
  By Lemma~\ref{lprobnonnegexp}, we compute 
  \begin{align*}
  E[X] &\le \sum_{i=1}^\infty \Pr[X \ge i] \le T + \sum_{i=T+1}^\infty \alpha \exp\left(-\frac{i-T}{\beta}\right) \\
  & = T - \alpha + \alpha \frac{1}{1-\exp(-1/\beta)} \le T + \alpha\beta, 
  \end{align*}
  where the last estimate uses~\eqref{eqprobeupper1}. 
  
  Similarly, we compute 
  \begin{align*}
  E[Y] &= \sum_{i=1}^T \Pr[Y \ge i] \ge \sum_{i=1}^T (1 - \Pr[Y \le i-1]) \\
  & \ge \sum_{\lambda=1}^T (1 - \alpha\exp(-\tfrac{\lambda}{\beta})) \ge T - \sum_{\lambda=1}^\infty \alpha \exp(-\tfrac{\lambda}{\beta}) \ge T - \alpha\beta.
  \end{align*}
  The last two claims are simple reformulations of the first two.
\end{proof}

In a similar vein, Lemma~\ref{lprobnonnegexp} yields an elegant analysis of the expectation of a geometric random variable\index{random variable!geometric}. Let $X$ be a geometric random variable with success probability~$p$. Intuitively, we feel that the expected waiting time for a success is~$\frac 1p$. This intuition is guided by the fact that after~$\frac 1p$ repetitions of the underlying binary random experiment, the expected number of successes is exactly one. This intuition led to the right result, the ``proof'' however is not correct. The correct proof either uses standard results in Markov chain theory, or elementary but  non-trivial calculations, or (as done below) the same reasoning as in the lemma above. 

\begin{lemma}[Waiting time argument]\label{lprobwaitingtime}\index{random variable!geometric}
  Let $X$ be a geometric random variable with success probability $p>0$. Then $E[X] = \frac 1p$.
\end{lemma}

\begin{proof}
  We have $\Pr[X \ge i] = (1-p)^{i-1}$, since $X \ge i$ is the event of having no success in the first $i-1$ rounds of the random experiment. Now Lemma~\ref{lprobnonnegexp} gives \[E[X] = \sum_{i=1}^\infty \Pr[X \ge i] = \sum_{i=1}^{\infty} (1-p)^{i-1} = \frac{1}{1-(1-p)} = \frac 1p.\] 
\end{proof}

  
An elementary, but very useful property is that expectation is linear.

\begin{lemma}[Linearity of expectation]\label{lproblinearity}
  Let $X_1, \ldots, X_n$ be arbitrary random variables and $a_1, \ldots, a_n \in \R$. Then \[E\bigg[\sum_{i = 1}^n a_i X_i\bigg] = \sum_{i = 1}^n a_i E[X_i].\]
\end{lemma}

This fact is very convenient when we can write a complicated random variable as sum of simpler ones. For example, let $X$ be a binomial random variable with parameters $n$ and $p$, that is, we have $\Pr[X = k] = \binom{n}{k} p^k (1-p)^{n-k}$. Since $X$ counts the number of successes in $n$ (independent) trials, we can write $X = \sum_{i=1}^n X_i$ as the sum of (independent) binary random variables $X_1, \ldots, X_n$, each with $\Pr[X_i = 1] = p$. Here $X_i$ is the indicator random variable for the event that the $i$-th trial is a success. Using linearity of expectation, we compute \[E[X] = E\bigg[\sum_{i=1}^n X_i\bigg] = \sum_{i=1}^n E[X_i] = n p.\] Note that we did not need that the $X_i$ are independent. 
We just proved the following.

\begin{lemma}[Expectation of binomial random variables]\label{lprobexpectationbinomial}\index{random variable!binomial}
  Let $X$ be a binomial random variable with parameters $n$ and $p$. Then $E[X] = pn$.
\end{lemma}

In the same fashion, we can compute the following elementary facts. 

\begin{lemma}\label{lprobhamming}
Let $x, y, x^* \in \{0,1\}^n$. Denote by $H(x,y) := |\{i \in [1..n] \mid x_i \neq y_i\}|$ the Hamming distance of $x$ and $y$.
\begin{enumerate}
	\item Let $z$ be obtained from $x$ via \emph{standard-bit mutation}\index{standard-bit mutation} with rate $p \in [0,1]$, that is, by flipping each bit of $x$ independently with probability $p$. Then $E[H(x,z)] = pn$ and $E[H(z,x^*)] = H(x,x^*) + p (n - 2H(x,x^*))$.
	\item Let $z$ be obtained from $x$ and $y$ via uniform crossover, that is, for each $i \in [1..n]$ independently, we set $z_i = x_i$ or $z_i = y_i$ each with probability $\frac 12$. Then $E[H(x,z)] = \frac 12 H(x,y)$ and $E[H(z,x^*)] = \frac 12 (H(x,x^*) + H(y,x^*))$.
	\item Let $z$ be obtained from the unordered pair $\{x,y\}$ via $1$-point crossover, that is, we choose $r$ uniformly at random from $[0..n]$ and then with probability $\frac 12$ each 
	\begin{itemize}
	\item define $z$ by $z_i = x_i$ for $i \le r$ and $z_i = y_i$ for $i > r$, or
	\item define $z$ by $z_i = y_i$ for $i \le r$ and $z_i = x_i$ for $i > r$.
	\end{itemize}
	Then $E[H(x,z)] = \frac 12 H(x,y)$ and $E[H(z,x^*)] = \frac 12 (H(x,x^*) + H(y,x^*))$.
\end{enumerate}
\end{lemma}
The fact that the results for the two crossover operators  are identical shows again that linearity of expectation does not care about possible dependencies. We have $\Pr[z_i = x_i] = \frac 12$ in both cases, and this is what is important for the result, whereas the fact that the events ``$z_i = x_i$'' are independent for uniform crossover and strongly dependent for $1$-point crossover has no influence on the result.

\subsection{Markov's Inequality}

\emph{Markov's inequality}\index{Markov's inequality} is an elementary large deviation bound valid for \emph{all} non-negative random variables.

\begin{lemma}[Markov's inequality]\label{lprobmarkov}
  Let $X$ be a non-negative random variable with $E[X] > 0$. Then for all $\lambda >0$, 
  \begin{align}
  &\Pr[X \ge \lambda E[X]] \le \tfrac 1 \lambda,\\
  &\Pr[X \ge \lambda] \le \tfrac {E[X]} \lambda.\label{eqprobmarkov2}
  \end{align}  
\end{lemma}

\begin{proof}
We have \[\displaystyle{E[X] = \sum\limits_\omega \omega \Pr[X = \omega] \ge \sum\limits_{\omega \ge \lambda} \lambda \Pr[X=\omega] = \lambda \Pr[X \ge \lambda]},\] proving the second formulation of Markov's inequality.
\end{proof}

We note that~\eqref{eqprobmarkov2} also holds without the assumption $E[X] > 0$. More interestingly, the proof above shows that Markov's inequality is always strict (that is, holds with ``$<$'' instead of ``$\le$'') when $X$ takes at least three different values with positive probability.

It is important to note that Markov's inequality, without further assumptions, only gives information about deviations above the expectation. If $X$ is a (not necessarily non-negative) random variable taking only values not larger than some $u \in \R$, then the random variable $u-X$ is non-negative and Markov's inequality gives the bound 
\begin{equation}
\Pr[X \le \lambda] \le \frac{u - E[X]}{u - \lambda},
\end{equation}
which is sometimes called \iemph{reverse Markov's inequality}. An equivalent formulation of this bound is \
\begin{equation}
\Pr[X > \lambda] \ge \frac{E[X] - \lambda}{u - \lambda}.
\end{equation}

Markov's inequality is useful if not much information is available about the random variable under consideration. Also, when the expectation of $X$ is very small, then the following elementary corollary is convenient and, in fact, often quite tight.

\begin{corollary}[First moment method]\label{corprobmarkov}
  If $X$ is a non-negative random variable, then $\Pr[X \ge 1] \le E[X]$.
\end{corollary}

Corollary~\ref{corprobmarkov} together with linearity of expectation often gives the same results as the union bound. For an example, recall that in Section~\ref{secprobbonferroni} we observed that in a run of the blind random search heuristic, the probability that the $t$-th search point $x_t$ is the unique optimum of a given function $f : \{0,1\}^n \to \R$, is $2^{-n}$. Denote this event by $E_t$ and let $X_t$ be the indicator random variable for this event. Then the probability that one of the first $L$ search points is the optimum, can be estimated equally well via the union bound or the above corollary and linearity of expectation:
\begin{align*}
  &\Pr\bigg[\bigcup_{t=1}^L E_t\bigg] \le \sum_{t=1}^L \Pr[E_t] = L 2^{-n},\\
  &\Pr\bigg[\sum_{t=1}^L X_t \ge 1\bigg] \le E\bigg[\sum_{t=1}^L X_t\bigg] = \sum_{t=1}^L E[X_t] = L 2^{-n}.
\end{align*}

\subsection{Chebyshev's Inequality}\index{Chebyshev's inequality}

The second elementary large deviation bound is Chebyshev's inequality, sometimes called Bienaym\'e-Chebyshev inequality as it was first stated in Bienaym\'e~\cite{Bienayme53} and later proven in Chebyshev~\cite{Tchebichef67}. It seems less often used in the theory of randomized search heuristics (exceptions being~\cite{NeumannSW10,DoerrJWZ13}). 

Recall that the \emph{variance} of a discrete random variable $X$  is 
\begin{equation}
  \Var[X] = E[(X - E[X])^2] = E[X^2] - E[X]^2.\label{eqprobdefvar}
\end{equation}
Just by definition, the variance already is a measure of how well $X$ is concentrated around its mean. 

From the variance, we also obtain a bound on the expected (absolute) deviation from the mean. Applying the well-known estimate $E[X]^2 \le E[X^2]$, which follows from the second equality in~\eqref{eqprobdefvar}, to the random variable $|X - E[X]|$, we obtain
\begin{equation}
  E[|X - E[X]|] \le \sqrt{E[(X - E[X])^2]} = \sqrt{\Var[X]}.
\end{equation}

More often, we use the variance to bound the probability of deviating from the expectation by a certain amount. Applying Markov's inequality to the random variable $(X - E[X])^2$ easily yields the following very useful inequality.

\begin{lemma}[Chebyshev's inequality]\label{lprobchebyshev}
  Let $X$ be a random variable with $\Var[X] > 0$. Then for all $\lambda > 0$, 
  \begin{align}
  &\Pr\big[|X - E[X]| \ge \lambda \sqrt{\Var[X]}\,\big] \le \tfrac 1 {\lambda^2},\\
  &\Pr\big[|X - E[X]| \ge \lambda\big] \le \tfrac {\Var[X]} {\lambda^2}.\label{eqprobcheby2}
  \end{align}
\end{lemma}

Similar to Markov's inequality, the second estimate is valid also without the assumption $\Var[X] > 0$. Note that Chebyshev's inequality automatically yields a two-sided tail bound (that is, for both cases that the random variable is larger and smaller than its expectation), as opposed to Markov's inequality (giving just a bound for exceeding the expectation). There is a one-sided version of Chebyshev's inequality that is often attributed to Cantelli, though Hoeffding~\cite{Hoeffding63} sees Chebyshev~\cite{Tchebichef74} as its inventor.

\begin{lemma}[Cantelli's inequality]\label{lprobcantelli}
  Let $X$ be a random variable with $\Var[X] > 0$. Then for all $\lambda > 0$, 
  \begin{align}
  &\Pr\big[X \ge E[X] + \lambda \sqrt{\Var[X]}\,\big] \le \tfrac 1 {\lambda^2+1},\\
  &\Pr\big[X \le E[X] - \lambda \sqrt{\Var[X]}\,\big] \le \tfrac 1 {\lambda^2+1}.
  \end{align}
\end{lemma}

In many applications, the slightly better bound of Cantelli's inequality is not very interesting. Cantelli's inequality has, however, the charm that the right-hand side is always less than one, hence one can obtain non-trivial probabilities also for deviations smaller than $\sqrt{\Var[X]}$. We shall exploit this in the proof of Lemma~\ref{lprobsqrtn}. 

While Markov's inequality can be used to show that a non-negative random variable $X$ rarely is positive (first moment method), Chebyshev's inequality can serve the opposite purpose, namely showing that $X$ is positive with good probability. By taking $\lambda = E[X]$ in~\eqref{eqprobcheby2}, we obtain the first estimate of the following lemma. Using the Cauchy-Schwarz inequality and computing 
\[E[X]^2 = E[X {\bf 1}_{X \neq 0}]^2 \le E[X^2] E[{\bf 1}_{X \neq 0}] = E[X^2] \Pr[X \neq 0],\] 
we obtain the second estimate, which has the nice equivalent formulation 
\begin{equation}
\Pr[X \neq 0] \ge \frac{E[X]^2}{E[X^2]}.
\end{equation} 
Since $E[X^2] \ge E[X]^2$, the second estimate gives a stronger bound for $\Pr[X=0]$ than the first. While the lemma below does not require that $X$ is non-negative, the typical application of showing that $X$ is positive requires that $X$ is non-negative in the second bound, so that $\Pr[X \neq 0] = \Pr[X > 0]$.
\begin{lemma}[Second moment method]
  For a  random variable $X$ with $E[X] \neq 0$, 
  \begin{align}
  &\Pr[X = 0] \le \Pr[X \le 0] \le \frac{\Var[X]}{E[X]^2},\\
  &\Pr[X = 0] \le \frac{\Var[X]}{E[X^2]}.
  \end{align}
\end{lemma}

In the (purely academic) example of finding a unique global optimum via blind random search\index{blind random search} (see Section~\ref{secprobbonferroni}), let $X_t$ be indicator random variable for the event that the $t$-th search point is the optimum. Let $X = \sum_{t=1}^L X_t$. Then the probability that the optimum is found within the first $L$ iterations is \[\Pr[X > 0] = 1 - \Pr[X=0] \ge 1 - \frac{\Var[X]}{E[X]^2}.\]
The variance for a sum of binary random variables is \[\Var[X] = \sum_{t=1}^L \Var[X_t] + \sum_{s < t} \Cov[X_s,X_t] \le E[X] + \sum_{s < t} \Cov[X_s,X_t],\] where we recall the definition of the \emph{covariance}\index{covariance} 
\[\Cov[U,V] := E[UV] - E[U]E[V]\] 
of two arbitrary random variables $U$ and $V$\/. Here we have $\Cov[X_s,X_t] = 0$, since the $X_t$ are independent. Consequently, \[\Pr[X > 0] \ge 1 - \frac{1}{E[X]}.\] Hence the probability to find the optimum within $L$ iterations is $\Pr[T \le L] = \Pr[X > 0] \ge 1 - \frac 1{L 2^{-n}}$. Note that this estimate is, for interesting case that $E[X]$ is large, much better than the bound $\Pr[T \le L] \ge L 2^{-n} - \frac{L(L-1)}{2} 2^{-2n}$ which we obtained from the second Bonferroni inequality. 

\section{Conditioning}

In the analysis of randomized heuristics, we often want to argue that a certain desired event $C$ already holds and the continue arguing under this condition. Formally, this gives rise to a new probability space where each of the original events $A$ now has a probability of \[\Pr[A \mid C] := \frac{\Pr[A \cap C]}{\Pr[C]}\,.\] Obviously, this only makes sense for events $C$ with $\Pr[C] > 0$. In an analogous fashion, we define the expectation of a random variable $X$ conditional on $C$ by $E[X \mid C] = \sum_{\omega \in C} X(\omega) \Pr[\omega \mid C]$. The random variable behind this definition, which takes a value $x$ with probability $\Pr[X=x]/\Pr[C]$, is sometimes denoted by $(X \mid C)$. 

While we shall not use this notation, we still feel the need to warn the reader that there is the related notion of the conditional expectation with respect to a random variable, which sometimes creates confusion. If $X$ and $Y$ are two random variables defined on the same probability space, then $E[X \mid Y]$ is a function (that is, a random variable) defined on the range of $Y$ by $E[X \mid Y](y) = E[X \mid Y=y]$.

Conditioning as a proof technique has many faces, among them the following.

\subsection{Decomposing Events} 

If we can write some event $A$ as the intersection of two events $A_1$ and $A_2$, then it can be useful to first compute the probability of $A_1$ and then the probability of $A_2$ conditional on $A_1$. Right from the definition, we have $\Pr[A_1 \cap A_2] = \Pr[A_1] \Pr[A_2 \mid A_1]$. Of course, this requires that we have some direct way of computing $\Pr[A_1 \mid A_2]$.

\subsection{Case Distinctions} 

Let $C_1, \dots, C_k$ be a partition of our probability space. If it is easy to analyze our problem conditional on each of these events (``in the that case $C_i$ holds''), then the following \emph{law of total probability} and \emph{law of total expectation} are useful.

\begin{lemma}[Laws of total probability and total expectation]
  Let $C_1, \dots, C_k$ be a partition of our probability space. Let $A$ be some event and $X$ be some random variable. Then
  \begin{align*}
  \Pr[A] &= \sum_{i=1}^k \Pr[A \mid C_i] \, \Pr[C_i],\\*
  E[X] &= \sum_{i=1}^k E[X \mid C_i] \, \Pr[C_i].
  \end{align*}
\end{lemma}

\subsection{Excluding Rare Events} 

Quite often in the analysis of nature-inspired search heuristics, we would like to exclude some rare unwanted event. For example, assume that we analyze an evolutionary algorithm using standard-bit mutation with mutation rate $\frac 1n$. Then it is very unlikely that in an application of this mutation operator more than $n^{1/4}$ bits are flipped. So it could be convenient to exclude this rare event, say by stating that ``with probability $1 - 2^{-\Omega(n^{1/4})}$, in none of the first $n^2$ applications of the mutation operator more than $n^{1/4}$ bits are flipped; let us in the following condition on this event''. See the proofs of Theorem~7 and~8 in~\cite{DrosteJW02} for examples where such an argument is useful. 

What could be a problem with this approach is that as soon as we condition on such an event, we change the probability space and thus arguments valid in the unconditional setting are not valid anymore. As a simple example, note that once we condition on that we flip at most $n^{1/4}$ bits, the events $E_i$ that the $i$-th bit flips are not independent anymore. Fortunately, we can safely ignore this in most cases (and many authors do so without saying a word on this affair). The reason is that when conditioning on an almost sure event, then the probabilities of all events change only very little (see the lemma below for this statement made precise). Hence in our example, we can compute the probability of some event assuming that the bit flips are independent and then correct this probability by a minor amount. 

\begin{lemma}
  Let $C$ be some event with probability $1-p$. Let $A$ be any event. Then \[\frac{\Pr[A]-p}{1-p} \le \Pr[A \mid C] \le \frac{\Pr[A]}{1-p}.\] In particular, for $p \le \tfrac 12$, we have $\Pr[A] - p \le \Pr[A \mid C] \le \Pr[A] + 2p$.
\end{lemma}  

The proof of this lemma follows right from the definition of conditional probabilities and the elementary estimate $\Pr[A] - p \le \Pr[A \setminus \overline C] = \Pr[A \cap C] \le \Pr[A]$, where $\overline C$ denotes the complement of $C$. From this, we also observe the natural fact that when $A \subseteq C$, that is, the event $A$ implies $C$, then conditioning on $C$ rather increases the probability of $A$: 
\begin{equation}
\Pr[A \mid C] = \frac{\Pr[A \cap C]}{\Pr[C]} = \frac{\Pr[A]}{\Pr[C]} \ge \Pr[A].
\end{equation} 
Likewise, when $A \supseteq \overline C$, then 
\begin{equation}
\Pr[A \mid C] = \frac{\Pr[A \setminus \overline C]}{\Pr[C]} = \frac{\Pr[A]-p}{1-p} \le \Pr[A].
\end{equation}
For example, if $X$ is the number of bits flipped in an application of standard-bit mutation, then $\Pr[X \le 10 \mid X \le \frac n2] \ge \Pr[X \le 10]$ and $\Pr[X \ge 10 \mid X \le \frac n2] \le \Pr[X \ge 10]$.

\subsection{Conditional Binomial Random Variables:}

We occasionally need to know the expected value of a binomially distributed random variable $X \sim \Bin(n,p)$ conditional on that the variable has at least a certain value $k$. An intuitive (but wrong) argument is that  $E[X \mid X \ge k]$ should be around $k+p(n-k)$, because we know already that $k$ of the $n$ independent trials are successes and the remaining $(n-k)$ trials still have their independent success probability of $p$. While this argument is wrong (as we might need more than $k$ trials to have $k$ successes), the result is correct as an upper bound as shown in this lemma from~\cite[Lemma~1]{DoerrD18}.

\begin{lemma}\label{lprobcondbinomial}
  Let $X$ be a random variable with binomial distribution with parameters $n$ and $p \in [0,1]$. Let $k \in [0..n]$. Then \[E[X \mid X \ge k] \le k + (n-k)p \le k + E[X].\]
\end{lemma}

\begin{proof}
  Let $X = \sum_{i=1}^n X_i$ with $X_1, \dots, X_n$ being independent binary random variables with $\Pr[X_i = 1] = p$ for all $i \in [1..n]$. Conditioning on $X \ge k$, let $\ell := \min\{i \in [1..n] \mid \sum_{j=1}^i X_j = k\}$. Then \[E[X \mid X \ge k] = \sum_{i = 1}^n \Pr[\ell = i \mid X \ge k ] E[X \mid \ell = i].\] Note that $\ell \ge k$ by definition. Note also that $(X \mid \ell = i) = k + \sum_{j = i+1}^n X_j$ with unconditioned $X_j$. In particular, $E[X \mid \ell = i] = k + (n-i)p$. Consequently, 
\begin{align*}
  E[X \mid X \ge k]& = \sum_{i = 1}^n \Pr[\ell = i \mid X \ge k] E[X \mid \ell = i] \\
  &\le \sum_{i=k}^n \Pr[\ell = i \mid X \ge k ] (k + (n-k)p) = k + (n-k)p.
\end{align*}
\end{proof}

We note that, in the language introduced in the following section, we have actually shown the stronger statement that $(X \mid X \ge k)$ is dominated by $k + \Bin(n-k,p)$. This stronger version can be useful to obtain tail bounds for $(X \mid X \ge k)$.

\section{Stochastic Domination and Coupling}

In this section, we discuss two concepts that are not too often used explicitly, but where we feel that mastering them can greatly help in the analysis of randomized search heuristics. The first of these is \emph{stochastic domination}\index{stochastic domination}, which is a very strong way of saying that one random variable is better than another even when they are not defined on the same probability space. The second concept is \emph{coupling}, which means defining two random variables suitably over the same probability space to facilitate comparing them. These two concepts are strongly related: If the random variable $Y$ dominates $X$, then $X$ and $Y$ can be coupled in a way that $Y$ is point-wise not smaller than $X$, and vice versa. The results of this section and some related ones have appeared, in a more condensed form, in~\cite{Doerr18evocop}.

\subsection{The Notion of Stochastic Domination}

Possibly the first to use the notion of stochastic domination was Droste, who in~\cite{Droste03,Droste04} employed it to make precise an argument often used in an informal manner, namely that some artificial random process is not faster than the process describing a run of the algorithm under investigation.  

\begin{definition}[Stochastic domination]
  Let $X$ and $Y$ be two random variables not necessarily defined on the same probability space. We say that $Y$ stochastically dominates $X$, written as $X \preceq Y$, if for all $\lambda \in \R$ we have $\Pr[X \le \lambda] \ge \Pr[Y \le \lambda]$.
\end{definition}

If $Y$ dominates $X$, then the cumulative distribution function of $Y$ is point-wise not larger than the one of $X$. The definition of domination is equivalent to \[\forall \lambda \in \R : \Pr[X \ge \lambda] \le \Pr[Y \ge \lambda],\] which is maybe a formulation making it more visible why we feel that $Y$ is at least as large as $X$. 

Concerning nomenclature, we remark that some research communities in addition require that the inequality is strict for at least one value of $\lambda$. Hence, intuitively speaking, $Y$ is strictly larger than $X$. From the mathematical perspective, this appears not very practical. Consequently, our definition above is  more common in computer science. We also note that stochastic domination is sometimes called \emph{first-order stochastic domination}. For an extensive treatment of various forms of stochastic orders, we refer to~\cite{MullerS02}.

The usual way of explaining stochastic domination is via games. Let us consider the following three games.\\

\begin{samepage}
\noindent \textbf{Game A:} With probability $\frac12$, each you win 500 and 1500.\\
\textbf{Game B:} With probability $\frac13$, you win 500, with probability $\frac16$, you win 800, and with probability $\frac12$, you win 1500.\\
\textbf{Game C:} With probability $\frac{1}{1000}$, you win 2,000,000. Otherwise, you win nothing.\\
\end{samepage}

Which of these games is best to play? It is intuitively clear that you prefer Game~B over Game~A. However, it is not clear whether you should prefer Game~C over Game~B. Clearly, the expected win in Game C is 2000 compared to only 1050 in Game~B. However, the chance of winning something at all is really small in Game~C. If you do not like to go home empty-handed, you might prefer Game~B. 

The mathematical take on these games is that the random variable $X_B$ describing the win in Game B stochastically dominates the one $X_A$ for Game A. This captures our intuitive feeling that it cannot be wrong to prefer Game B over Game A. For Games B and C, neither of $X_B$ and $X_C$ dominates the other. Consequently, it depends on the precise utility function of the player which game he prefers. This statement is made precise in the following lemma.

\begin{lemma}
  The following two conditions are equivalent.
  \begin{enumerate}
	  \item $X \preceq Y$.
	  \item For all monotonically non-decreasing functions $f : \R \to \R$, we have \[E[f(X)] \le E[f(Y)].\]
  \end{enumerate}
\end{lemma}

As a simple corollary we note the following.
\begin{corollary}\label{corprobdomexp}
  If $X \preceq Y$, then $E[X] \le E[Y]$.
\end{corollary}

We note another simple, but useful property.
\begin{lemma}\label{lprobdomsum}
  Let $X_1, \dots, X_n$ be independent random variables defined over some common probability space. Let $Y_1, \dots, Y_n$ be independent random variables defined over a possibly different probability space. If $X_i \preceq Y_i$ for all $i \in [1..n]$, then \[\sum_{i=1}^n X_i \preceq \sum_{i=1}^n Y_i.\]
\end{lemma}

For discrete random variables, this result is a special case of Lemma~\ref{lprobmoderate} stated further below.

Finally, we note two trivial facts.
\begin{lemma}\label{lprobtrivialdom}
Let $X$ and $Y$ be random variables.
\begin{enumerate}
  \item If $X$ and $Y$ are defined on the same probability space and $X \le Y$, then $X \preceq Y$.
  \item If $X$ and $Y$ are identically distributed, then $X \preceq Y$.
\end{enumerate}
\end{lemma}

\subsection{Stochastic Domination in Runtime Analysis}\label{secprobdomruntime}

From the perspective of algorithms analysis, stochastic domination allows to very clearly state that one algorithm is better than another. If the runtime distribution $X_A$ of algorithms A dominates the one $X_B$ of Algorithm B, then from the runtime perspective Algorithm~B is always preferable to Algorithm~A.

In a similar vein, we can use domination also to give more detailed descriptions of the runtime of an algorithm. For almost all algorithms, we will not be able to precisely determine the runtime distribution. However, via stochastic domination we can give a lot of useful information beyond, say, just the expectation. We demonstrate this via an extension of the classic fitness level method, which is implicit in Zhou, Luo, Lu, and Han~\cite{ZhouLLH12}.

\begin{theorem}[Domination version of the fitness level method]\label{tproblevel}
  Consider an iterative randomized search heuristic $\calA$ maximizing a function $f : \Omega \to \R$. Let $A_1, \dots, A_m$ be a partition of $\Omega$ such that for all $i, j \in [1..m]$ with $i < j$ and all $x \in A_i$, $y \in A_j$, we have $f(x) < f(y)$. Set $A_{\ge i} := A_{i} \cup \dots \cup A_{m}$. Let $p_1, \dots, p_{m-1}$ be such that for all $i \in [1..m-1]$ we have that if the best-so-far search point is in $A_i$, then regardless of the past $\calA$ has a probability of at least $p_i$ to generate a search point in $A_{\ge i+1}$ in the next iteration. 
  
  Denote by $T$ the (random) number of iterations $\calA$ takes to generate a search point in $A_m$. Then \[T \preceq \sum_{i=1}^{m-1} \Geom(p_i),\] where this sum is to be understood as a sum of independent geometric distributions. 
\end{theorem}

To prove this theorem, we need a technical lemma which we defer to the subsequent subsection to ease reading this part.

\begin{proof}
  Consider a run of the algorithm $\calA$. For all $i \in [1..m]$, let $T_i$ be the first time (iteration) when $\calA$ has generated a search point in $A_{\ge i}$. Then $T = T_m = \sum_{i=1}^{m-1} (T_{i+1} - T_{i})$. By assumption, $T_{i+1} - T_i$ is dominated by a geometric random variable with parameter $p_i$ regardless what happened before time $T_i$. Consequently, Lemma~\ref{lprobmoderate} gives the claim.
\end{proof}

Note that a result like Theorem~\ref{tproblevel} implies various statements on the runtime. By Corollary~\ref{corprobdomexp}, the expected runtime satisfies $E[T] \le \sum_{i=1}^{m-1} \frac 1{p_i}$, which is the common version of the fitness level theorem~\cite{Wegener01}. By using tail bounds for sums of independent geometric random variables (see Section~\ref{secprobcgeom}), we also obtain runtime bounds that hold with high probability. This was first proposed in~\cite{ZhouLLH12}. We defer a list of examples where previous results can profitably be turned into a domination statement to Section~\ref{secprobcgeom}, where we will also have the large deviation bounds to exploit such statements.

\subsection{Domination by Independent Random Variables}\label{secprobmoderate}

A situation often encountered in the analysis of algorithms is that a sequence of random variables is not independent, but that each member of the sequence has a good chance of having a desired property no matter what was the outcome of its predecessors. In this case, the random variables in some sense can be treated as if they were independent.

\begin{lemma}\label{lprobmoderatebinary}
Let $X_1,\dots,X_n$ be arbitrary binary random variables and let $X^*_1,\dots,X^*_n$ be independent binary random variables.
\begin{enumerate}
\item If we have \[\Pr[X_i=1 \mid X_1=x_1,\dots,X_{i-1}=x_{i-1}]\le \Pr[X_i^*=1]\] for all $i \in [1..n]$ and all $x_1,\dots, x_{i-1}\in\{0,1\}$ with $\Pr[X_1=x_1,\dots,X_{i-1}=x_{i-1}]>0$, then \[\sum_{i=1}^n X_i \preceq \sum_{i=1}^n X_i^*.\]
\item If we have \[\Pr[X_i=1 \mid X_1=x_1,\dots,X_{i-1}=x_{i-1}]\ge \Pr[X_i^*=1]\] for all $i \in [1..n]$ and all $x_1,\dots, x_{i-1}\in\{0,1\}$ with $\Pr[X_1=x_1,\dots,X_{i-1}=x_{i-1}]>0$, then \[\sum_{i=1}^n X_i^* \preceq \sum_{i=1}^n X_i.\]
\end{enumerate}
\end{lemma}

Note that here and in the following, we view ``$X_1=x_1,\dots,X_{i-1}=x_{i-1}$'' for $i = 1$ as an empty intersection of events, that is, an intersection over an empty index set. As most textbooks, we define this to be the whole probability space. 

Both parts of the lemma are simple corollaries from the following, slightly technical, general result, which might be of independent interest. 

For two sequences $(X_1, \dots, X_n)$ and $(X_1^*,\dots,X_n^*)$ of random variables, we say that $(X_1^*,\dots,X_n^*)$ \emph{unconditionally sequentially dominates} $(X_1, \dots, X_n)$ if for all $i \in [1..n]$ and all $x_1,\dots, x_{i-1} \in \R$ with $\Pr[X_1=x_1,\dots,X_{i-1}=x_{i-1}]>0$, we have $(X_i \mid X_1=x_1,\dots,X_{i-1}=x_{i-1}) \preceq X_i^*$. Analogously, we speak of \emph{unconditional sequential subdomination} if the last condition is replaced by $X_i^* \preceq (X_i \mid X_1=x_1,\dots,X_{i-1}=x_{i-1})$.

The following lemma shows that unconditional sequential (sub-)domination and independence of the $X_i^*$ imply (sub-)domination for the sums of these random variables. Note that unconditional sequential (sub-)domination is inherited by subsequences, so the following lemma immediately extends to sums over arbitrary subsets $I$ of the index set $[1..n]$. 

\begin{lemma}\label{lprobmoderate}
  Let $X_1,\dots,X_n$ be arbitrary discrete random variables. Let $X^*_1,\dots,X^*_n$ be independent discrete random variables. 
  \begin{enumerate}
	  \item If $(X_1^*,\dots,X_n^*)$ unconditionally sequentially dominates $(X_1, \dots, X_n)$, then $\sum_{i=1}^n X_i \preceq \sum_{i=1}^n X_i^*$.
	  \item If $(X_1^*,\dots,X_n^*)$ unconditionally sequentially subdominates $(X_1, \dots, X_n)$, then $\sum_{i=1}^n X_i^* \preceq \sum_{i=1}^n X_i$.
  \end{enumerate}
\end{lemma}

\begin{proof}
The two parts of the lemma imply each other (as seen by multiplying the random variables with $-1$), so it suffices to prove the first statement.

Since the statement of the theorem is independent of the correlation between the $X_i$ and the $X^*_i$, we may assume that they are independent. Let $\lambda \in \R$. Define \[P_j:=\Pr\Big[\sum_{i=1}^j X_i+\sum_{i=j+1}^n X^*_i\ge \lambda\Big]\] for $j \in [0..n]$. We show $P_{j+1} \le P_j$ for all $j \in [0..n-1]$. 

For $m \in \R$, let $\Omega_m$ denote the set of all $(x_1,\dots,x_j,x_{j+2},\dots,x_n) \in \R^{n-1}$ such that $\Pr[X_1=x_1,\dots,X_{j}=x_{j}] > 0$ and $\sum_{i\in [1..n]\setminus\{j+1\}} x_i = \lambda - m$. Let $M := \{m \in \R \mid \Omega_m \neq \emptyset\}$. Then 
 \begin{align*}
   P_{j+1} {}={}&\Pr\Big[\sum_{i=1}^{j+1} X_i+\sum_{i=j+2}^n X^*_i\ge \lambda\Big] \\
    ={}& \sum_{m \in M} \Pr\Big[\sum_{i=1}^{j} X_i+\sum_{i=j+2}^n X^*_i = \lambda-m \wedge X_{j+1}\ge m\Big]\\
    ={}& \sum_{m \in M} \sum_{(x_1,\dots,x_j,x_{j+2},\dots,x_n)\in \Omega_m}  \Pr\big[X_1=x_1,\dots,X_{j}=x_{j}\Big]\cdot \\
    {}&\Pr\big[X_{j+1} \ge m \,\big|\, X_1=x_1,\dots,X_j=x_j\big]\cdot \prod_{i=j+2}^n \Pr\big[X^*_i=x_i\big] \\
    \le{}& \sum_{m \in M} \Pr\Big[\sum_{i=1}^{j} X_i+\sum_{i=j+2}^n X^*_i = \lambda-m\Big]\cdot \Pr\big[X^*_{j+1}\ge m\big]\\
    ={}& \Pr\Big[\sum_{i=1}^{j} X_i+\sum_{i=j+1}^n X^*_i\ge \lambda\Big]\\
    ={}& P_j.
 \end{align*}
 Thus, we have 
  \[
    \Pr\Big[\sum_{i=1}^n X_i\ge \lambda\Big]=P_n \le P_{n-1} \le \dots \le P_1 \le P_0=\Pr\Big[\sum_{i=1}^n X_i^*\ge \lambda\Big].
  \]
\end{proof}

\subsection{Coupling}

Coupling is an analysis technique that consists of defining two unrelated random variables over the same probability space to ease comparing them. As an example, let us regard standard-bit mutation with rate $p$ and with rate $q$, where $p < q$. Intuitively, it seem obvious that we flip more bits when using the higher rate $q$. We could make this precise by looking at the distributions of the random variables $X_p$ and $X_q$ describing the numbers of bits that flip and computing that $X_p \preceq X_q$. For that, we would need to show that for all $k \in [0..n]$, we have \[\sum_{i = 0}^k \binom ni p^i (1-p)^{n-i} \ge \sum_{i = 0}^k \binom ni q^i (1-q)^{n-i}.\]
Coupling is a way to get the same result in a more natural manner. 

Consider the following random experiment. For each $i \in [1..n]$, let $r_i$ be a random number chosen independently and uniformly distributed in $[0,1]$. Let $\tilde X_p$ be the number of $r_i$ that are less than $p$ and let $X_q$ be the number of $r_i$ that are less than $q$. We immediately see that $\tilde X_p \sim \Bin(n,p)$ and $\tilde X_q \sim \Bin(n,q)$. However, we know more. We defined $\tilde X_p$ and $\tilde X_q$ in a common probability space in a way that we have $\tilde X_p \le \tilde X_q$ with probability one: $X_p$ and $X_q$, viewed as functions on the (hidden) probability space $\Omega = \{(r_1,\dots,r_n) \mid r_1, \dots, r_n \in [0,1]\}$ satisfy $\tilde X_p(\omega) \le \tilde X_q(\omega)$ for all $\omega \in \Omega$. Consequently, by the trivial Lemma~\ref{lprobtrivialdom}, we have $X_p \preceq \tilde X_p \preceq \tilde X_q \preceq X_q$ and hence $X_p \preceq X_q$. 

The same argument works for geometric distributions. We summarize these findings (and two more) in the following lemma. Part~\ref{it:probdomnm} follows from the obvious embedding (which is a coupling as well) of the $\Bin(n,p)$ probability space into the one of $\Bin(m,p)$. The first inequality of part~\ref{it:probdomcawi} is easily computed right from the definition of domination (and holds in fact for all random variables), the second part was proven in~\cite[Lemma~1]{KrejcaW17}.

\begin{lemma}\label{lprobdomdistr}
Let $X$ and $Y$ be two random variables. Let $p, q \in [0,1]$ with $p \le q$.
\begin{enumerate}
	\item If $X \sim \Bin(n,p)$ and $Y \sim \Bin(n,q)$, then $X \preceq Y$.
	\item\label{it:probdomnm} If $n \le m$, $X \sim \Bin(n,p)$ and $Y \sim \Bin(m,p)$, then $X \preceq Y$.
	\item\label{it:probdomcawi} If $X \sim \Bin(n,p)$ and $x \in [0..n]$, then $X \preceq (X \mid X \ge x) \preceq (X+x)$.
	\item If $p > 0$, $X \sim \Geom(p)$, and $Y \sim \Geom(q)$, then $X \preceq Y$.
\end{enumerate}
\end{lemma}  

Let us now formally define what we mean by coupling. Let $X$ and $Y$ be two random variables not necessarily defined over the same probability space. We say that $(\tilde X, \tilde Y)$ is a \emph{coupling} of $(X,Y)$ if $\tilde X$ and $\tilde Y$ are defined over a common probability space and if $X$ and $X'$ as well as $Y$ and $Y'$ are identically distributed. 

This definition itself is very weak. $(X,Y)$ have many couplings and most of them are not interesting. So the art of coupling as a proof and analysis technique is to find a coupling of $(X,Y)$ that allows to derive some useful information.

It is not a coincidence that we could use coupling to prove stochastic domination. The following theorem is well-known. 
\begin{theorem}\label{tprobdomcou}
  Let $X$ and $Y$ be two random variables. Then the following two statements are equivalent.
  \begin{enumerate}
	\item $X \preceq Y$.
	\item There is a coupling $(\tilde X, \tilde Y)$ of $(X,Y)$ such that $\tilde X \le \tilde Y$.
\end{enumerate}
\end{theorem}

We remark without giving much detail that coupling as a proof technique found numerous powerful applications beyond its connection to stochastic domination. In the analysis of population-based evolutionary algorithms, a powerful strategy to prove lower bounds is to couple the true population of the algorithm with the population of an artificial process without selection and by this overcome the difficult dependencies introduced by the variation-selection cycle of the algorithm. This was first done in~\cite{Witt06} and~\cite{Witt08} for the analysis of the \mpoea and an elitist steady-state GA. This technique then found applications for memetic algorithms~\cite{Sudholt09}, aging-mechanisms~\cite{JansenZ11tcs}, non-elitist algorithms~\cite{LehreY12}, multi-objective evolutionary algorithms~\cite{DoerrKV13}, and the \mplea~\cite{AntipovDFH18}.

\subsection{Domination in Fitness or Distance}

So far we have used stochastic domination to compare runtime distributions. We now show that stochastic domination is a powerful proof tool also when applied to other distributions. To do so, we give a short and elegant proof for the result of Witt~\cite{Witt13} that compares the runtimes of mutation-based algorithms. The main reason why our proof is significantly shorter than the one of Witt is that we use the notion of stochastic domination also for the distance from the optimum. This will also be an example where we heavily exploit the connection between coupling and stochastic domination (Theorem~\ref{tprobdomcou}).

To state this result, we need the notion of a \emph{$(\mu,p)$ mutation-based algorithm} introduced in~\cite{Sudholt13}. This class of algorithms is  called only \emph{mutation-based} in~\cite{Sudholt13}, but since (i)~it does not include all adaptive algorithms using mutation only, e.g., those regarded in~\cite{JansenW06,OlivetoLN09,BottcherDN10,BadkobehLS14,DangL16self,DoerrGWY17,DoerrWY18}, (ii)~it does not include all algorithms using a different mutation operator than standard-bit mutation, e.g., those in~\cite{DoerrDY16PPSN,DoerrDY16gecco,LissovoiOW17,DoerrLMN17}, and (iii)~this notion collides with the notion of unary unbiased black-box complexity algorithms (see~\cite{LehreW12}), which without greater justification could also be called the class of mutation-based algorithms, we feel that a notion making these restrictions precise is more appropriate.

The class of $(\mu,p)$ mutation-based algorithms comprises all algorithms which first generate a set of $\mu$ search points uniformly and independently at random from $\{0,1\}^n$ and then repeat generating new search points from any of the previous ones via standard-bit mutation with probability~$p$. This class includes all $(\mu+\lambda)$ and $(\mu,\lambda)$ EAs which only use standard-bit mutation with static mutation rate~$p$.

Denote by \oeamu the following algorithm in this class. It first generates $\mu$ random search points. From these, it selects uniformly at random one with highest fitness and then continues from this search point as a \oea, that is, repeatedly generates a new search point from the current one via standard-bit mutation with rate $p$ and replaces the previous one by the new one if the new one is not worse (in terms of the fitness). This algorithm was called \oea with BestOf($\mu$) initialization in~\cite{LaillevaultDD15}.

For any algorithm $\calA$ from the class of $(\mu,p)$ mutation-based algorithms and any fitness function $f : \{0,1\}^n \to \R$, let us denote by $T(\calA,f)$ the runtime of the algorithm $\calA$ on the fitness function $f$, that is, the number of the first individual generated that is an optimal solution. Usually, this will be $\mu$ plus the number of the iteration in which the optimum was generated. To cover also the case that one of the random initial individuals is optimal, let us assume that these initial individuals are generated sequentially. As a final technicality, for reasons of convenience, let us assume that the \oeamu in iteration $\mu+1$ does not choose as parent a random one of the previous search points with maximal fitness, but the last one with maximal fitness. Since the first $\mu$ individuals are generated independently, this modification does not change the distribution of this parent.

In this language, Witt~\cite[Theorem~6.2]{Witt13} shows the following remarkable result.
\begin{theorem}\label{tprobdom}
  For any $(\mu,p)$ mutation-based algorithm $\calA$ and any $f : \{0,1\}^n \to \R$ with unique global optimum,
  \[T(\mbox{\oeamu},\onemax) \preceq T(\calA,f).\]
\end{theorem}
This result significantly extends results of a similar flavor in~\cite{BorisovskyE08,DoerrJW12algo,Sudholt13}. The importance of such types of results is that they allow to prove lower bounds for the performance of many algorithm on essentially arbitrary fitness functions by just regarding the performance of the \oeamu on \onemax. 

Let us denote by $|x|_1$ the number of ones in the bit string $x \in \{0,1\}^n$. In other words, $|x|_1 = \|x\|_1$, but the former is nicer to read. Then Witt~\cite[Lemma~6.1]{Witt13} has shown the following natural domination relation between offspring generated via standard-bit mutation.
\begin{lemma}\label{lprobdommut}
  Let $x, y \in \{0,1\}^n$. Let $p \in [0,\frac 12]$. Let $x', y'$ be obtained from $x, y$ via standard-bit mutation with rate $p$. If $|x|_1 \le |y|_1$, then $|x'|_1 \preceq |y'|_1$.
\end{lemma}

We are now ready to give our alternate proof for Theorem~\ref{tprobdom}. While it is clearly shorter that the original one in~\cite{Witt13}, we also feel that it is more natural. In very simple words, it shows that $T(\calA,f)$ dominates $T(\mbox{\oeamu},\onemax)$ because the search points generated in the run of the \oeamu on $\onemax$ always are at least as close to the optimum (in the domination or coupling sense) as the ones in the run of $\calA$ on $f$.

\begin{proof}
  Since $\calA$ treats bit-positions and bit-values in a symmetric fashion, we may without loss of generality assume that the unique optimum of $f$ is $(1,\dots,1)$. 
  
  Let $x^{(1)}, x^{(2)}, \dots$ be the sequence of search points generated in a run of $\calA$ on the fitness function $f$. Hence $x^{(1)}, \dots, x^{(\mu)}$ are independently and uniformly distributed in $\{0,1\}^n$ and all subsequent search points are generated from suitably chosen previous ones via standard-bit mutation with rate $p$. Let $y^{(1)}, y^{(2)}, \dots$ be the sequence of search points generated in a run of the \oeamu on the fitness function $\onemax$. 
  
  We show how to couple these random sequences of search points in a way that $|\tilde x^{(t)}|_1 \le |\tilde y^{(t)}|_1$ for all $t \in \N$. We take as common probability space $\Omega$ simply the space that $(x^{(t)})_{t \in \N}$ is defined on and let $\tilde x^{(t)} = x^{(t)}$ for all $t \in \N$. 
  
  We define the $\tilde y^{(t)}$ inductively as follows. For $t \in [1..\mu]$, let $\tilde y^{(t)} = x^{(t)}$. Note that this trivially implies $|\tilde x^{(t)}|_1 \le |\tilde y^{(t)}|_1$ for these search points. Let $t > \mu$ and assume that $|\tilde x^{(t')}|_1 \le |\tilde y^{(t')}|_1$ for all $t' < t$. Let $s \in [1..t-1]$ be maximal such that $\tilde y^{(s)}$ has maximal $\onemax$-fitness among $\tilde y^{(1)}, \dots, \tilde y^{(t-1)}$. Let $r \in [1..t-1]$ be such that $x^{(t)}$ was generated from $x^{(r)}$ in the run of $\calA$ on $f$. By induction, we have $|x^{(r)}|_1 \le |\tilde y^{(r)}|_1$. By the choice of $s$ we have $|\tilde y^{(r)}|_1 \le |\tilde y^{(s)}|_1$. Consequently, we have $|x^{(r)}|_1 \le |\tilde y^{(s)}|_1$. By Lemma~\ref{lprobdommut} and Theorem~\ref{tprobdomcou}, there is a random $\tilde y^{(t)}$ (defined on $\Omega$) such that $\tilde y^{(t)}$ has the distribution of being obtained from $\tilde y^{(s)}$ via standard-bit mutation with rate $p$ and such that $|x^{(t)}|_1 \le |\tilde y^{(t)}|_1$. 
  
  With this construction, the sequence $(\tilde y^{(t)})_{t \in \N}$ has the same distribution as $(y^{(t)})_{t \in \N}$. This is because the first $\mu$ elements are random and then each subsequent one is generated via standard-bit mutation from the current-best one, which is just the way the \oeamu is defined. At the same time, we have $|\tilde x^{(t)}|_1 \le |\tilde y^{(t)}|_1$ for all $t \in \N$. Consequently, we have $\min\{t \in \N \mid |\tilde y^{(t)}|_1 = n\} \le \min\{t \in \N \mid |x^{(t)}|_1 = n\}$. Since  $T(\mbox{\oeamu},\onemax)$ and $\min\{t \in \N \mid |\tilde y^{(t)}|_1 = n\}$ are identically distributed and also $T(\calA,f)$ and $\min\{t \in \N \mid |x^{(t)}|_1 = n\}$ are identically distributed, we have $T(\mbox{\oeamu},\onemax) \preceq T(\calA,f)$. 
\end{proof}

While not explicitly using the notion of stochastic domination, the result and proof in~\cite{BorisovskyE08} bear some similarity to those above. In very simple words and omitting many details, the result~\cite[Theorem~1]{BorisovskyE08} states the following. Assume that you run the \oea and some other algorithm $A$ (from a relatively large class of algorithms) to maximize a function $f$. Denote by $x^{(t)}$ and $y^{(t)}$ the best individual produced by the \oea and $A$ up to iteration $t$. Assume that for all $t$ and all possible runs of the algorithms up to iteration $t$ we have that $f(x^{(t)}) \ge f(y^{(t)})$ implies $f(x^{(t+1)}) \succeq f(y^{(t+1)})$. Assume further that the random initial individual of the \oea is at least as good (in terms of $f$) as all initial individuals of algorithm $A$. Then $f(x^{(t)}) \succeq f(y^{(t)})$ for all $t$.

The proof of this result (like the one of the fitness domination statement in our proof of Theorem~\ref{tprobdom}) uses induction over the time $t$. Since~\cite{BorisovskyE08} do not use the notion of stochastic domination explicitly, they cannot simply couple the two processes, but have to compare the two distributions manually using an argument they call Abel transform. 

\section{The Coupon Collector Process}\label{secprobcoupon}\index{Coupon Collector}

The coupon collector process is one of the central building blocks in the analysis of randomized algorithms. It is particularly important in the theory of randomized search heuristics, where it often appears as a subprocess. 

The \emph{coupon collector process} is the following simple randomized process. Assume that there are $n$ types of coupons available. Whenever you buy a certain product, you get one coupon with a type chosen uniformly at random from the $n$~types. How long does it take until you have a coupon of each type? We denote, in this section,  the random variable describing the first round after which we have all types by $T_n$ and call it \emph{coupon collecting time}. Hence in simple words, this is the number of rounds it takes to obtain all types.

As an easy example showing the coupon collector problem arising in the theory of randomized search heuristics, let us regard how the \emph{randomized local search (RLS)}\index{randomized local search (RLS)} heuristic optimizes strictly monotonically increasing functions. The RLS heuristic when maximizing a given function $f : \{0,1\}^n \to \R$ starts with a random search point. Then, in each iteration of the process, a single random bit is flipped in the current solution. If this gives a solution worse than the current one (in terms of~$f$), then the new solution is discarded. Otherwise, the process is continued from this new solution. 

Assume that $f$ is \emph{strictly monotonically increasing}\index{strictly monotonically increasing}, that is, flipping any $0$-bit to $1$ increases the function value. Then the optimization process of RLS on $f$ strongly resembles a coupon collector process. In each round, we flip a random bit. If this bit was $1$ in our current solution, then nothing changes (we discard the new solution as it has a smaller $f$-value). If this bit was $0$, then we keep the new solution, which now has one extra $1$. Hence taking the $1$- bits as coupons, in each round we obtain a random coupon. This has no effect if we had this coupon already, but is good if not. 

We observe that the optimization time (number of solutions evaluated until the optimal solution is found) of RLS on strictly monotonic functions is exactly the coupon collecting time when we start with an initial stake of coupons that follows a $\Bin(n,\tfrac 12)$ distribution. This shows that the optimization time is at most the classic coupon collector time (where we start with no coupons). See~\cite{DoerrD16} for a very precise analysis of this process.

The expectation of the coupon collecting time is easy to determine. Recall from Section~\ref{secprobharmonic} the definition of the harmonic number $H_n := \sum_{k=1}^n \frac 1k$\index{harmonic number}.

\begin{theorem}[Coupon collector, expectation]\label{tprobcouponE}\index{Coupon Collector}
  The expected time to collect all $n$ coupons is $E[T_n] = n H_n = (1 + o(1)) n \ln n$.
\end{theorem}

\begin{proof}
  Given that we already have $k$ different coupons for some $k \in [0..n-1]$, the probability that the next coupon is one that we do not already have, is $\frac{n-k}{n}$. By the waiting time argument (Lemma~\ref{lprobwaitingtime}), we see that the time $T_{n,k}$ needed to obtain a new coupon given that we have exactly $k$ different ones, satisfies $E[T_{n,k}] = \frac{n}{n-k}$. Clearly, the total time $T_n$ needed to obtain all coupons is $\sum_{k = 0}^{n-1} T_{n,k}$. Hence, by linearity of expectation (Lemma~\ref{lproblinearity}), $E[T_n] = \sum_{k = 0}^{n-1} E[T_{n,k}] = n H_n$.  
\end{proof}

We proceed by trying to gain more information on $T_n$ than just the expectation. The tools discussed so far (and one to come in a later section) lead to the following results.
\begin{itemize}
	\item Markov's inequality (Lemma~\ref{lprobmarkov}) gives $\Pr[T_n \ge \lambda n H_n] \le \frac 1 {\lambda}$ for all $\lambda \ge 1$.
	\item Chebyshev's inequality (Lemma~\ref{lprobchebyshev}) can be used to prove $\Pr[|T_n - n H_n| \ge \eps n] \le \frac{\pi^2}{6\eps^2}$ for all $\eps \ge \frac{6}{\pi^2} \approx 0.6079$. This builds on the fact (implicit in the proof above) that the coupon collector time is the sum of independent geometric random variables $T_n = \sum_{k=0}^{n-1} \Geom(\frac{n-k}{n})$. Hence the variance is $\Var[T_n] = \frac{\pi^2 n^2}{6}$. 
	\item Again exploiting $T_n = \sum_{k=0}^{n-1} \Geom(\frac{n-k}{n})$, Witt's Chernoff bound for geometric random variables (Theorem~\ref{tprobchernoffgeomwitt}) gives 
	\begin{align*}
	&\Pr[T_n \ge E[T_n] + \eps n] \le \begin{cases} \exp(- \tfrac{3\eps^2}{\pi^2}) &\mbox{if } \eps \le \tfrac{\pi^2}{6} \\
\exp(-\tfrac{\eps}{4}) & \mbox{if } \eps > \tfrac{\pi^2}{6} \end{cases} 	\\
	&\Pr[T_n \le E[T_n] - \eps n] \le \exp(- \tfrac{3 \eps^2}{\pi^2})
  \end{align*}
  for all $\eps \ge 0$. See~\cite{Witt14} for the details.
\end{itemize}

Interestingly, asymptotically stronger tail bounds for $T_n$ can be derived by fairly elementary means. The key idea is to not regard how the number of coupons increases over time, but instead to regard the event that we miss a particular coupon for some period of time. Note that the probability that a particular coupon is not obtained for $t$ rounds is $(1 - \frac 1n)^t$. By a union bound\index{Union Bound} argument (see Lemma~\ref{lprobunionbound}), the probability that there is a coupon not obtained within $t$ rounds, and equivalently, that $T_n > t$, satisfies \[\Pr[T_n > t] \le n (1 - \tfrac 1n)^t.\] Using the simple estimate of Lemma~\ref{lprobelower}, we obtain the following (equivalent) bounds.

\begin{theorem}[Coupon collector, upper tail]\label{tprobcouponU}
  For all $\eps \ge 0$, 
  \begin{align}
  &\Pr[T_n \ge (1+\eps)n \ln n] \le n^{-\varepsilon},\\
  &\Pr[T_n \ge n \ln n + \eps n] \le \exp(-\varepsilon).
  \end{align}
\end{theorem}

Surprisingly, prior to the following result from~\cite{Doerr11bookchapter}, no good lower bound for the coupon collecting time was published. 

\begin{theorem}[Coupon collector, lower tail]\label{tprobcouponL}\index{Coupon Collector}
  For all $\varepsilon \ge 0$, 
  \begin{align}
  &\Pr[T_n \le (1 - \varepsilon) (n-1) \ln n] \le \exp(-n^\varepsilon),\\
  &\Pr[T_n \le (n-1) \ln n - \eps(n-1)] \le \exp(-e^\eps).
  \end{align}
\end{theorem}  

Theorem~\ref{tprobcouponL} was proven in~\cite{Doerr11bookchapter} by showing that the events of having a coupon after a certain time are $1$-negatively correlated. The following proof defers this work to Lemma~\ref{lprobmaxhyper}.

\begin{proof}
  Let $t = (1 - \varepsilon) (n-1) \ln n$. For $i \in [1..n]$, let $X_i$ be the indicator random variable for the event that a coupon of type $i$ is obtained within the first $t$ rounds. Then $\Pr[X_i = 1] = 1 - (1 - \frac 1n)^t \le 1 - \exp(-(1-\varepsilon)\ln n) = 1 - n^{-1+\varepsilon}$, where the estimate follows from Corollary~\ref{corprobesbm}.
  
  Since in the coupon collector process in each round $j$ we choose a random set $S_j$ of cardinality $1$, by Lemma~\ref{lprobmaxhyper} the $X_i$ are $1$-negatively correlated. Consequently,
  \begin{align*}
    \Pr[T_n \le (1 - \varepsilon) (n-1) \ln n] &= \Pr[\forall i \in [1..n] : X_i = 1] \\
    &\le \prod_{i = 1}^n \Pr[X_i = 1]\\
    &\le (1 - n^{-1+\varepsilon})^n \le \exp(-n^\varepsilon)
  \end{align*}
by Lemma~\ref{lprobelower}.  
\end{proof}

We may remark that good mathematical understanding of the coupon collector process not only is important because such processes directly show up in some randomized algorithms, but also because it might give us the right intuitive understanding of other processes. Consider, for example, a run of the \oea on some pseudo-Boolean function $f : \{0,1\}^n \to \R$ with a unique global maximum. 

The following intuitive consideration leads us to believe that the \oea with high probability needs at least roughly $n \ln \tfrac n2$ iterations to find the optimum of~$f$: By the strong concentration of the binomial distribution, the initial search point differs in at least roughly $\tfrac n2$ bits from the global optimum. To find the global optimum, it is necessary (but clearly not sufficient) that each of these missing bits is flipped at least once in some mutation step. Now that the \oea in average flips one bit per iteration, this looks like a coupon collector process started with an initial stake of $\tfrac n2$ coupons, so we expect to need at least roughly $n \ln \tfrac n2$ iterations to perform the $n \ln \tfrac n2$ bit-flips necessary to have each missing bit flipped at least once. Clearly, this argumentation is not rigorous, but it suggested to us the right answer. 

\begin{theorem}\label{tproboealower}
  The optimization time $T$ of the \oea on any function  $f : \{0,1\}^n \to \R$ with unique global maximum satisfies \[\Pr[T \le (1-\eps)(n-1)\ln \tfrac n2] \le \exp(-n^{\eps}).\]
\end{theorem}

\begin{proof}
  By symmetry, we may assume that the unique global optimum of $f$ is $(1,\dots,1)$. Let $t = (1-\eps)(n-1)\ln \tfrac n2$. For all $i \in [1..n]$, let $Y_i$ denote the event that the $i$-th bit was zero in the initial search point and was not flipped in any application of the mutation operator in the first $t$ iterations. Let $X_i = 1 - Y_i$. Then $\Pr[X_i=1] = 1 - \frac 12 (1-\frac 1n)^t \le 1 - n^{-1+\eps}$. The events $X_i$ are independent, so we compute
  \begin{align*}
  \Pr[T \le t] &\le \Pr[\forall i \in [1..n] : X_i = 1]\\
  & = \prod_{i=1}^n \Pr[X_i = 1]\\
  & = (1 - n^{-1+\eps})^n = \exp(-n^\eps).
  \end{align*} 
\end{proof}

We stated the above theorem to give a simple example how understanding the coupon collector process can help understanding also randomized search heuristics that do not directly simulate a coupon collecting process. We remark that the theorem above is not best possible, in particular, it does not outrule an expected optimization time of $n \ln \tfrac n2$. In contrast it is known that the optimization time of the \oea on the \onemax function is $E[T] \ge e n \ln n - O(n)$, see~\cite{DoerrFW11}, improving over the minimally weaker bound $E[T] \ge en \ln n - O(n \log\log n)$ from, independently, \cite{DoerrFW10} and \cite{Sudholt13}. By Theorem~\ref{tprobdom}, this lower bound holds for the performance of the \oea on any function $f : \{0,1\}^n \to \R$ with unique optimum.

\section{Large Deviation Bounds}\label{secproblargedev}

Often, we are not only interested in the expectation\index{expectation} of some random variable, but we need a bound that holds with high probability. We have seen in the proof of Theorem~\ref{tprobneedle2} that such high-probability statements can be very useful: If a certain bad event occurs in each iteration with a very small probability only, then a simple union bound is enough to argue that this event is unlikely to occur even over a large number of iterations. The better the original high-probability statement is, the more iterations we can cover. For this reason, the tools discussed in this chapter are among the most employed in the theory of randomized search heuristics. 

Since computing the expectation often is easy, a very common approach is to first compute the expectation of a random variable and then bound the probability that the random variable deviates from this expectation by a too large amount. The tools for this second step are called \emph{tail inequalities} or \emph{large deviation inequalities}, and this is the topic of this section. In a sense, Markov's and Chebyshev's inequality discussed in Section~\ref{secprobexp} can be seen as large deviation inequalities as well, but usually the term is reserved for exponential tail bounds.

A large number of large deviation bounds have been developed in the past. They differ in the situations they are applicable to, but also in their sharpness. Often, the sharpest bounds give expressions for the tail probability that are very difficult to work with. Hence some experience is needed to choose a tail bound that is not overly complicated, but sharp enough to give the desired result. 

To give the novice to this topic some orientation, here is a short list of results that are particularly useful and which are sufficient in many situations.
\begin{enumerate}
	\item The simple multiplicative Chernoff bounds~\eqref{eqprobCMUeasy} and~\eqref{eqprobCMLeasy} showing that for sums of independent $[0,1]$ random variables, a constant factor deviation from the expectation occurs only with probability negatively exponential in the expectation.
	\item The additive Chernoff bound of Theorem~\ref{tprobchernoffadditive01} showing that a sum of $n$~independent $[0,1]$ random variables deviates from the expectation by more than an additive term of $\lambda$ only with probability $\exp(-\frac{2\lambda^2}{n})$.
	\item The fact that essentially all large deviation bounds can be used also with a pessimistic estimate for the expectation instead of the precise expectation (Section~\ref{secprobEE}).
	\item The method of bounded differences (Theorem~\ref{tprobboundeddiff}), which states that the additive Chernoff bounds remain valid if $X$ is functionally dependent of independent random variables each having a small influence on $X$.
\end{enumerate}

For the experienced reader, the following results may be interesting as they go beyond what most introductions to tail bounds cover.
\begin{enumerate}
	\item In Section~\ref{secprobnegcor} we show that essentially all large deviation bounds usually stated for sums of independent random variables are also valid for negatively correlated random variables. An important application of this result are distributions arising from sampling without replacement or with partial replacement.
  \item In Section~\ref{secprobcgeom}, we present a number of large deviation bounds for sums of independent geometrically distributed random variables. These seem to be particularly useful in the analysis of randomized search heuristics, whereas they are rarely used in classic randomized algorithms.
  \item In Theorem~\ref{tprobboundedexp}, we present a version of the bounded differences method which only requires that the $t$-th random variable has a bounded influence on the \emph{expected} outcome resulting from variables $t+1$ to $n$. This is much weaker than the common bounded differences assumption that each random variable, regardless how we condition on the remaining variables, has a bounded influence on the result. We feel that this new version (which is an easy consequence of known results) may be very useful in the analysis of iterative improvement heuristics. In particular, it may lead to elementary proofs for results which so far could only be proven via tail bounds for martingales.
\end{enumerate}

\subsection{Chernoff Bounds for Sums of Independent Bounded Random Variables}\label{secprobchernoffindependent}

In this longer subsection, we assume that our random variable of interest is the sum of $n$~independent random variables, each taking values in some bounded range, often $[0,1]$. While some textbooks present these bounds for discrete random variables, e.g., taking the values $0$ and $1$ only, all the results are true without this restriction.

The bounds presented below are all known under names like \emph{Chernoff} or \emph{Hoeffding} bounds, referring to the seminal papers by Chernoff~\cite{Chernoff52} and Hoeffding~\cite{Hoeffding63}. Since the first bounds of this type were proven by Bernstein~\cite{Bernstein24}---via the so-called exponential moments method that is used in essentially all proofs of such results, see Section~\ref{secprobcproof}---the name Bernstein inequalities would be more appropriate. We shall not be that precise and instead use the most common name \emph{Chernoff inequalities} for all such bounds. 

For the readers' convenience, as in the remainder of this chapter, we shall not be shy to write out also minor reformulations of some results. We believe that it helps a lot to have seen them and we think that it is convenient, both for using the bounds or for referring to them, if all natural version are visible in the text.

\subsubsection{Multiplicative Chernoff Bounds for the Upper Tail}

The multiplicative Chernoff bounds presented in this and the next section bound the probability to deviate from the expectation by at least a given factor. Since in many algorithmic analyses we are interested only in the asymptotic order of magnitude of some quantity, a constant-factor deviation can be easily tolerated and knowing that larger deviations are very unlikely is just what we want to know. For this reason, the multiplicative Chernoff bounds are often the right tool. 

The following theorem collects a number of bounds for the upper tail, that is, for deviations above the expectation.

\begin{theorem}\label{tprobCMU}
  Let $X_1, \ldots, X_n$ be independent random variables taking values in $[0,1]$. Let $X = \sum_{i = 1}^n X_i$. Let $\delta \ge 0$. Then 
  \begin{align}
  \Pr[X \ge (&1+\delta) E[X]] \nonumber\\
  &\le \bigg(\frac{1}{1+\delta}\bigg)^{(1+\delta)E[X]} \bigg(\frac{n-E[X]}{n-(1+\delta)E[X]}\bigg)^{n-(1+\delta)E[X]}\label{eqprobCMUstrongest}\\ 
  &\le \bigg(\frac{e^\delta}{(1+\delta)^{1+\delta}}\bigg)^{E[X]} = \exp(-((1+\delta)\ln(1+\delta)-\delta)E[X])\label{eqprobCMUstrong}\\ 
  &\le \exp\bigg(-\frac{\delta^2 E[X]}{2 + \frac 23 \delta}\bigg)\label{eqprobCMUlin1}\\
  &\le \exp\bigg(-\frac{\min\{\delta^2,\delta\} E[X]}{3}\bigg),\label{eqprobCMUlin2}
  \end{align}
  where the bound in~\eqref{eqprobCMUstrongest} is read as $0$ for $\delta > \frac{n-E[X]}{E[X]}$ and as $(\frac{E[X]}{n})^n$ for $\delta = \frac{n-E[X]}{E[X]}$. For $\delta \le 1$, equation~\eqref{eqprobCMUlin2} simplifies to 
  \begin{equation}
  \Pr[X \ge (1+\delta) E[X]] \le \exp\bigg(-\frac{\delta^2 E[X]}{3}\bigg).\label{eqprobCMUeasy}
  \end{equation}
\end{theorem}

\begin{figure}
\begin{tikzpicture}
\newcommand{\probCMUxmax}{2.5}
\begin{axis}[
    axis lines = left,
    xlabel = $x$,
    ylabel = {$f(x)$},
    ymin = 0,
    legend style={
      cells={anchor=east},
      legend pos=outer north east,},
    ]

\addplot [
    domain=0.1:\probCMUxmax, 
    samples=100, 
    color=black,
]
{(1+x)*ln(1+x)-x};
\addlegendentry[right]{$f(x) = (1+x) \ln(1+x) - x$}
 
\addplot [densely dotted,
    domain=0.1:\probCMUxmax, 
    samples=100, 
    color=black,
    ]
    {x^2 / (2 + 2*x/3)};
\addlegendentry[right]{$f(x) = \frac{x^2}{2+2x/3}$}

\addplot [dotted,
    domain=0.1:\probCMUxmax, 
    samples=100, 
    color=black,
    ]
    {min(x^2,x) / 3};
\addlegendentry[right]{$f(x) =  \min\{x^2,x\} / 3$}
 
%

 
\end{axis}
\end{tikzpicture}
\caption{Visual comparison of the bounds~\eqref{eqprobCMUstrong},~\eqref{eqprobCMUlin1},~\eqref{eqprobCMUlin2}. Depicted is the term $f(x)$ leading to the bound $\Pr[X \ge (1+x) E[X]] \le \exp(-f(x) E[X])$.}\label{figprobCMU}
\end{figure}

The first and strongest bound~\eqref{eqprobCMUstrongest} was first stated explicitly by Hoeffding~\cite{Hoeffding63}. It improves over Chernoff's~\cite{Chernoff52} tail bounds in particular by not requiring that the $X_i$ are identically distributed. Hoeffding also shows that~\eqref{eqprobCMUstrongest} is the best bound that can be shown via the exponential moments methods under the assumptions of Theorem~\ref{tprobCMU}. 

For $E[X]$ small, say $E[X] = o(n)$ when taking a view asymptotic in~$n \to \infty$, the second bound~\eqref{eqprobCMUstrong} is easier to use, but essentially as strong as~\eqref{eqprobCMUstrongest}. More precisely, it is larger only by a factor of $(1+o(1))^{E[X]}$, since we estimated 
\begin{equation}
\bigg(\frac{n-E[X]}{n-(1+\delta)E[X]}\bigg)^{n-(1+\delta)E[X]} = \bigg(1 + \frac{\delta E[X]}{n-(1+\delta)E[X]}\bigg)^{n-(1+\delta)E[X]} \le e^{\delta E[X]} \label{eqprobeweg}
\end{equation}
using Lemma~\ref{lprobelower}.

Equation~\eqref{eqprobCMUlin1} is derived from~\eqref{eqprobCMUstrong} by noting that $(1+\delta)\ln(1+\delta)-\delta \ge \frac{3\delta^2}{6+2\delta}$ holds for all $\delta \ge 0$, see Theorem 2.3 and Lemma~2.4 in McDiarmid~\cite{McDiarmid98}. Equations~\eqref{eqprobCMUlin2} and~\eqref{eqprobCMUeasy} are trivial simplifications of~\eqref{eqprobCMUlin1}.

In general, to successfully use Chernoff bound in one's research, it greatly helps to look a little behind the formulas and understand their meaning. Very roughly speaking, we can distinguish three different regimes relative to $\delta$, namely that the tail probability is of order $\exp(-\Theta(\delta \log(\delta) E[X]))$, $\exp(-\Theta(\delta E[X]))$, and $\exp(-\Theta(\delta^2 E[X]))$. Here, in principle, the middle regime, referring to the case of $\delta$ constant, could be seen as a subcase of either of the two other regimes. Since this case of constant-factor deviations from the expectation occurs very frequently, we discuss it separately.

\paragraph{Superexponential regime:} Equation~\eqref{eqprobCMUstrong} shows a tail bound of order $\delta^{-\Theta(\delta E[X])} = \exp(-\Theta(\delta \log(\delta) E[X]))$, where the asymptotics are for $\delta \to \infty$. In this regime, the deviation $\delta E[X]$ from the expectation $E[X]$ is much larger than the expectation itself. It is not very often that we need to analyze such large deviations, so this Chernoff bound is rarely used. It can be useful in the analysis of evolutionary algorithms with larger offspring populations, where the most extreme behavior among the offspring can deviate significantly from the expected behavior. See~\cite{DoerrK15,DoerrD18} for examples how to use Chernoff bounds for large deviations occurring in the analysis of the \oplea. Note that in~\cite{JansenJW05}, the first theoretical work on the \oplea, and in~\cite{GiessenW17} such Chernoff bounds could have been used as well, but the authors found it easier to directly estimate the tail probability by estimating binomial coefficients.

Weaker forms of~\eqref{eqprobCMUstrong} are
  \begin{align}
  \Pr[X \ge (1+\delta) E[X]] &\le \bigg(\frac{e}{(1+\delta)}\bigg)^{(1+\delta)E[X]}\label{eqprobCMUstrongA},\\ 
  \Pr[X \ge (1+\delta) E[X]] &\le \bigg(\frac{e}{\delta}\bigg)^{\delta E[X]}\label{eqprobCMUstrongB},
  \end{align}
where the first one is stronger for those values of $\delta$ where the tail probability is less than one (that is, $\delta > e -1$).

It is not totally obvious how to find a value for $\delta$ ensuring that $(\frac{e}{\delta})^{\delta}$ is less than a desired bound. The following lemma solves this problem.

\begin{lemma}\label{lprobsuperexp}
  Let $t \ge e^{e^{1/e}} \approx 4.24044349...$ . Let $\delta = \frac{\ln t}{\ln(\frac{\ln t}{e \ln\ln t})}$. Then $(\frac{e}{\delta})^\delta \le \frac 1t$.
\end{lemma}

\begin{proof}
  We compute $\delta \ln \frac{\delta}{e} = \delta \ln\big(\frac{\ln t}{e \ln(\frac{\ln t}{e \ln\ln t})}\big) \ge \delta \ln(\frac{\ln t}{e\ln\ln t}) = \ln t$.
\end{proof}

We use this estimate to bound the number of bits flipped in an application of the standard-bit mutation operator\index{standard-bit mutation} defined  in Lemma~\ref{lprobhamming}. 
By linearity of expectation, it is clear that the expected Hamming distance $H(x,y)$ of parent $x$ and offspring $y$ is $E[H(x,y)] = \alpha$ when the mutation rate is $\frac \alpha n$, see Lemma~\ref{lprobhamming}. Using Chernoff bounds, we now give an upper bound on how far we can exceed this value. Such arguments are often useful in the analysis of evolutionary algorithms, see, e.g., Lemma~26 in~\cite{DoerrK15} for an example. 

\begin{lemma}\label{lprobSBM}
  \begin{enumerate}
	  \item Let $x \in \{0,1\}^n$ and $y$ be obtained from $x$ via standard-bit mutation\index{standard-bit mutation} with mutation rate $\frac \alpha n$. Then $\Pr[H(x,y) \ge k] \le (\frac{e \alpha}{k})^k$.
	  \item Let $0 < p \le \exp(-\alpha \exp(\frac 1e))$. Let $k \ge k_p := \frac{\ln(1/p)}{\ln\left(\frac{\ln(1/p^{1/\alpha})}{e \ln\ln(1/p^{1/\alpha})}\right)}$. Then $\Pr[H(x,y) \ge k] \le p$. 
	  \item Let $T \in \N$ and $0 < p \le \frac 1T \exp(-\alpha \exp(\frac 1e))$. Let $y_1, \dots, y_T$ be obtained from $x_1, \dots, x_T$, respectively, via standard-bit mutation. Let $k \ge \frac{\ln(T/p)}{\ln\left(\frac{\ln((T/p)^{1/\alpha})}{e \ln\ln((T/p)^{1/\alpha})}\right)}$. Then \[\Pr[\exists i \in [1..T] : H(x_i,y_i) \ge k] \le p.\] 
  \end{enumerate}
\end{lemma}

\begin{proof}
  Note that $H(x,y) \sim \Bin(n,\frac \alpha n)$, hence $H(x,y)$ can be written as a sum of $n$ independent random variables $X_1,\dots,X_n$ with $\Pr[X_i = 1] = \frac \alpha n$ and $\Pr[X_i=0] = 1 - \frac \alpha n$ for all $i \in [1..n]$. Since $E[H(x,y)]=\alpha$, we can apply equation~\eqref{eqprobCMUstrongA} with $(\delta+1) = \frac k \alpha$. This proves (a). 
  
  For part (b), we use part~(a) and Lemma~\ref{lprobsuperexp} and compute $\Pr[H(x,y) \ge k] \le \Pr[H(x,y) \ge k_p] \le ((\frac{e}{k_p/\alpha})^{k_p/\alpha})^\alpha \le (p^{1/\alpha})^\alpha = p$. Similarly, for (c) we obtain $\Pr[H(x_i,y_i) \ge k] \le \frac{p}{T}$ and use the union bound (Lemma~\ref{lprobunionbound}).
\end{proof}

Observe that the bounds in Lemma~\ref{lprobSBM} are independent of $n$. Also, the bounds in parts~(b) and~(c) depend only mildly on $\alpha$. By applying part (c) with $p=n^{-c_1}$ and $T=n^{c_2}$, we see that the probability that an evolutionary algorithm using standard-bit mutation with rate $\frac \alpha n$, where $\alpha$ is a constant, flips more than $(c_1+c_2+o(1))\frac{\ln n}{\ln\ln n}$ bits in any of the first $n^{c_2}$ applications of the mutation operator, is at most $n^{-c_1}$. 

We gave the results above to demonstrate the use of Chernoff bounds for sums of independent bounded random variables. Since the number of bits flipping in standard-bit mutation follows a binomial distribution, similar bounds can also (and by more elementary arguments) be obtained from analyzing the binomial distribution. See Lemma~\ref{lprobbino} for an example.

\paragraph{Exponential regime:} When $\delta = \Theta(1)$, then all bounds give a tail probability of order $\exp(-\Theta(\delta E[X]))$. Note that the difference between these bounds often is not very large. For $\delta=1$, the bounds in~\eqref{eqprobCMUstrong}, \eqref{eqprobCMUlin1}, and~\eqref{eqprobCMUlin2} become $(0.67957...)^{E[X]}$, $(0.68728...)^{E[X]}$, and $(0.71653...)^{E[X]}$, respectively. So there is often no reason to use the unwieldy equation~\eqref{eqprobCMUstrong}.

We remark that also for large $\delta$, where the bound~\eqref{eqprobCMUstrong} gives the better asymptotics $\exp(-\Theta(\delta \log(\delta) E[X]))$, one can, with the help of Section~\ref{secprobEE} resort to the easier-to-use bounds \eqref{eqprobCMUlin1} and~\eqref{eqprobCMUlin2} when the additional logarithmic term is not needed. For example, when $X$ is again the number of bits that flip in an application of the standard-bit mutation operator with mutation rate $p=\frac \alpha n$, then for all $c > 0$ and $n \in \N$ with $c \ln n \ge \alpha$ equation~\eqref{eqprobCMUlin2} with $E[X] \le \mu^+ := c \ln n$ and the argument of Section~\ref{secprobEE} gives $\Pr[X \ge 2c\ln n] = \Pr[X \ge (1+1) \mu^+] \le \exp(-\frac 13 \mu^+) = n^{-c/3}$, which in many applications is fully sufficient. 

A different way of stating an $\exp(-\Theta(\delta E[X]))$ tail bound, following directly from applying~\eqref{eqprobCMUstrongA} for $\delta \ge 2e-1$, is the following.

\begin{corollary}\label{corprob2hoch}
  Under the assumptions of Theorem~\ref{tprobCMU}, we have 
  \begin{equation}
  \Pr[X \ge k] \le 2^{-k} \label{eqprob2hoch}
  \end{equation} 
  for all $k \ge 2eE[X]$.
\end{corollary}

\paragraph{Sub-exponential regime:} Since Chernoff bounds give very low probabilities for the tail events, we can often work with $\delta = o(1)$ and still obtain sufficiently low probabilities for the deviations. Therefore, this regime occurs frequently in the analysis of randomized search heuristics. Since the tail probability is of order $\exp(-\Theta(\delta^2 E[X]))$, we need $\delta$ to be at least of order $(E[X])^{-1/2}$ to obtain useful statements. Note that for $E[X]$ close to $\frac n 2$, Theorem~\ref{tprobchernoffadditive01} below gives slightly stronger bounds. A typical application in this regime is showing that the random initial search points of an algorithms with high probability all have a Hamming distance of at least $\frac n2 (1-o(1))$ from the optimum. See Lemma~\ref{lprobinitial} below for further details.

\subsubsection{Multiplicative Chernoff Bounds for the Lower Tail}\label{secprobCML}

In principle, of course, there is no difference between bounds for the upper and lower tail. If in the situation of Theorem~\ref{tprobCMU} we set $Y_i := 1 - X_i$, then the $Y_i$ are independent random variables taking values in $[0,1]$ and any upper tail bound for $X$ turns into a lower tail bound for $Y := \sum_{i=1}^n Y_i$ via $\Pr[Y \le t] = \Pr[X \ge n - t]$. However, since this transformation also changes the expectation, that is, $E[Y] = n - E[X]$, a convenient bound like~\eqref{eqprobCMUeasy} becomes the cumbersome estimate $\Pr[Y \le (1-\delta)E[Y]] \le \exp(-\frac 13 (1+\delta\frac{E[Y]}{n-E[Y]})^2(n-E[Y]))$. 

For this reason, usually the tail bounds for the lower tail are either proven completely separately (however, using similar ideas) or are derived from significantly simplifying the results stemming from applying the above symmetry argument to~\eqref{eqprobCMUstrongest}. Either approach can be used to show the following bounds. As a visible result of the asymmetry of the situation for upper and lower bounds, note the better constant of $\frac 12$ in the exponent of~\eqref{eqprobCMLeasy} as compared to the $\tfrac 13$ in~\eqref{eqprobCMUeasy}.

\begin{theorem}\label{tprobCML}
  Let $X_1, \ldots, X_n$ be independent random variables taking values in $[0,1]$. Let $X = \sum_{i = 1}^n X_i$. Let $\delta \in [0,1]$. Then 
  \begin{align}
  \Pr[X \le (1-\delta) E[X]] 
  &\le \bigg(\frac{1}{1-\delta}\bigg)^{(1-\delta)E[X]} \bigg(\frac{n-E[X]}{n-(1-\delta)E[X]}\bigg)^{n-(1-\delta)E[X]}\label{eqprobCMLstrongest}\\ 
  &\le \bigg(\frac{e^{-\delta}}{(1-\delta)^{1-\delta}}\bigg)^{E[X]}\label{eqprobCMLstrong}\\ 
  &\le \exp\bigg(-\frac{\delta^2 E[X]}{2}\bigg),\label{eqprobCMLeasy}
  \end{align}
  where the first bound reads as $(1 - \frac{E[X]}{n})^n$ for $\delta = 1$.
\end{theorem}

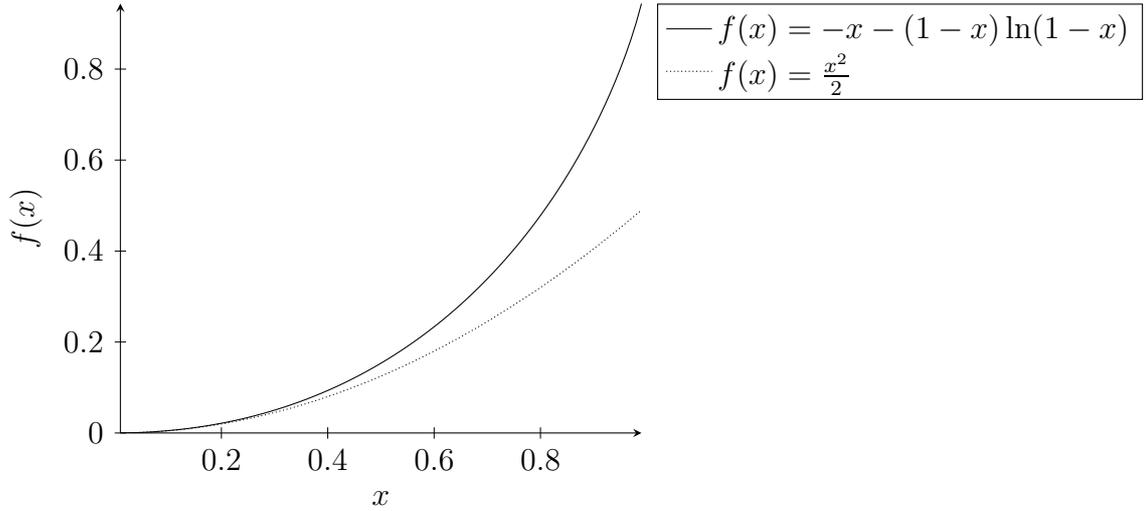
\begin{figure}
\begin{tikzpicture}

\begin{axis}[
    axis lines = left,
    xlabel = $x$,
    ylabel = {$f(x)$},
    legend style={
      cells={anchor=east},
      legend pos=outer north east,},
    ]

\addplot [
    domain=0.01:0.99, 
    samples=100, 
    color=black,
]
{+x+(1-x)*ln(1-x)};
\addlegendentry[right]{$f(x) = -x-(1-x)\ln(1-x)$}
 
\addplot [densely dotted,
    domain=0.01:0.99, 
    samples=100, 
    color=black,
    ]
    {x^2 / 2};
\addlegendentry[right]{$f(x) = \frac{x^2}{2}$}


\end{axis}
\end{tikzpicture}
\caption{Visual comparison of the bounds~\eqref{eqprobCMLstrong} and~\eqref{eqprobCMLeasy}. Depicted is the term $f(x)$ leading to the bound $\Pr[X \le (1-x) E[X]] \le \exp(-f(x) E[X])$.}\label{figprobCMlow}
\end{figure}

For the not so interesting boundary cases, recall our definition $0^0 := 1$. The first bound~\eqref{eqprobCMLstrongest} follows from~\eqref{eqprobCMUstrongest} by regarding the random variables $Y_i :=1-X_i$. Allowing the following easy derivation is maybe the main strength of~\eqref{eqprobCMUstrongest}. Setting $Y = \sum_{i=1}^n Y_i$ and $\delta' = \delta \frac{E[X]}{E[Y]}$, we compute
\begin{align*} 
  \Pr[X \le (1-\delta) E[X]] 
  & = \Pr[Y \ge (1+\delta') E[Y]] \\
  & \le \bigg(\frac{1}{1+\delta'}\bigg)^{(1+\delta')E[Y]} \bigg(\frac{n-E[Y]}{n-(1+\delta')E[Y]}\bigg)^{n-(1+\delta')E[Y]}\\ 
  & = \bigg(\frac{n-E[X]}{n-(1-\delta)E[X]}\bigg)^{n-(1-\delta)E[X]} \bigg(\frac{1}{1-\delta}\bigg)^{(1-\delta)E[X]}.
\end{align*}

Obviously, in an analoguous fashion, \eqref{eqprobCMUstrongest} can be derived from~\eqref{eqprobCMLstrongest}, so the two bounds are equivalent. Equation~\eqref{eqprobCMLstrong} follows from~\eqref{eqprobCMLstrongest} using an elementary estimate analoguous to~\eqref{eqprobeweg}. Equation~\eqref{eqprobCMLeasy} follows from~\eqref{eqprobCMLstrong} using elementary calculus, see, e.g., the proof of Theorem~4.5 in~\cite{MitzenmacherU05}. 

Theorems~\ref{tprobCMU} and~\ref{tprobCML} in particular show that constant-factor deviations from the expectation appear only with exponentially small probability.
\begin{corollary}\label{corprobCMboth}
  Let $X_1, \ldots, X_n$ be independent random variables taking values in $[0,1]$. Let $X = \sum_{i = 1}^n X_i$. Let $\delta \in [0,1]$. Then \[\Pr\big[|X - E[X]| \ge \delta E[X]\big] \le 2 \exp\bigg(-\frac{\delta^2 E[X]}{3}\bigg).\]
\end{corollary}

\subsubsection{Additive Chernoff Bounds}

We now present a few bounds for the probability that a random variable deviates from its expectation by an additive term independent of the expectation. The advantage of such bounds is that they are identical for upper and lower tails and that they are invariant under additive rescalings. 

Already from Theorem~\ref{tprobCMU}, equation~\eqref{eqprobCMUstrongest}, by careful estimates (see, e.g., Hoeffding~\cite{Hoeffding63}) and exploiting the obvious symmetry, we obtain the following estimates. As mentioned earlier, when $E[X]$ is close to $\frac n2$, this additive Chernoff bound gives (slightly) stronger results than the simplified bounds of Theorems~\ref{tprobCMU} and~\ref{tprobCML}.

\begin{theorem}\label{tprobchernoffadditive01}
  Let $X_1, \ldots, X_n$ be independent random variables taking values in $[0,1]$. Let $X = \sum_{i = 1}^n X_i$. Then for all $\lambda \ge 0$,  
  \begin{align}
    \Pr[X \ge E[X] + \lambda] & \le \exp\bigg(-\frac{2\lambda^2}{n}\bigg),\label{eqprobchernoffadditive01U}\\
    \Pr[X \le E[X] - \lambda] & \le \exp\bigg(-\frac{2\lambda^2}{n}\bigg).\label{eqprobchernoffadditive01L}
  \end{align}
\end{theorem}

A second advantage of additive Chernoff bounds is that they are often very easy to apply. As a typical application in evolutionary computation, let us regard the Hamming distance $H(x,x^*)$ of a random search point $x \in \{0,1\}$ from a given search point $x^*$. This could be, e.g., the distance of a random initial solution from the optimum.

\begin{lemma}\label{lprobinitial}
  Let $x^* \in \{0,1\}^n$. Let $x \in \{0,1\}^n$ be chosen uniformly at random. Then for all $\lambda \ge 0$, \[\Pr\left[\,\left|H(x,x^*) - \frac n2\right| \ge \lambda\right] \le 2 \exp\left(-\frac{2\lambda^2}{n}\right).\]
\end{lemma}

\begin{proof}
  Note that if $x \in \{0,1\}^n$ is uniformly distributed, then the $x_i$ are independent random variables uniformly distributed in $\{0,1\}$. Hence regardless of $x^*$, the indicator random variables $X_i$ for the event that $x_i \neq x_i^*$ are also independent random variables uniformly distributed in $\{0,1\}$. Since $H(x,x^*) = \sum_{i=1}^n X_i$, the claim follows immediately from applying Theorem~\ref{tprobchernoffadditive01} to the events ``$H(x,x^*) \ge E[H(x,x^*)] + \lambda$'' and ``$H(x,x^*) \le E[H(x,x^*)] - \lambda$''. 
\end{proof}

The lemma implies that even among a polynomial number of initial search points there is none which is closer to the optimum than $\frac n2 - O(\sqrt{n \log n})$. This arguments has been used numerous times in lower bound proofs. This argument is also the reason why the best-known black-box algorithm for the optimization of \onemax, namely repeatedly sampling random search points until the fitness values observed determine the optimum, also works well for jump functions~\cite{BuzdalovDK16}.

The following theorem, again due to Hoeffding~\cite{Hoeffding63}, non-trivially extends Theorem~\ref{tprobchernoffadditive01} by allowing the $X_i$ to take values in arbitrary intervals $[a_i,b_i]$. 

\begin{theorem}\label{tprobchernoffadditive}
  Let $X_1, \ldots, X_n$ be independent random variables. Assume that each $X_i$ takes values in a real interval $[a_i,b_i]$ of length $c_i := b_i - a_i$. Let $X = \sum_{i = 1}^n X_i$. Then for all $\lambda > 0$,  
  \begin{align}
    &\Pr[X \ge E[X] + \lambda]  \le  \exp\bigg(- \frac{2\lambda^2}{\sum_{i=1}^n c_i^2}\bigg),\label{eqprobCAU}\\
    &\Pr[X \le E[X] - \lambda]  \le  \exp\bigg(- \frac{2\lambda^2}{\sum_{i=1}^n c_i^2}\bigg).\label{eqprobCAL}
  \end{align}
\end{theorem}

For comparison, we now reformulate Theorems~\ref{tprobCMU} and~\ref{tprobCML} as additive bounds. There is no greater intellectual challenge hidden, but we feel that it helps to have seen these bounds at least once. Note that, since the resulting bounds depend on the expectation, we need that the $X_i$ take values in $[0,1]$. In other words, different from the bounds presented so far in this subsection, the following bounds are not invariant under additive rescaling and are not symmetric for upper and lower tails.

\begin{theorem}[equivalent to Theorem~\ref{tprobCMU}]\label{tprobCMUA}
  Let $X_1, \ldots, X_n$ be independent random variables taking values in $[0,1]$. Let $X = \sum_{i = 1}^n X_i$. Let $\lambda \ge 0$. Then 
  \begin{align}
  \Pr&(X \ge  E[X] + \lambda) \nonumber\\
  &\le \bigg(\frac{E[X]}{E[X]+\lambda}\bigg)^{E[X]+\lambda} \bigg(\frac{n-E[X]}{n-E[X]-\lambda}\bigg)^{n-E[X]-\lambda}\label{eqprobCMUAstrongest}\\ 
  &\le e^{\lambda} \bigg(\frac{E[X]}{E[X]+\lambda}\bigg)^{E[X]+\lambda} = \exp\bigg(-(E[X]+\lambda)\ln\bigg(1+\frac{\lambda}{E[X]}\bigg)+\lambda\bigg)\label{eqprobCMUAstrong}\\ 
  &\le \exp\bigg(- \frac{\lambda^2}{2 E[X] + \frac 23 \lambda}\bigg)\label{eqprobCMUAlin1}\\
  &\le \exp\bigg(- \frac 13 \min\bigg\{\frac{\lambda^2}{E[X]}, \lambda \bigg\}\bigg),\label{eqprobCMUAlin2}
  \end{align}
  where the bound in~\eqref{eqprobCMUAstrongest} is read as $0$ for $\lambda > n-E[X]$ and as $(\frac{E[X]}{n})^n$ for $\lambda = n-E[X]$. For $\lambda \le E[X]$, equation~\eqref{eqprobCMUAlin2} simplifies to 
  \begin{equation}
  \Pr[X \ge E[X] + \lambda] \le \exp\bigg(-\frac{\lambda^2}{3 E[X]}\bigg).\label{eqprobCMUAeasy}
  \end{equation}
\end{theorem}

\begin{theorem}[equivalent to Theorem~\ref{tprobCML}]\label{tprobCMLA}
  Let $X_1, \ldots, X_n$ be independent random variables taking values in $[0,1]$. Let $X = \sum_{i = 1}^n X_i$. Let $\lambda \ge 0$. Then 
  \begin{align}
  \Pr[X \le E[X] - \lambda] 
  &\le \bigg(\frac{E[X]}{E[X] - \lambda}\bigg)^{E[X]-\lambda} \bigg(\frac{n-E[X]}{n-E[X]+\lambda}\bigg)^{n-E[X]+\lambda} \label{eqprobCMLAstrongest}\\ 
  &\le e^{-\lambda} \bigg(\frac{E[X]}{E[X] - \lambda}\bigg)^{E[X]-\lambda} \label{eqprobCMLAstrong}\\ 
  &\le \exp\bigg(-\frac{\lambda^2}{2E[X]}\bigg).\label{eqprobCMLAeasy}
  \end{align}
\end{theorem}

\subsubsection{Chernoff Bounds Using the Variance}

There are several versions of Chernoff bounds that take into account the variance\index{variance}. In certain situations, they can give significantly stronger bounds than the estimates discussed so far. Hoeffding~\cite{Hoeffding63} proves essentially the following result.

\begin{theorem}\label{tprobcvar}
  Let $X_1, \ldots, X_n$ be independent random variables such that $X_i \le E[X_i] + 1$ for all $i = 1, \ldots, n$. Let $X = \sum_{i = 1}^n X_i$. Let $\sigma^2 = \sum_{i=1}^n \Var[X_i] = \Var[X]$. Then for all $\lambda \ge 0$, 
\begin{align}
	\Pr[&X \ge E[X] + \lambda]\nonumber\\ 
	&\le \bigg(\bigg(1 + \frac{\lambda}{\sigma^2}\bigg)^{-\left(1+\frac{\lambda}{\sigma^2}\right)\frac{\sigma^2}{n+\sigma^2}} \bigg(1-\frac{\lambda}{n}\bigg)^{-\left(1-\frac{\lambda}{n}\right)\frac{n}{n+\sigma^2}}\bigg)^n \label{eqprobcvarstrongest}\\
	\begin{split}
	&\le \exp\bigg(-\lambda\bigg(\bigg(1+\frac{\sigma^2}{\lambda}\bigg) \ln\bigg(1+\frac{\lambda}{\sigma^2}\bigg)-1\bigg)\bigg)\\ 
	&= \exp\bigg(-\sigma^2 \bigg(\bigg(1+\frac{\lambda}{\sigma^2}\bigg) \ln\bigg(1+\frac{\lambda}{\sigma^2}\bigg)-\frac{\lambda}{\sigma^2}\bigg)\bigg)
	\end{split}\label{eqprobcvarstrong}\\
	&\le \exp\bigg(-\frac{\lambda^2}{2 \sigma^2 + \frac 23 \lambda}\bigg)\label{eqprobcvarlin1}\\
	&\le \exp\bigg(-\frac 13 \min\bigg\{\frac{\lambda^2}{\sigma^2}, \lambda\bigg\}\bigg)\label{eqprobcvarlin2},
\end{align}
	where~\eqref{eqprobcvarstrongest} is understood to mean $0$ when $\lambda > n$ and  $(\frac{\sigma^2}{n+\sigma^2})^n$ when $\lambda = n$.
\end{theorem}

The estimate from~\eqref{eqprobcvarstrongest} to~\eqref{eqprobcvarstrong} is non-trivial and can be found, e.g., in Hoeffding~\cite{Hoeffding63}. From~\eqref{eqprobcvarstrong} we derive~\eqref{eqprobcvarlin1} in the same way that was used to derive~\eqref{eqprobCMUlin1} from~\eqref{eqprobCMUstrong}. 

By replacing $X_i$ with $-X_i$, we obtain the analoguous bounds for the lower tail.

\begin{corollary}\label{corprobcvar}
  If the condition $X_i \le E[X_i] + 1$ in Theorem~\ref{tprobcvar} is replaced by $X_i \ge E[X_i] - 1$, then $\Pr[X \le E[X] - \lambda]$ satisfies the estimates of~\eqref{eqprobcvarstrongest} to~\eqref{eqprobcvarlin2}.
\end{corollary}

As discussed in Hoeffding~\cite{Hoeffding63}, the bound~\eqref{eqprobcvarstrong} is the same as inequality (8b) in Bennett~\cite{Bennett62}, which is  stronger than the bound~\eqref{eqprobcvarlin1} due to Bernstein~\cite{Bernstein24} and the bound of $\exp(-\frac 12 \lambda \arcsinh(\frac{\lambda}{2\sigma^2}))$ due to Prokhorov~\cite{Prokhorov56}.

In comparison to the additive version of the usual Chernoff bounds for the upper tail (Theorem~\ref{tprobCMUA}), very roughly speaking, we see that the Chernoff bounds working with the variance allow to replace the expectation of $X$ by its variance. When the $X_i$ are binary random variables with $\Pr[X_i = 1]$ small, then $E[X] \approx \Var[X]$ and there is not much value in using Theorem~\ref{tprobcvar}. For this reason, Chernoff bounds taking into account the variance have not been used a lot in the theory of randomized search heuristics. They can, however, be convenient when we have random variables with $\Pr[X_i=1]$ close to $1$. 

For example, consider that a search point $y \in \{0,1\}^n$ is obtained from a given $x \in \{0,1\}^n$ via standard-bit mutation\index{standard-bit mutation} with mutation rate $p$. Assume for simplicity that we are interested in estimating the number of ones in $y$ (the same argument would hold for the Hamming distance of $y$ to some other search point $z \in \{0,1\}^n$, e.g., a unique optimum). Now the number of ones in $y$ is simply $X = \sum_{i=1}^n y_i$ and thus $X$ is a sum of independent binary random variables. However, different from, e.g., the situation in Lemma~\ref{lprobSBM}, the expectation of $X$ may be big. If $x_i=1$, then $E[y_i] = 1 -p$. Hence  if $x$ has many ones, then $E[Y]$ is large. However, since $\Var[y_i] = p (1-p)$ regardless of $x_i$,  the variance $\Var[X] = n p (1-p)$ is small (assuming that $p$ is small). Consequently, here the Chernoff bounds of this subsection give better estimates than, e.g., Theorem~\ref{tprobCMUA}. See, e.g.,~\cite{DoerrGWY17} for an example where this problem appeared in a recent research paper.  

When not too precise bounds are needed, looking separately at the number of zeros and ones of $x$ that flip (and bounding these via simple Chernoff bounds) is a way to circumvent the use of Chernoff bounds taking into account the variance. Several research works follow this approach despite the often more technical computations. 

Chernoff bounds using the variance can also be useful in ant colony algorithms and estimation of distribution algorithms, where again pheromone values or frequencies close to $0$ or $1$ can lead to a small variance. See~\cite{NeumannW09,Witt17} for examples.

The bounds of Theorem~\ref{tprobcvar} can be written in a multiplicative form, e.g.,
\begin{align}
\Pr[&X \ge (1+\delta) E[X]] \nonumber\\
&\le \bigg(\bigg(1 + \frac{\delta E[X]}{\sigma^2}\bigg)^{-\left(1+\frac{\delta E[X]}{\sigma^2}\right)\frac{\sigma^2}{n+\sigma^2}} \bigg(1-\frac{\delta E[X]}{n}\bigg)^{-\left(1-\frac{\delta E[X]}{n}\right)\frac{n}{n+\sigma^2}}\bigg)^n \label{eqprobcvarmult}\\
&\le \exp\bigg(- \frac{\delta^2 E[X]^2}{2 \sigma^2 + \frac 23 \delta E[X]}\bigg).\label{eqprobcvarmultlin}
\end{align}
This is useful when working with relative errors, however, it seems that unlike for some previous bounds (compare, e.g., \eqref{eqprobCMUstrong} and~\eqref{eqprobCMUAstrong}) the multiplicative forms are not much simpler here.

Obviously, the case that all $X_i$ satisfy $X_i \le E[X_i] + b$ for some number~$b$ (instead of $1$) can be reduced to the case $b=1$ by dividing all random variable by~$b$. For the reader's convenience, we here state the resulting Chernoff bounds.

\begin{theorem}[equivalent to Theorem~\ref{tprobcvar} and Corollary~\ref{corprobcvar}]\label{tprobcvara}
  Let $X_1, \ldots, X_n$ be independent random variables. Let $b$ be such that $X_i \le E[X_i]+b$ for all $i = 1, \ldots, n$. Let $X = \sum_{i = 1}^n X_i$. Let $\sigma^2 = \sum_{i=1}^n \Var[X_i] = \Var[X]$. Then for all $\lambda \ge 0$,
\begin{align}
	\Pr[X \ge E[X] + \lambda] &\le \bigg(\bigg(1 + \frac{b\lambda}{\sigma^2}\bigg)^{-\left(1+\frac{b\lambda}{\sigma^2}\right)\frac{\sigma^2}{nb^2+\sigma^2}} \bigg(1-\frac{\lambda}{nb}\bigg)^{-\left(1-\frac{\lambda}{nb}\right)\frac{nb^2}{nb^2+\sigma^2}}\bigg)^n\label{eqprobcvarstrongesta}\\
	&\le \exp\bigg(-\frac{\lambda}{b}\bigg(\bigg(1+\frac{\sigma^2}{b\lambda}\bigg) \ln\bigg(1+\frac{b\lambda}{\sigma^2}\bigg)-1\bigg)\bigg)\\
	&\le \exp\bigg(-\frac{\lambda^2}{\sigma^2(2+\frac 23 \frac{b \lambda}{\sigma^2})}\bigg)\label{eqprobcvarlin1a}\\
	&\le \exp\bigg(-\frac 13 \min\bigg\{\frac{\lambda^2}{\sigma^2}, \frac \lambda b\bigg\}\bigg)\label{eqprobcvarlin2a},
\end{align}
	where~\eqref{eqprobcvarstrongesta} is understood to mean $0$ when $\lambda > nb$ and  $(\frac{\sigma^2}{nb^2+\sigma^2})^n$ when $\lambda = nb$.
	 
	When we have $X_i \ge E[X_i]-b$ instead of $X_i \le E[X_i]+b$ for all $i = 1, \ldots, n$, then the above estimates hold for $\Pr[X \le E[X] - \lambda]$.
\end{theorem}

\subsubsection{Relation Between the Different Chernoff Bounds}\label{secprobrelation}

We proceed by discussing how the bounds presented so far are related. The main finding will be that the Chernoff bounds depending on the variance imply all other bounds discussed so far with the exception of the additive Chernoff bound for random variables having different ranges (Theorem~\ref{tprobchernoffadditive}).

Surprisingly, this fact is not stated in Hoeffding's paper~\cite{Hoeffding63}. More precisely, in~\cite{Hoeffding63} the analogue of Theorems~\ref{tprobcvar} and~\ref{tprobcvara} uses the additional assumption that all $X_i$ have the same expectation. Since this assumption is not made for the theorems not involving the variance, Hoeffding explicitly states that the latter are stronger in this respect (see the penultimate paragraph of Section~3 of~\cite{Hoeffding63}). 

It is however quite obvious that the common-expectation assumption can be easily removed. From random variables with arbitrary means we can obtain random variables all having mean zero by subtracting their expectation. This operation does not change the variance and does not change the distribution of $X - E[X]$. Consequently, Hoeffding's result for variables with identical expectations immediately yields our version of this result (Theorem~\ref{tprobcvar} and~\ref{tprobcvara}). Theorem~\ref{tprobcvar} implies Theorem~\ref{tprobCMU} via the equivalent version of Theorem~\ref{tprobCMUA}, see again the penultimate paragraph of Section~3 of~\cite{Hoeffding63}. 

Consequently, the first (strongest) bound in Theorem~\ref{tprobcvar} (equivalently the first bound of Theorem~\ref{tprobcvara}) implies the first (strongest) bound of Theorem~\ref{tprobCMU}, which is equivalent to the first (strongest) bound in Theorem~\ref{tprobCML}. Essentially all other bounds presented so far can be derived from these main theorems via simple, sometimes tedious, estimates. The sole exception is Theorem~\ref{tprobchernoffadditive}, which can lead to significantly stronger estimates when the random variables have ranges of different size. 

As an example, let $X_1, \dots, X_n$ be independent random variables such that $X_1, \dots, X_{n-1}$ take the values $0$ and $(n-1)^{-1/2}$ with equal probability $\frac 12$ and such that $X_n$ takes the values $0$ and $1$ with equal probability $\frac 12$. Let $X = \sum_{i=1}^n X_i$. Then $E[X] = \tfrac 12 (\sqrt{n-1} + 1)$. Theorem~\ref{tprobchernoffadditive}, taking $c_i = (n-1)^{-1/2}$ for $i \in [1..n-1]$ and $c_n = 1$, yields the estimate 
\begin{equation}
 \Pr[X \ge E[X] + \lambda] \le \exp\bigg(-\frac{2\lambda^2}{\sum_{i=1}^n c_i^2}\bigg) = \exp(-\lambda^2).\label{eqprobyy}
\end{equation}

Note that $\Var[X] = \frac 12 =: \sigma^2$. Consequently, the strongest Chernoff bound of Theorem~\ref{tprobcvar}, equation~\eqref{eqprobcvarstrongest}, gives an estimate larger than $(1+\frac{\lambda}{\sigma^2})^{-(1+\frac{\lambda}{\sigma^2})\frac{n\sigma^2}{n+\sigma}} = \exp(-\Theta(\lambda \log \lambda))$. Consequently, in this case Theorem~\ref{tprobchernoffadditive} gives a significantly stronger estimate than Theorem~\ref{tprobcvar}.

\subsubsection{Tightness of Chernoff Bounds, Lower Bounds for Deviations (Anti-Concentration)}

As a very general and not at all precise rule of thumb, we can say that often the sharpest Chernoff bounds presented so far give an estimate for the tail probability that is near-tight. This is good to know from the perspective of proof design, since it indicates that failing to prove a desired statement usually cannot be overcome by trying to invent sharper Chernoff bounds. We shall not try to make this statement precise. 

However, occasionally, we also need lower bounds for the deviation from the expectation as a crucial argument in our analysis. For example, when generating several offspring independently in parallel, as, e.g., in a \oplea, we expect the best of these to be significantly better than the expectation and the efficiency of the algorithm relies on such desired deviations from the expectation. 

Lower bounds for deviations from the expectation, occasionally called anti-concentration results, seem to be harder to work with. For this reason, we only brief{}ly give some indications how to handle them and refer the reader to the literature. We note that there is a substantial body of mathematics literature on this topic, see, e.g.,~\cite{Nagaev01} and the references therein, which however is not always easy to use for algorithmic problems. 
We also note that for binomially distributed random variables, also the estimates in Theorem~\ref{tprobpbino} can be used to derive lower bounds for tail probabilities.

\paragraph{Estimating binomial coefficients:} For binomial distributions, estimating the (weighted) sum of binomial coefficients arising in the expression of the tail probability often works well (though the calculations may become tedious). In the theory of randomized search heuristics, this approach was used, among others, in the analysis of the \oplea in~\cite{JansenJW05,DoerrK15,GiessenW17,DoerrGWY17,DoerrWY18} and the \opllga in~\cite{DoerrD18}. The following elementary bound was shown in~\cite[Lemma~3]{Doerr14}. 

\begin{lemma}\label{lprobsqrtn12}
Let $n \in \N$ and $X \sim \Bin(n,\frac12)$. Then 
\begin{align}
&\Pr\Big[X \ge E[X] + \tfrac 12 \sqrt{E[X]}\Big] \ge \tfrac 18,\\
&\Pr\Big[X \le E[X] - \tfrac 12 \sqrt{E[X]}\Big] \ge \tfrac 18.
\end{align}
\end{lemma}

\paragraph{Two-stage rounding trick:} Estimating binomial coefficients works well for binomial distributions. However, a neat trick allows to extend such results to sums of independent, non-identically distributed binary random variables. The rough idea is that we can sample a binary random variable $X$ with $\Pr[X = 1] = p$ by first sampling the unique random variable $Y$ which takes values in $\{\frac 12, \lfloor p+\frac 12\rfloor\}$ and satisfies $E[Y] = E[X] = p$, and then, if $Y=\frac 12$, replacing $Y$ with a uniform choice in $\{0,1\}$. If we view sampling $X$ as rounding $p$ randomly to $0$ or $1$ in a way that the expectation is $p$, then this two-stage procedure consists of first rounding $p$ to $\{0,\frac 12\}$ or $\{\frac 12, 1\}$ with expectation $p$ and then (if necessary) rounding the result to $\{0,1\}$ without changing the expectation.

We use this trick below to show by elementary means two results which previously were shown only via deeper methods. We first extend Lemma~\ref{lprobsqrtn12} above from fair coin flips to sums of independent binary random variables having different distributions. A similar result was shown in~\cite[first item of Lemma~6]{OlivetoW15} for $X \sim \Bin(n,p)$, that is, for sums of identically distributed binary random variables (the result is stated without a lower bound on the variance, but by regarding, e.g., $\Bin(n,n^{-2})$, it becomes clear that a restriction like $p \in [\frac 1n, 1-\frac 1n]$ is necessary). We did not find the general result of Lemma~\ref{lprobsqrtn} in the literature, even though it is clear that such results can be shown via a normal approximation.

\begin{lemma}\label{lprobsqrtn}
  Let $v_0 > 0$. There are constants $c, C > 0$ such that the following is true. Let $n \in \N$. Let $p_1, \dots, p_n \in [0,1]$. For all $i \in [1..n]$, let $X_i$ be a binary random variable with $\Pr[X_i = 1] = p_i$. Assume that $X_1, \dots, X_n$ are independent. Let $X = \sum_{i=1}^n X_i$. Assume that $\Var[X] = \sum_{i=1}^n p_i (1-p_i) \ge v_0$. Then 
  \begin{align}
  \Pr\Big[X \ge E[X] + c\sqrt{\Var[X]}\Big] &\ge C,\label{eqprobsqrtnupper}\\
  \Pr\Big[X \le E[X] - c\sqrt{\Var[X]}\Big] &\ge C.\label{eqprobsqrtnlower}
  \end{align}
\end{lemma}

\begin{proof}
  Let us first assume that $p_i \le \frac 12$ for all $i \in [1..n]$ and show the claim under the weaker assumption $\Var[X] \ge \frac 12 v_0$. We define independent random variables $Y_i$ such that  
  \begin{align*}
  \Pr[Y_i = \tfrac 12] &= 2p_i,\\
  \Pr[Y_i = 0] &= 1 - 2p_i.
  \end{align*}
  Let $Y = \sum_{i = 1}^n Y_i$ and note that $E[Y] = E[X]$.
  
  Based on the $Y_i$, we define independent binary random variables $Z_i$ as follows. If $Y_i = 0$, then $Z_i := 0$. Otherwise, that is, if $Y_i = \frac 12$, then we let $Z_i$ be uniformly distributed in $\{0,1\}$. An elementary calculation shows that $\Pr[Z_i = 1] = p_i$, that is, the $Z_i$ have the same distribution as the $X_i$. Hence it suffices to show our claim for $Z := \sum_{i = 1}^n Z_i$.
  
  Let $c$ be a sufficiently small constant.  
  Our main argument for the lower bound on the upper tail~\eqref{eqprobsqrtnupper} shall be that with constant probability we have the event 
  \[A := \mbox{``$Y \ge E[Y] - \frac 12 c \sqrt{\Var[X]}$''}.\] 
  In this case, again with constant probability, we have $Z \ge E[Z \mid A] + c \sqrt{\Var[X]}$, which implies $Z \ge E[Y] - \frac 12 c \sqrt{\Var[X]} + c \sqrt{\Var[X]} = E[X] + \frac 12 c \sqrt{\Var[X]}$. In other words, we have 
  \begin{align*}
  \Pr\left[X \ge E[X] + \tfrac 12 c \sqrt{\Var[X]}\right]
  \ge \Pr[A] \cdot \Pr\left[(Z \mid A) \ge E[Z \mid A] + c \sqrt{\Var[X]} \right]
  \end{align*}
  and we shall argue that both factors are at least constant.

  For the first factor, we note that for all $i \in [1..n]$, we have $E[Y_i] = p_i$. An elementary calculation thus shows $\Var[Y_i] = \frac 12 p_i (1 - 2p_i) \le \frac 12 p_i (1 - p_i) = \frac 12 \Var[X_i]$ and hence $\Var[Y] \le \frac 12 \Var[X]$. 
  With Cantelli's inequality (Lemma~\ref{lprobcantelli}), we compute
  \begin{align*}
  \Pr[A] &= \Pr\left[Y \ge E[Y] - \frac 12 c \sqrt{\Var[X]}\right] \\
  &\ge \Pr\left[Y \ge E[Y] - \frac 1{\sqrt{2}}  c \sqrt{\Var[Y]}\right] \\
  &\ge 1 - \Pr\left[Y \le E[Y] - \frac 1{\sqrt{2}}  c \sqrt{\Var[Y]}\right] \\
  &\ge 1 - \frac{1}{1 + c^2/2} = \frac{c^2}{2+c^2}.
  \end{align*}
  
  For the second factor, we note that once $Y$ is determined, $Z \sim \Bin(2Y,\frac 12)$. We estimate
  \begin{align*}
  E[Y] - \tfrac 12 c \sqrt{\Var[X]}
  &\ge E[Y] - \tfrac c{\sqrt{2 v_0}} \Var[X]\\
  &\ge E[Y] - \tfrac c{\sqrt{2 v_0}} E[X]
  \ge (1 - \tfrac c{\sqrt{2 v_0}}) E[X] =: q,
  \end{align*}
  where we use that $\Var[X] \ge \frac 12 v_0$ implies $\sqrt{\Var[X]} \le \sqrt{2/v_0} \Var[X]$. Hence conditional on $A$, we have $Z \sim \Bin(2 \tilde q,\frac 12)$ for some $\tilde q \ge q$, and thus
  \begin{align*}
  \Pr\left[(Z\mid A) \ge E[Z \mid A] + c \sqrt{\Var[X]}\right] \ge \frac 18
  \end{align*}
  by Lemma~\ref{lprobsqrtn12} and $c \sqrt{\Var[X]} \le c \sqrt{E[X]} \le \frac 12 \sqrt{(1 - \tfrac c{\sqrt{2 v_0}}) E[X]} \le \frac 12 \sqrt{E[Z\mid A]}$, where the middle inequality assumes that $c$ is sufficiently small. 
  
  To prove~\eqref{eqprobsqrtnlower}, we argue as follows. Let $K = E[X] + \frac 12 c \sqrt{\Var[X]}$. Then
  \begin{align*}
  \Pr\Big[X \le E&[X] - \tfrac 12 c \sqrt{\Var[X]}\Big] \\
  & \ge \sum_{k=0}^{2K} \Pr[Y = \tfrac k2] \cdot \Pr\Big[(Z \mid Y=\tfrac k2) \le E[X] - \tfrac 12 c \sqrt{\Var[X]}\Big]
  \end{align*}
  Now $(Z \mid Y = \frac k2) \sim \Bin(k,\frac 12)$, hence Lemma~\ref{lprobdomdistr}~\ref{it:probdomnm} implies that the second factor is smallest for $k = 2K$. Consequently, 
  \[\Pr\Big[X \le E[X] - \tfrac 12 c \sqrt{\Var[X]}\Big] \ge \Pr[Y \le K] \cdot \Pr\Big[(Z \mid Y=K) \le E[X] - \tfrac 12 c \sqrt{\Var[X]}\Big].\]
  We estimate the two factors separately. For the first one, in an analogous fashion as before, we obtain $\Pr[Y \le K] = \Pr[Y \le E[Y] + \frac 12 c \sqrt{\Var[X]}] \ge \frac{c^2}{2 + c^2}$. For the second factor, we compute
  \begin{align*}
  \Pr\Big[&(Z \mid Y=K) \le E[X] - \tfrac 12 c \sqrt{\Var[X]}\Big] \\
  &= \Pr\Big[\Bin(2K,\tfrac 12) \le E[\Bin(2K,\tfrac 12)] - K + E[X] - \tfrac 12 c \sqrt{\Var[X]}\Big] \\
  &= \Pr\Big[\Bin(2K,\tfrac 12) \le E[\Bin(2K,\tfrac 12)] - c \sqrt{\Var[X]}\Big].
  \end{align*}
  For $c \le \frac 12$, we have $c \sqrt{\Var[X]} \le \frac 12 \sqrt{E[X]} \le \frac 12 \sqrt{E[\Bin(2K,\tfrac 12)]}$ and Lemma~\ref{lprobsqrtn12} yields $\Pr\Big[(Z \mid Y=K) \le E[X] - \tfrac 12 c \sqrt{\Var[X]}\Big] \ge \frac 18$.
 
 Now assume that the $p_i$ are not all in $[0,\frac 12]$. Let $I' = \{i \in [1..n] \mid p_i \le \frac 12\}$ and $I'' = [1..n] \setminus I'$. Let $X' = \sum_{i \in I'} X_i$ and $X'' = \sum_{i \in I''} X_i$. Since $\Var[X] = \Var[X'] + \Var[X'']$, by symmetry (possibly replacing the $p_i$ by $1-p_i$), we can assume that $\Var[X'] \ge \frac 12 \Var[X]$. Now $\Var[X] \ge v_0$ implies $\Var[X'] \ge \frac 12 v_0$, and by the above we have $X' \ge E[X'] + \frac 12 c \sqrt{\Var[X']} \ge E[X'] + \frac 1{2\sqrt 2} c \sqrt{\Var[X]}$ with constant probability. By Cantelli's inequality again, we have $X'' \ge E[X''] - \frac{c}{4} \sqrt{\Var[X'']} \ge E[X''] - \frac{c}{4\sqrt{2}} \sqrt{\Var[X]}$ with constant probability. Hence $X = X' + X'' \ge E[X'] + E[X''] + \frac c{4\sqrt 2} \Var[X] = E[X] + \frac c{4 \sqrt 2}  \Var[X]$ with constant probability. The proof that $X \le E[X] - \frac c{4 \sqrt 2}  \Var[X]$ with constant probability is analogous. By replacing our original $c$ by $4\sqrt{2}c$, we obtain the precise formulation of the claim.
\end{proof}

We now use the two-stage rounding trick to give an elementary proof of the following result. 
\begin{lemma}\label{lprobhit}
  Let $n \in \N$ and $p_1, \dots, p_n \in [0,1]$. For all $i \in [1..n]$, let $X_i$ be a binary random variable with $\Pr[X_i = 1] = p_i$. Assume that $X_1, \dots, X_n$ are independent. Let $X = \sum_{i = 1}^n X_i$. If $\Var[X] \ge 1$, then for all $k \in [0..n]$, 
  \[\Pr[X = k] \le \frac{2}{\sqrt{\Var[X]}}.\]
\end{lemma}

This result (without making the leading constant precise) was proven in the special case that all $p_i$ are between $\frac 16$ and $\frac 56$ in~\cite[Lemma~9, arxiv version]{SudholtW16}. This proof uses several deep arguments from probability theory. In~\cite[Lemma~3]{KrejcaW17}, the result stated in~\cite{SudholtW16} was minimally extended to the case that only a linear number of the $p_i$ are between $\frac 16$ and $\frac56$.

\begin{proof}[Proof of Lemma~\ref{lprobhit}]
  In a similar fashion as in Lemma~\ref{lprobsqrtn}, we define independent random variables $Y_i$ such that  
  \begin{align*}
  \Pr[Y_i = \tfrac 12] &= 2p_i,\\
  \Pr[Y_i = 0] &= 1 - 2p_i
  \end{align*}
  when $p_i \le \frac 12$ and 
  \begin{align*}
  &\Pr[Y_i = 1] = 2(p_i-\tfrac 12) = 2p_i - 1,\\
  &\Pr[Y_i = \tfrac 12] = 1 - 2(p_i-\tfrac 12) = 2 - 2p_i
  \end{align*}
  for $p_i > \tfrac 12$. If $Y_i \in \{0,1\}$, then $Z_i := Y_i$, else (that is, when $Y_i = \frac 12$) we let $Z_i$ be uniformly distributed on $\{0,1\}$. As before, the $Z_i$ are just an alternative definition of the $X_i$. Hence $Z = \sum_{i=1}^n Z_i$ has the same distribution as $X$. 
  
  For $\ell \in \{0,\frac 12, 1\}$ denote by $I_{\ell} := |\{i \in [1..n] \mid Y_i = \ell\}|$. Since 
  \[\Pr[Y_i = \tfrac 12] = 2 \min\{p_i,1-p_i\} \ge 2 \Var[X_i],\] 
  we have $E[I_{\frac 12}] \ge 2\Var[X]$. Since the $Y_i$ are independent, we have $\Pr[I_{\frac 12} \le 2(1-\delta)\Var[X]] \le \exp(-\delta^2 \Var[X])$ for all $\delta \in [0,1]$ by~\eqref{eqprobCMLeasy}. 
  
  Finally, note that by~\eqref{eqprobmiddlesharp} we have $\Pr[\Bin(a,\frac 12) = k] \le \sqrt{\frac{2}{\pi a}}$ for all $a \in \N$ and $k \in \N_0$.
  
  Writing $a_0 = \lfloor 2(1-\delta)\Var[X] \rfloor$ and combining these arguments, we obtain
  \begin{align*}
  \Pr[X = k] = & \Pr[Z = k]\\
  = & \sum_{a=0}^n \Pr[I_{\frac 12} = a] \sum_{b = 0}^k \Pr[I_1 = b] \Pr[Z = k \mid I_{\frac 12} = a \wedge I_1 = b]\\
  = & \sum_{a=0}^n \Pr[I_{\frac 12} = a] \sum_{b = 0}^k \Pr[I_1 = b] \Pr[\Bin(a,\tfrac 12) = k-b]\\
  \le & \Pr[I_{\frac 12} \le a_0] \\
  &+ \sum_{a=a_0+1}^n \Pr[I_{\frac 12} = a] \sum_{b = 0}^k \Pr[I_1 = b] \Pr[\Bin(a,\tfrac 12) = k-b]\\
  \le & \exp(-\delta^2 \Var[X]) + \sum_{a=a_0+1}^n \Pr[I_{\frac 12} = a] \sum_{b = 0}^k \Pr[I_1 = b] \sqrt{\frac{2}{\pi (a_0+1)}}\\
  \le & \exp(-\delta^2 \Var[X]) +  \sqrt{\frac{1}{(1-\delta) \pi \Var[X]}}.
  \end{align*}
For $\Var[X] \ge 1$, by taking $\delta = 0.75$ and estimating $\exp(-\delta^2 \Var[X]) \le \frac{1}{e\delta^2 \Var[X]} \le \frac{1}{e\delta^2 \sqrt{\Var[X]}}$, where we used the estimate $e^x \ge ex$, an alternative version of Lemma~\ref{lprobelower}, we obtain the bound $2 \Var[X]^{-1/2}$.
\end{proof}

We did not aim at optimizing the implicit constants in the result above. We note that by taking $\delta = \Var[X]^{-1/4}$, the claimed probability becomes $(1+o(1)) \frac{1}{\sqrt{\pi \Var[X]}}$ for  $\Var[X] \to \infty$.

\paragraph{Approximation via the normal distribution:} The generic approach of approximating binomial distributions via normal distributions is not often used in the theory of randomized search. In~\cite{OlivetoW15}, the Berry-Esseen inequality was employed to prove a result similar to Lemma~\ref{lprobsqrtn} for the special case of binomial distributions. Unlike many other proof relying on the normal approximation, this proof is quite short and elegant.

In~\cite{LaillevaultDD15}, the normal approximation was used to show that the best of $k \in \omega(1) \cap o(\sqrt n)$ independent random initial search points in $\{0,1\}^n$ with probability $1 - o(1)$ has a distance of $\frac n2 - \sqrt{\frac n2 (\ln k - \frac 12 \ln\ln k \pm c_k)}$ from the optimum, where $c_k$ is an arbitrary sequence tending to infinity. 

In~\cite[Lemma~7, arxiv version]{SudholtW16}, a very general result on how a sum of independent random variables with bounded expectation and variance is approximated by a normal distribution was used to analyze the performance of an estimation-of-distribution algorithm. This analysis is highly technical.

\paragraph{Order statistics:} The result about the best of $k$ independent initial individuals in~\cite{LaillevaultDD15} actually says something about the maximum order statistics of $k$ independent $\Bin(n,\frac 12)$ random variables. In general, the maximum order statistics is strongly related to lower bounds for tail probabilities as the following elementary argument (more or less explicit in all works on the \oplea) shows: Let $X_1, \dots, X_\lambda$ be independent random variables following the same distribution. Let $X_{\max} = \max\{X_i \mid i \in [1..\lambda]\}$. Then 
\begin{align*}
\Pr[X^* \ge D] &\le \lambda \Pr[X_1 \ge D];\\
\Pr[X^* \ge D] &= 1 - (1 - \Pr[X_1 \ge D])^\lambda \ge 1-\exp(\lambda \Pr[X_1 \ge D]).
\end{align*}
Consequently, $\Pr[X^* \ge D]$ is constant if and only if $\Pr[X_1 \ge D] = \Theta(\frac 1 \lambda)$. 

For the maximum order statistics of binomially distributed random variables with small success probability, Gie\ss en and Witt~\cite[Lemma~4~(3)]{GiessenW17} proved the following result and used it in the analysis of the \oplea.
\begin{lemma}
  Let $\alpha \ge 0$ and $c > 0$ be constants. Let $n \in \N$ and let all of the following asymptotics be for $n \to \infty$. Let $k = n (\ln n)^{-\alpha}$ and $\lambda = \omega(1)$. Let $X_{\max}$ be the maximum of $\lambda$ independent random variables with distribution $\Bin(k,\frac cn)$. Then $E[X_{\max}] = (1 \pm o(1)) \frac{1}{1+\alpha} \frac{\ln \lambda}{\ln\ln \lambda}$.
\end{lemma}  

\paragraph{Extremal situations:} Occasionally, it is desirable to understand which situation gives the smallest or the largest deviations. For example, let $X_1, \dots, X_n$ be independent binary random variables with expectations $E[X_i] = p_i$. Then it could be useful to know that $X = \sum_{i=1}^n X_i$ deviates most (in some suitable sense) from its expectation when all $p_i$ are $\frac 12$. Such statements can be made formal and can be proven with the notions of majorization and Schur-convexity. We refer to~\cite{Scheideler00} for a nice treatment of this topic. Such arguments have been used to analyze estimation-of-distribution algorithms in~\cite{SudholtW16}.

\paragraph{Staying on one side of the expectation and Feige's inequality:} When it suffices to know that with reasonable probability we stay (more or less) on one side of the expectation, then the following results can be useful. 

A very general bound is Feige's inequality~\cite[Theorem~1]{Feige06}, which has found applications in the analysis of randomized search heuristics, among others, in~\cite{SudholtT12,DangL15,LehreN17,CorusDEL18}.
\begin{lemma}[Feige's inequality]
  Let $X_1, \dots, X_n$ be independent non-negative random variables with expectations $\mu_i := E[X_i]$ satisfying $\mu_i \le 1$. Let $X = \sum_{i=1}^n X_i$. Then \[\Pr[X \le E[X] + \delta] \ge \min\{\tfrac{1}{13},\tfrac{\delta}{\delta+1}\}.\]
\end{lemma}

For binomial distributions, we have stronger guarantees. Besides bounds comparing the binomial distribution with its normal approximation~\cite{Slud77}, the following specific bounds are known.

\begin{samepage}
\begin{lemma}\label{lprobfeigebin}
  Let $n \in \N$, $p \in [0,1]$, and $k = \lfloor np \rfloor$. Let $X \sim \Bin(n,p)$. 
  \begin{enumerate}
  \item\label{itprobGM} If $\frac 1n < p$, then $\Pr[X \ge E[X]] > \tfrac 14$.
  \item\label{itprobich1}  If $0.29/n \le p < 1$, then $\Pr[X > E[X]] \ge \tfrac 14$.
  \item\label{itprobPR} If $\frac 1n \le p \le 1-\frac 1n$, then $\Pr[X \ge E[X]] \ge \frac{1}{2\sqrt 2} \frac{\sqrt{np(1-p)}}{\sqrt{np(1-p)+1}+1}$.
  \item\label{itprobich} If $\frac 1n \le p < 1$, then $\Pr[X > E[X]] > \frac 12 - \sqrt{\frac{n}{2\pi k (n - k)}}$.
  \item\label{itprobich2}  If $\frac1n \le p < 1-\frac1n$, then $\Pr[X > E[X]+1] \ge 0.037$.
  \end{enumerate}
\end{lemma}
\end{samepage}

Surprisingly, all these results are quite recent. Bound~\ref{itprobGM} from~\cite{GreenbergM14} appears to be the first general result of this type at all.\footnote{For $p \in [\frac 1n,\frac12]$, this result follows from the proof of Lemma~6.4 in~\cite{RigolletT11}. The lemma itself only states the bound $\Pr[X \ge E[X]] \le \min\{p,\frac 14\}$ for $p \le \frac12$. The assumption $p \le \frac 12$ appears to be crucial for the proof.} It was followed up by estimate~\ref{itprobPR} from~\cite{PelekisR16}, which gives stronger estimates when $np(1-p) > 8$. Result~\ref{itprobich} from~\cite{Doerr18exceedexp} is the only one to give a bound tending to $\frac 12$ for both $np$ and $n(1-p)$ tending to infinity. Estimates~\ref{itprobich1} and~\ref{itprobich2} are also from~\cite{Doerr18exceedexp}. A lower bound for the probability of exceeding the expectation by more than one, like~\ref{itprobich2}, was needed in the analysis of an evolutionary algorithm with self-adjusting mutation rate (proof of Lemma~3 of the extended version of~\cite{DoerrGWY17}).

\subsubsection{Proofs of the Chernoff Bounds}\label{secprobcproof}

As discussed in Section~\ref{secprobrelation}, all Chernoff bounds stated so far can be derived from the strongest bounds of Theorem~\ref{tprobchernoffadditive} or~\ref{tprobcvar} via elementary estimates that have nothing to do with probability theory. We shall not detail these estimates---the reader can find them all in the literature, e.g., in~\cite{Hoeffding63}. We shall, however, sketch how to prove these two central inequalities~\eqref{eqprobCAU} and~\eqref{eqprobcvarstrongest}. One reason for this is that we can then argue in Section~\ref{secprobnegcor} that these proofs (and thus also all Chernoff bounds presented so far) not only hold for independent random variables, but also negatively correlated ones. 

A second reason is that occasionally, it can be profitable to have this central argument ready to prove Chernoff bounds for particular distributions, for which the classic bounds are not applicable or do not give sufficient results. This has been done, e.g., in~\cite{DoerrFW11,OlivetoW11,OlivetoW12,DoerrJWZ13,LehreW14,Witt14,BadkobehLS14}.

The central step in almost all proofs of Chernoff bounds, going back to Bernstein~\cite{Bernstein24}, is the following one-line argument. Let $h > 0$. Then
\begin{align}
  \Pr[X \ge t] & = \Pr[e^{hX} \ge e^{ht}] \le \frac{E[e^{hX}]}{e^{ht}} = e^{-ht} \prod_{i=1}^n E[e^{hX_i}].
\end{align}
Here the first equality simply stems from the fact that the function $x \mapsto e^{hx}$ is monotonically increasing. The inequality is  Markov's inequality (Lemma~\ref{lprobmarkov}) applied to the (non-negative) random variable $e^{hX}$. The last equality exploits the independence of the $X_i$, which carries over to the $e^{hX_i}$. 

It now remains to estimate $E[e^{hX_i}]$ and choose $h$ as to minimize the resulting expression. We do this exemplarily for the case that all $X_i$ take values in $[0,1]$ and that $E[X] < t < n$. Since the exponential function is convex, $E[e^{hX_i}]$ is maximized (which is the worst-case for our estimate) when $X_i$ is concentrated on the values $0$ and $1$, that is, we have $\Pr[X_i = 1] = E[X_i]$ and $\Pr[X_i=0] = 1 - E[X_i]$. In this case, $E[e^{hX_i}] = (1-E[X_i])e^0 + E[X_i] e^h$. By the inequality of arithmetic and geometric means, we compute 
\begin{align*}
\prod_{i=1}^n E[e^{hX_i}] &\le \prod_{i=1}^n (1 - E[X_i] + E[X_i] e^h) \\
&\le \bigg(\frac 1n \sum_{i=1}^n (1 - E[X_i] + E[X_i] e^h)\bigg)^n \\
&\le \bigg(\frac 1n (n  - E[X] + E[X] e^h)\bigg)^n.
\end{align*}
This gives the tail estimate $\Pr[X \ge t] \le e^{-ht} (\frac 1n (n  - E[X] + E[X] e^h))^n$, which is minimized by taking \[h = \ln\bigg(\frac{(n-E[X])t}{(n-t) E[X]}\bigg),\] which then gives the strongest multiplicative Chernoff bound~\eqref{eqprobCMUstrongest} by rewriting $t = (1+\delta) E[X]$.

Since it may help reading the literature, we add that $E[e^X]$ is called the \emph{exponential moment} of $X$ and $h \mapsto E[e^{hX}]$ is called the \emph{moment-generating function} of~$X$.

From the above proof sketch together with the remark on the tightness of Markov's inequality following Lemma~\ref{lprobmarkov}, we see that in almost all cases, our Chernoff bounds are not absolutely tight, that is, hold with ``$<$'' instead of ``$\le$''. The sole exceptions are (i) that $X$ takes only two values with positive probability, (ii)~that the tail event consists of a single point, e.g., $X \ge n$ or $X \le 0$ when $X$ is a sum of $n$ binary random variables, or (iii)~the tail event is empty, e.g., $X \ge n+1$ when $X$ is a sum of $n$ binary random variables. Having a ``$<$'' in a Chernoff bound will not drastically change things, but can occasionally be nice for cosmetic reasons.

\subsubsection{Chernoff Bounds with Estimates for the Expectation}\label{secprobEE}

Often we do not know the precise value of the expectation or it is tedious to compute it. In such cases, we can exploit the fact that \emph{all Chernoff bounds discussed in this work are valid also when the expectation is replaced by an upper or lower bound for it}. This is obvious for many bounds, e.g., from equations~\eqref{eqprobchernoffadditive01U},~\eqref{eqprobchernoffadditive01L}, and~\eqref{eqprobCMLeasy} we immediately derive the estimates 
\begin{align*}
  &\Pr[X \ge \mu^+ + \lambda] \le \Pr[X \ge \mu + \lambda] \le \exp\bigg(-\frac{2\lambda^2}{n}\bigg),\\
  &\Pr[X \le \mu^- - \lambda] \le \Pr[X \le \mu - \lambda] \le \exp\bigg(-\frac{2\lambda^2}{n}\bigg),\\
  &\Pr[X \le (1-\delta) \mu^-] \le \Pr[X \le (1-\delta) \mu] \le \exp\bigg(-\frac{\delta^2 \mu}{2}\bigg) \le \exp\bigg(-\frac{\delta^2 \mu^-}{2}\bigg) 
\end{align*}
for all $\mu^+ \ge E[X] =: \mu$ and $\mu^- \le E[X]$.

This is less obvious for a bound like $\Pr[X \ge (1+\delta) \mu^+] \le \exp(-\tfrac 13 \delta^2 \mu^+)$, since now also the probability of the tail event decreases for increasing $\mu^+$. However, also for such bounds we can replace $E[X]$ by an estimate as the following argument shows.

\begin{theorem}\label{tprobEE}
\begin{enumerate}
	\item Upper tail: Let $X_1, \dots, X_n$ be independent random variables taking values in $[0,1]$. Let $X = \sum_{i=1}^n X_i$. Let $\mu^+ \ge E[X]$. Then for all $\delta \ge 0$, 
  \begin{equation}
    \Pr[X \ge (1+\delta) \mu^+] \le \bigg(\frac{1}{1+\delta}\bigg)^{(1+\delta)\mu^+} \bigg(\frac{n-\mu^+}{n-(1+\delta)\mu^+}\bigg)^{n-(1+\delta)\mu^+},\label{eqprobCMUstrongestplus}
  \end{equation}
  where this bound is read as $0$ for $\delta > \frac{n-\mu^+}{\mu^+}$ and as $(\frac{\mu^+}{n})^n$ for $\delta = \frac{n-\mu^+}{\mu^+}$. Consequently, all Chernoff bounds of Theorem~\ref{tprobCMU} (including equations~\eqref{eqprobCMUstrongA} and~\eqref{eqprobCMUstrongB} and Corollary~\ref{corprob2hoch}) as well as those of Theorem~\ref{tprobCMUA} are valid when all occurrences of $E[X]$ are replaced by $\mu^+$. The additive bounds~\eqref{eqprobchernoffadditive01U} and \eqref{eqprobCAU} as well as those of Theorem~\ref{tprobcvar} and~\ref{tprobcvara} are trivially valid with the expectation replaced by an upper bound for it.  
  \item Lower tail: All Chernoff bounds of Theorem~\ref{tprobCML} and~\ref{tprobCMLA}, the ones in equations~\eqref{eqprobchernoffadditive01L} and \eqref{eqprobCAL}, as well as those of Corollary~\ref{corprobcvar} are valid when all occurrences of $E[X]$ are replaced by $\mu^- \le E[X]$.
  \end{enumerate}
\end{theorem}

\begin{proof}
  We first show~\eqref{eqprobCMUstrongestplus}. There is nothing to do when $(1+\delta)\mu^+ > n$, so let us assume $(1+\delta)\mu^+ \le n$. Let $\gamma = \frac{\mu^+ - E[X]}{n-E[X]}$. For all $i \in [1..n]$, define $Y_i$ by $Y_i = X_i + \gamma (1-X_i)$. Since $\gamma \le 1$, $Y_i \le 1$. By definition, $Y_i \ge X_i$, and thus also $Y \ge X$ for $Y := \sum_{i=1}^n Y_i$. Also, $\mu^+ = E[Y]$. Hence 
  \[\Pr[X \ge (1+\delta) \mu^+] \le \Pr[Y \ge (1+\delta) \mu^+] = \Pr[Y \ge (1+\delta) E[Y]].\]
  Now \eqref{eqprobCMUstrongestplus} follows immediately from Theorem~\ref{tprobCMU}, equation~\eqref{eqprobCMUstrongest}. Since~\eqref{eqprobCMUstrongestplus} implies all other Chernoff bounds of Theorem~\ref{tprobCMU} (including equations~\eqref{eqprobCMUstrongA} and~\eqref{eqprobCMUstrongB} and Corollary~\ref{corprob2hoch}) via elementary estimates, all these bounds are valid with $E[X]$ replaced by $\mu^+$ as well. This extends to Theorem~\ref{tprobCMUA}, since it is just a reformulation of Theorem~\ref{tprobCMU}. For the remaining (additive) bounds, replacing $E[X]$ by an upper bound only decreases the probability of the tail event, so clearly these remain valid.
  
  To prove our claim on lower tail bounds, it suffices to note that all bounds in Theorem~\ref{tprobCML} are monotonically decreasing in $E[X]$. So replacing $E[X]$ by some $\mu^- < E[X]$ makes the tail event less likely and increases the probability in the statement. Similarly, the additive bounds are not affected when replacing $E[X]$ with $\mu^-$.
\end{proof}

We note without proof that the variance of the random variable $Y$ constructed above is at most the one of $X$. Since the tail bound in Theorem~\ref{tprobcvar} is increasing in $\sigma^2$, the same argument as above also shows that multiplicative versions of Theorem~\ref{tprobcvar} such as equation~\eqref{eqprobcvarmult} remain valid when all occurrences of $E[X]$ are replaced by an upper bound $\mu^+ \ge E[X]$.

\subsection{Chernoff Bounds for Sums of Dependent Random Variables}

In the previous subsection, we discussed large deviation bounds for the classic setting of sums of independent random variables. In the analysis of algorithms, often we cannot fully satisfy the assumption of independence. The dependencies may appear minor, maybe even in our favor in some sense, so we could hope for some good large deviations bounds. 

In this section, we discuss three such situations which all lead to (essentially) the known Chernoff bounds being applicable despite perfect independence missing. The first of these was already discussed in Section~\ref{secprobmoderate}, so we just note here how it also implies the usual Chernoff bounds.

\subsubsection{Unconditional Sequential Domination}

In the analysis of sequential random processes such as iterative randomized algorithms, we rarely encounter that the events in different iterations are independent, simply because the actions of our algorithm depend on the results of the previous iterations. However, due to the independent randomness used in each iteration, we can often say that, independent of what happened in iterations $1, \dots, t-1$, in iteration $t$ we have a particular event with at least some probability $p$. 

This property was made precise in the definition of \emph{unconditional sequential domination} before Lemma~\ref{lprobmoderate}. The lemma then showed that unconditional sequential domination leads to domination by a sum of independent random variables. Any upper tail bounds for these naturally are valid also for the original random variables. We make this elementary insight precise in the following lemma. This type of argument was used, among others, in works on shortest path problems~\cite{DoerrJ10,DoerrHK11,DoerrHK12}. There one can show that, in each iteration, independent of the past, with at least a certain probability an extra edge of a desired path is found. This type of argument was also used in~\cite{DoerrJSWZ13} to construct a monotonic function that is difficult to optimize.

\begin{lemma}
Let $(X_1, \dots, X_n)$ and $(X_1^*, \dots, X_n^*)$ be finite sequences of discrete random variables. Assume that $X_1^*, \dots, X_n^*$ are independent. 
\begin{enumerate}
	\item If $(X_1^*, \dots, X_n^*)$ unconditionally sequentially dominates $(X_1, \dots, X_n)$, then for all $\lambda \in \R$, we have $\Pr[\sum_{i=1}^n X_i \ge \lambda] \le \Pr[\sum_{i=1}^n X_i^* \ge \lambda]$ and the latter expression can be bounded by Chernoff bounds for the upper tail of independent random variables.
	\item If $(X_1^*, \dots, X_n^*)$ unconditionally sequentially subdominates $(X_1, \dots, X_n)$, then for all $\lambda \in \R$, we have $\Pr[\sum_{i=1}^n X_i \le \lambda] \le \Pr[\sum_{i=1}^n X_i^* \le \lambda]$ and the latter expression can be bounded by Chernoff bounds for the lower tail of independent random variables.
\end{enumerate}
\end{lemma}

\subsubsection{Negative Correlation}\label{secprobnegcor}

Occasionally, we encounter random variables that are not independent, but that display an intuitively even better negative correlation behavior. 
Take as example the situation that we do not flip bits independently with probability $\frac kn$, but that we flip a set of exactly $k$ bits chosen uniformly at random among all sets of $k$ out of $n$ bits. Let $X_1, \dots, X_n$ be the indicator random variables for the events that bit $1, \dots, n$ flips. Clearly, the $X_i$ are not independent. If $X_1 = 1$, then $\Pr[X_2 = 1] = \frac{k-1}{n-1}$, which is different from the unconditional probability $\frac kn$. However, things feel even better than independent: Knowing that $X_1 = 1$ actually reduces the probability that $X_2 =1$. This intuition is made precise in the following notion of \emph{negative correlation}\index{negative correlation}. 

Let $X_1, \ldots, X_n$ binary random variables. 
We say that $X_1, \ldots, X_n$ are  \emph{\mbox{$1$-negatively} correlated} if for all $I \subseteq [1..n]$ we have 
\begin{equation*}
	\Pr[\forall i \in I : X_i = 1] \le \prod_{i \in I} \Pr[X_i = 1].
\end{equation*}
We say that $X_1, \ldots, X_n$ are  \emph{$0$-negatively correlated} if for all $I \subseteq [1..n]$ we have 
\begin{equation*}
	\Pr[\forall i \in I : X_i = 0] \le \prod_{i \in I} \Pr[X_i = 0].
\end{equation*}
Finally, we call $X_1, \ldots, X_n$ \emph{negatively correlated} if they are both $0$-negatively correlated and $1$-negatively correlated.

In simple words, these conditions require that the event that a set of variables is all zero or all one, is at most as likely as in the case of independent random variables. It seems natural that sums of such random variables are at least as strongly concentrated as independent random variable, and in fact, Panconesi and Srinivasan~\cite{PanconesiS97} were able to prove that negatively correlated random variables admit Chernoff bounds. To be precise, they only proved that $1$-negative correlation implies Chernoff bounds for the upper tail, but it is not too difficult to show (see below) that their main argument works for all bounds proven via Bernstein's exponential moments method. In particular, for sums of $1$-negatively correlated random variable we obtain all Chernoff bounds for the upper tail that were presented in this work for independent random variables (as far as they can be applied to binary random variables). We prove a slightly more general result as this helps arguing that we can also work with upper bounds for the expectation instead of the precise expectation. Further below, we then use a symmetry argument to argue that $0$-negative correlation implies all lower tail bounds presented so far. 

\begin{theorem}[$1$-negative correlation implies upper tail bounds]\label{tprobnegcor1}
  Let $X_1, \ldots, X_n$ be $1$-negatively correlated binary random variables. Let $a_1, \ldots, a_n, b_1, \dots, b_n \in \R$ with $a_i \le b_i$ for all $i \in [1..n]$. Let $Y_1, \dots, Y_n$ be random variables with $\Pr[Y_i = a_i] = \Pr[X_i = 0]$ and $\Pr[Y_i = b_i] = \Pr[X_i = 1]$. Let $Y = \sum_{i=1}^n Y_i$.  
  \begin{enumerate}
  \item If $a_1, \ldots, a_n, b_1, \dots, b_n \in [0,1]$, then $Y$ satisfies the Chernoff bounds given in equations~\eqref{eqprobCMUstrongest} to~\eqref{eqprobCMUeasy}, \eqref{eqprobCMUstrongA} to~\eqref{eqprob2hoch}, \eqref{eqprobchernoffadditive01U}, \eqref{eqprobCMUAstrongest} to~\eqref{eqprobCMUAeasy}, and \eqref{eqprobcvarstrongest} to \eqref{eqprobcvarmultlin}, where in the latter we use $\sigma^2 := \sum_{i=1}^n \Var[Y_i]$.
  \item Without the restriction to $[0,1]$ made in the previous paragraph, $Y$ satisfies the Chernoff bound of~\eqref{eqprobCAU} with $c_i := b_i - a_i$ and the bounds of~\eqref{eqprobcvarstrongesta} to~\eqref{eqprobcvarlin2a} with $\sigma^2 := \sum_{i=1}^n \Var[Y_i]$.
  \end{enumerate}
  Each of these results also holds when all occurrences of $E[Y]$ are replaced by $\mu^+$ for some $\mu^+ \ge E[Y]$.
\end{theorem}

\begin{proof}
  Let $X'_1, \dots, X'_n$ be independent binary random variables such that for each $i \in [1..n]$, the random variables $X_i$ and $X'_i$ are identically distributed. Let $c_i := b_i - a_i$ for all $i \in [1..n]$. Note that $Y_i = a_i + c_i X_i$ for all $i \in [1..n]$. Let $Y'_i = a_i + c_i X'_i$. Let $Y' = \sum_{i=1}^n Y'_i$. 
  
  We first show that the $1$-negative correlation of the $X_i$ implies $E[Y^\ell] \le E[(Y')^\ell]$ for all $\ell \in \N_0$. There is nothing to show for $\ell = 0$ and $\ell = 1$, so let $\ell \ge 2$. Since $Y = (\sum_{i=1}^n a_i) + (\sum_{i=1}^n c_i X_i)$, we have $Y^\ell = \sum_{k=0}^\ell \binom{\ell}{k} (\sum_{i=1}^n a_i)^{\ell-k} (\sum_{i=1}^n c_i X_i)^{k}$. By linearity of expectation, it suffices to show $E[(\sum_{i=1}^n c_i X_i)^k] \le E[(\sum_{i=1}^n c_i X'_i)^k]$. We have $(\sum_{i=1}^n c_i X_i)^k = \sum_{(i_1,\dots,i_k) \in [1..n]^k} \prod_{j=1}^k c_{i_j} X_{i_j}$. Applying the definition of $1$-negative correlation to the set $I = \{i_1, \dots, i_k\}$, we compute
\begin{align*}
E\bigg[\prod_{j=1}^k c_{i_j} X_{i_j}\bigg] 
&= \bigg(\prod_{j=1}^k c_{i_j}\bigg) \Pr[\forall j \in [1..k] : X_{i_j} = 1] \\
&\le \bigg(\prod_{j=1}^k c_{i_j}\bigg) \bigg(\prod_{i \in I} \Pr[X_i = 1]\bigg) \\
&= \bigg(\prod_{j=1}^k c_{i_j}\bigg) \bigg(\prod_{i \in I} \Pr[X'_i = 1]\bigg) = E\bigg[\prod_{j=1}^k c_{i_j} X'_{i_j}\bigg].
\end{align*}
Consequently, by linearity of expectation, $E[(\sum_{i=1}^n c_i X_i)^k] \le E[(\sum_{i=1}^n c_i X'_i)^k]$ for all $k \in \N$ and thus $E[Y^\ell] \le E[(Y')^\ell]$.
  
  We recall from Section~\ref{secprobcproof} that essentially all large deviation bounds are proven via upper bounds on the exponential moment $E[e^{hY}]$ of the random variable~$hY$, where $h > 0$ is suitably chosen. Since the random variable $Y$ is bounded, by Fubini's theorem we have 
\begin{equation}
  E[e^{hY}] = E\bigg[\sum_{\ell=0}^\infty \frac{h^\ell Y^\ell}{\ell!}\bigg] = \sum_{\ell=0}^\infty \frac{h^\ell E[Y^\ell]}{\ell!}.\label{eqprobfubini}
\end{equation}
Since $E[Y^\ell] \le E[(Y')^\ell]$, we have $E[e^{hY}] \le E[e^{hY'}]$. Consequently, we obtain for $Y$ all Chernoff bounds which we could prove with the classic methods for $Y'$. 

It remains to show that we can also work with an upper bound $\mu^+ \ge E[Y]$. For this, note that when we apply the construction of Theorem~\ref{tprobEE} to our random variables $Y_i$, that is, we define $Z_i = Y_i + \gamma(1-Y_i)$ for a suitable $\gamma \in [0,1]$, then the resulting random variables $Z_i$ satisfy the same properties as the $Y_i$, that is, there are $a'_i$ and $c'_i$ such that $Z_i = a'_i + c'_i X_i$. Consequently, we have $\Pr[Y \ge (1+\delta) \mu^+] \le \Pr[Z \ge (1+\delta)\mu^+] = \Pr[Z \ge (1+\delta) E[Z]]$ for the sum $Z=\sum_{i=1}^n Z_i$ and the last expression can bounded via the results we just proved. 
\end{proof}

\begin{theorem}[$0$-negative correlation implies lower tail bounds]\label{tprobnegcor0}
  Let $X_1, \ldots, X_n$ be $0$-negatively correlated binary random variables. Let $a_1, \ldots, a_n, b_1, \dots, b_n \in \R$ with $a_i \le b_i$ for all $i \in [1..n]$. Let $Y_1, \dots, Y_n$ be random variables with $\Pr[Y_i = a_i] = \Pr[X_i = 0]$ and $\Pr[Y_i = b_i] = \Pr[X_i = 1]$. Let $Y = \sum_{i=1}^n Y_i$.  
  \begin{enumerate}
  \item If $a_1, \ldots, a_n, b_1, \dots, b_n \in [0,1]$, then $Y$ satisfies the Chernoff bounds given in equations~\eqref{eqprobCMLstrongest} to~\eqref{eqprobCMLeasy}, \eqref{eqprobchernoffadditive01L}, \eqref{eqprobCMLAstrongest} to~\eqref{eqprobCMLAeasy}, and those in Corollary~\ref{corprobcvar} with $\sigma^2 := \sum_{i=1}^n \Var[Y_i]$.
  \item Without the restriction to $[0,1]$ made in the previous paragraph, $Y$ satisfies the Chernoff bound of~\eqref{eqprobCAL} with $c_i := b_i - a_i$ and those of the last paragraph of Theorem~\ref{tprobcvara} with $\sigma^2 := \sum_{i=1}^n \Var[Y_i]$.
  \end{enumerate}
  Each of these results also holds when all occurrences of $E[Y]$ are replaced by $\mu^-$ for some $\mu^- \le E[Y]$.
\end{theorem}

\begin{proof}
  Let $\tilde Y_i := 1 - Y_i$. Then the $\tilde Y_i$ satisfy the assumptions of Theorem~\ref{tprobnegcor1} (with $\tilde a_i = 1- b_i$, $\tilde b_i = 1 - a_i$, and $\tilde X_i = 1-X_i$; note that the latter are $1$-negatively correlated since the $X_i$ are $0$-negatively correlated, note further that $\tilde a_i, \tilde b_i \in [0,1]$ if $a_i, b_i \in [0,1]$). Hence Theorem~\ref{tprobnegcor1} gives the usual Chernoff bounds for the~$\tilde Y_i$. As in Section~\ref{secprobCML}, these translate into the estimates~\eqref{eqprobCMLstrongest} to~\eqref{eqprobCMLeasy} and these imply~\eqref{eqprobCMLAstrongest} to~\eqref{eqprobCMLAeasy}. Bound~\eqref{eqprobchernoffadditive01U} for the $\tilde Y_i$ immediately translates to~\eqref{eqprobchernoffadditive01L} for the~$Y_i$. Finally, the results of Theorem~\ref{tprobcvar} imply those of Corollary~\ref{corprobcvar}. All these results are obviously weaker when $E[Y]$ is replaced by some $\mu^- \le E[Y]$.
\end{proof}

\subsubsection{Hypergeometric Distribution}

It remains to point out some situations where we encounter negatively correlated random variables. One typical situations (but by far not the only one) is sampling without replacement, which leads to the \emph{hypergeometric distribution}\index{random variable!hypergeometric}. 

Say we choose randomly $n$ elements from a given $N$-element set $S$ \emph{without replacement}. For a given $m$-element subset $T$ of $S$, we wonder how many of its elements we have chosen. This random variable is called hypergeometrically distributed with parameters $N$, $n$ and $m$. 

More formally, let $S$ be any $N$-element set. Let $T \subseteq S$ have exactly $m$ elements. Let $U$ be a subset of $S$ chosen uniformly among all $n$-element subsets of $S$. Then $X = |U \cap T|$ is a random variable with hypergeometric distribution (with parameters $N$, $n$ and $m$). By definition, \[\Pr[X = k] = \frac{\binom mk \binom {N-m}{n-k}}{\binom Nn}\] for all $k \in [\max\{0,n+m-N\}..\min\{n,m\}]$.

It is easy to see that $E[X] = \frac{|U||T|}{|S|} = \frac{mn}{N}$: Enumerate $T = \{t_1, \ldots, t_m\}$ in an arbitrary manner (before choosing $U$). For $i = 1, \ldots, m$, let $X_i$ be the indicator random variable for the event $t_i \in U$. Clearly, $\Pr[X_i = 1] = \frac{|U|}{|S|} = \frac nN$. Since $X = \sum_{i = 1}^m X_i$, we have $E[X] = \frac{mn}{N}$ by linearity of expectation (Lemma~\ref{lproblinearity}). 

It is also obvious that the $X_i$ are not independent. If $n < m$ and $X_1 = \ldots = X_n = 1$, then necessarily we have $X_i = 0$ for $i > n$. Fortunately, however, these dependencies are of the negative correlation type. This is intuitively clear, but also straight-forward to prove. 

Let $I \subseteq [1..m]$, $W = \{t_i \mid i \in I\}$, and $w = |W| = |I|$. Then $\Pr[\forall i \in I : X_i = 1] = \Pr[W \subseteq U]$. Since $U$ is uniformly chosen, it suffices to count the number of $U$ that contain $W$, these are $\binom{|S \setminus W|}{|U \setminus W|}$, and to compare them with the total number of possible $U$. Hence 
\[\Pr[W \subseteq U] = \binom{N-w}{n-w} \Big/ \binom{N}{n} = \frac{n \cdot \ldots \cdot (n-w+1)}{N \cdot \ldots \cdot (N - w +1)} < \left(\frac nN\right)^w = \prod_{i \in I} \Pr[X_i = 1].\] 
In a similar fashion, we have
\begin{align*}
\Pr[\forall i \in I : X_i = 0] 
&= \Pr[U \cap W = \emptyset]\\
&= \binom{N-w}{n} \Big/ \binom Nn \\
&= \frac{(N-n) \dots (N-n-w+1)}{N \dots (N-w+1)}\\
&\le \left(\frac{N-n}{N}\right)^w = \prod_{i \in I} \Pr[X_i = 0],
\end{align*}
where we read $\binom{N-w}{n} = 0$ when $n > N-w$. 

Together with Theorems~\ref{tprobnegcor1} and~\ref{tprobnegcor0}, we obtain the following theorem.

\begin{theorem}\label{tprobhypergeometric}
  Let $N \in \N$. Let $S$ be some set of cardinality $N$, for convenience, $S = [1..N]$. Let $n \le N$ and let $U$ be a subset of $S$ having cardinality $n$ uniformly chosen among all such subsets. For $i \in [1..N]$ let $X_i$ be the indicator random variable for the event $i \in U$. Then $X_1, \dots, X_N$ are negatively correlated. 
  
  Consequently, if $X$ is a random variable having a hypergeometric distribution with parameters $N$, $n$, and $m$, then the usual Chernoff bounds for sums of $n$ independent binary random variables (listed in Theorems~\ref{tprobnegcor1} and~\ref{tprobnegcor0}) hold.
\end{theorem}\index{Chernoff inequality!moderate independence}

Note that for hypergeometric distributions we have symmetry in $n$ and $m$, that is, the hypergeometric distribution with parameters $N$, $n$, and $m$ is the same as the hypergeometric distribution with parameters $N$, $m$, and $n$. Hence for Chernoff bounds depending on the number of random variables, e.g., Theorem~\ref{tprobchernoffadditive01}, we can make this number being $\min\{n,m\}$ by interpreting the random experiment in the right fashion.

That the hypergeometric distribution satisfies the Chernoff bounds of Theorem~\ref{tprobchernoffadditive01} has recently in some works 
been attributed to Chv\'atal~\cite{Chvatal79}, but this is not correct. As Chv\'atal writes, the aim of his note is solely to give an elementary proof of the fact that the hypergeometric distribution satisfies the strongest Chernoff bound of Theorem~\ref{tprobCMU} (which implies the bounds of Theorem~\ref{tprobchernoffadditive01}), whereas the result itself is from Hoeffding~\cite{Hoeffding63}. 

For a hypergeometric random variable $X$ with parameters $N$, $n$, and~$m$, \cite[Lemma~2]{BadkobehLS14} shows that if $m < \frac{N}{2e}$ and $z \ge \frac n2$, then $\Pr[X = z] \le (\frac{2em}{N})^z$. With Theorem~\ref{tprobhypergeometric}, we can use the usual Chernoff bound of equation~\eqref{eqprobCMUAstrong} and obtain the stronger bound 
\begin{equation}
\Pr[X \ge z] \le e^{z - E[X]} \left(\frac{E[X]}{z}\right)^z \le \left(\frac{e E[X]}{z}\right)^z = \left(\frac{enm}{z N}\right)^z,
\end{equation} 
which is at most $(\frac{2em}{N})^z$ for $z \ge n/2$.

Theorem~\ref{tprobhypergeometric} can be extended to point-wise maxima of several families like~$(X_i)$ in Theorem~\ref{tprobhypergeometric} if these are independent. This result was used in the analysis of a population-based genetic algorithm in~\cite{DoerrD18}, but might be useful also in other areas of discrete algorithmics.

\begin{lemma}\label{lprobmaxhyper}
  Let $k, N \in \N$. For all $j \in [1..k]$, let $n_j \in [1..N]$. Let $S$ be some set of cardinality $N$, for convenience, $S = [1..N]$. For all $j \in [1..k]$, let $U_j$ be a subset of $S$ having cardinality $n_j$ uniformly chosen among all such subsets. Let the $U_j$ be stochastically independent. For all $i \in S$, let $X_i$ be the indicator random variable for the event that $i \in U_j$ for some $j \in [1..k]$. Then the random variables $X_1, \dots, X_N$ are negatively correlated. 
\end{lemma}

Note that the situation in the lemma above can be seen as sampling with partial replacement. We sample a total of $\sum_j n_j$ elements, but we replace the elements chosen only after round $n_1$, $n_1+n_2$, ... We expect that other partial replacement scenarios also lead to negatively correlated random variables, and thus to the usual Chernoff bounds.

We recall that negative correlation can be useful also without Chernoff bounds. For example, in Section~\ref{secprobcoupon} we used the lemma above to prove a lower bound on the coupon collector time (equivalently, on the runtime of the randomized local search heuristic on monotonic functions).

\subsection{Chernoff Bounds for Functions of Independent Variables, Martingales, and Bounds for Maxima}

So far we discussed tail bounds for random variables which can be written as sum of (more or less) independent random variables. Sometimes, the random variable we are interested in is determined by the outcomes of many independent random variables, however, not simply as a sum of these. Nevertheless, if each of the independent random variables has only a limited influence on the result, then bounds similar to those of Theorem~\ref{tprobchernoffadditive} can be proven. Such bounds can be found under the names \emph{Azuma's inequality}\index{Azuma inequality}, \emph{Martingale inequalities}, or \emph{method of bounded differences}. 

To the extent possible, we shall try to avoid the use of martingales. The following two bounds due to McDiarmid~\cite{McDiarmid98} need martingales in their proof, but not in their statement.

\begin{theorem}[Method of bounded differences]\label{tprobboundeddiff}
  Let $X_1, \ldots, X_n$ be independent random variables taking values in the sets $\Omega_1, \ldots, \Omega_n$, respectively. Let $\Omega := \Omega_1 \times \ldots \times \Omega_n$. Let $f : \Omega \to \R$. For all $i \in [1..n]$ let $c_i > 0$ be such that for all $\omega, \bar\omega \in \Omega$ we have that if for all $j \neq i$, $\omega_j = \bar\omega_j$, then $|f(\omega) - f(\bar\omega)| \le c_i$. 
  
  Let $X = f(X_1, \ldots, X_n)$. Then for all $\lambda \ge 0$,  
  \begin{align*}
    \Pr[X \ge E[X] + \lambda] & \le \exp\bigg(-\frac{2\lambda^2}{\sum_{i=1}^n c_i^2}\bigg),\\
    \Pr[X \le E[X] - \lambda] & \le \exp\bigg(-\frac{2\lambda^2}{\sum_{i=1}^n c_i^2}\bigg).
  \end{align*}
\end{theorem}

The version of Azuma's inequality given above is due to McDiarmid~\cite{McDiarmid98} and is stronger than the bound $\exp(-\lambda^2 / 2\sum_{i=1}^n c_i^2)$ given by several other authors.

Theorem~\ref{tprobboundeddiff} found numerous applications in discrete mathematics and computer science, however, only few in the analysis of randomized search heuristics (the only one we are aware of is~\cite{BuzdalovD17}). All other analyses of randomized search heuristics that needed Chernoff-type bounds for random variables that are determined by independent random variables, but in a way other than as simple sum, resorted to the use of martingales.

One reason for this might be that the bounded differences assumption is easily proven in discrete mathematics problems like the analysis of random graphs, whereas in algorithms the sequential nature of the use of randomness makes it hard to argue that a particular random variable sampled now has a bounded influence on the final result regardless of how we condition on all future random variables. A more natural condition might be that the outcome of the current random variable has only a limited influence on the expected result determined by the future random variables. For this reason, we are optimistic that the following result might become useful in the analysis of randomized search heuristics. This result is a weak version of Theorem~3.7 in~\cite{McDiarmid98}.

\begin{theorem}[Method of bounded conditional expectations]\label{tprobboundedexp}
  Let $X_1, \ldots, X_n$ be independent random variables taking values in the sets $\Omega_1, \ldots, \Omega_n$, respectively. Let $\Omega := \Omega_1 \times \ldots \times \Omega_n$. Let $f : \Omega \to \R$. For all $i \in [1..n]$ let $c_i > 0$ be such that for all $\omega_1 \in \Omega_1, \dots, \omega_{i-1} \in \Omega_{i-1}$ and all $\omega_i, \bar\omega_i \in \Omega_i$ we have 
  \[|E[f(\omega_1,\dots,\omega_{i-1},\omega_i,X_{i+1}, \dots, X_n)] - E[f(\omega_1,\dots,\omega_{i-1},\bar\omega_i,X_{i+1}, \dots, X_n)]| \le c_i.\]
  Let $X = f(X_1, \ldots, X_n)$. Then for all $\lambda \ge 0$,  
  \begin{align*}
    \Pr[X \ge E[X] + \lambda] & \le \exp\bigg(-\frac{2 \lambda^2}{\sum_{i=1}^n c_i^2}\bigg),\\
    \Pr[X \le E[X] - \lambda] & \le \exp\bigg(-\frac{2 \lambda^2}{\sum_{i=1}^n c_i^2}\bigg).
  \end{align*}
\end{theorem}

Here is an example of how the new theorem can be helpful. The \emph{compact genetic algorithms (cGA)} without frequency boundaries maximizes a function $f: \{0,1\}^n \to \R$ as follows. There is a (hypothetical) population size $K \in \N$, which we assume to be an even integer. The cGA sets the initial \emph{frequency vector} $\tau^{(0)} \in [0,1]^n$ to $\tau^{(0)} = (\frac 12, \dots, \frac 12)$. Then, in each iteration $t = 1, 2, \dots$ it generates two search points $x^{(t,1)}, x^{(t,2)} \in \{0,1\}^n$ randomly such that, independently for all $j \in \{1, 2\}$ and $i \in [1..n]$, we have $\Pr[x_i^{(t,j)} = 1] = \tau_i^{(t)}$. If $f(x^{(t,1)}) < f(x^{(t,2)})$, then we swap the two variables, that is, we set $(x^{(t,1)},x^{(t,2)}) \assign (x^{(t,2)},x^{(t,1)})$. Finally, in this iteration, we update the frequency vector by setting $\tau^{(t+1)} \assign \tau^{(t)} + \frac 1K (x^{(t,1)}-x^{(t,2)})$. 

Let us analyze the behavior of the frequency $\tau_i^{(t)}$ of a \emph{neutral bit} $i\in [1..n]$, that is, one that has property that $f(x) = f(y)$ for all $x$ and $y$ which differ only in the $i$-th bit. 

\begin{lemma}
  Let $K$ be an even integer. Consider a run of the cGA with hypothetical population size $K$ on an objective function $f : \{0,1\}^n \to \R$ having a neutral bit~$i$. For all $T \in \N$, the probability that within the first $T$ iterations the frequency of the $i$-th bit has converged to one of the absorbing states $0$ or $1$ is at most $2 \exp\left(-\frac{K^2}{32 T}\right)$.
\end{lemma}

\begin{proof}
To ease reading, let $X_t := \tau_i^{(t)}$. We have $X_0 = \frac 12$ with probability one. Once $X_t$ is determined, we have
\begin{align*}
  \Pr[X_{t+1} = X_t + \tfrac 1K] &= X_t (1-X_t),\\
  \Pr[X_{t+1} = X_t - \tfrac 1K] &= X_t (1-X_t),\\
  \Pr[X_{t+1} = X_t] &= 1 - 2 X_t (1-X_t).
\end{align*}
In particular, we have $E[X_{t+1} \mid X_0, \dots, X_t] = E[X_{t+1} \mid X_t] = X_t$. By induction, we have $E[X_T \mid X_t] = X_t$ for all $T > t$. 

Our aim is to show that with probability at least $1 -2 \exp\left(-\frac{K^2}{32 T}\right)$, $X_T$ has not yet converged to one of the absorbing states $0$ and $1$. We first write the frequencies as results of independent random variables. For convenience, these will be continuous random variables, but it is easy to see that instead we could have used discrete ones as well. For all $t = 1, 2, \dots$ let $R_t$ be a random number uniformly distributed in the interval $[0,1]$. Define $Y_0, Y_1, \dots$ as follows. We have $Y_0 = \frac 12$ with probability one. For $t \in \N_0$, we set
\begin{align*}
  \Pr[Y_{t+1} = Y_t + \tfrac 1K] &\mbox{ if $R_t \ge 1 - Y_t (1-Y_t)$},\\
  \Pr[Y_{t+1} = Y_t - \tfrac 1K] &\mbox{ if $R_t \le  Y_t (1-Y_t)$},\\
  \Pr[Y_{t+1} = Y_t] &\mbox{ else.}
\end{align*}
It is easy to see that $(X_0, X_1, \dots)$ and $(Y_0, Y_1, \dots)$ are identically distributed. Note that $Y_T$ is a function $g$ of $(R_1, \dots, R_T)$. For concrete values $r_1, \dots, r_t \in [0,1]$, we have $E[g(r_1, \dots, r_t, R_{t+1}, \dots, R_T)] = E[Y_T \mid Y_t] = Y_t$. Consequently, for all $\overline r_t \in [0,1]$,  the two expectations $E[g(r_1, \dots, r_{t-1}, r_t, R_{t+1}, \dots, R_T)]$ and $E[g(r_1, \dots, r_{t-1}, \overline r_t, R_{t+1}, \dots, R_T)]$ are two possible outcomes of $Y_t$ given a common value for $Y_{t-1}$ (which is determined by $r_1, \dots, r_{t-1}$), and hence differ by at most $c_t = \frac 2K$. We can thus apply Theorem~\ref{tprobboundedexp} as follows.
\begin{align*}
  \Pr[Y_T \in \{0,1\}] &= \Pr\big[|Y_T - \tfrac 12| \ge \tfrac 12\big]\\
  & = \Pr\big[|g(R_1,\dots,R_T) - E[g(R_1,\dots,R_T)]| \ge \tfrac 12\big]\\
  & \le 2 \exp\left(-\frac{(\frac 12)^2}{2 T (\frac 2K)^2}\right) = 2 \exp\left(-\frac{K^2}{32 T}\right).
\end{align*} 
\end{proof}

Note that it is not obvious how to obtain this result with the classic method of bounded differences (Theorem~\ref{tprobboundeddiff}). In particular, the above construction does not satisfy the bounded differences condition, that is, there are values $r_1, \dots, r_T$ and $\overline r_t$ such that $g(r_1, \dots, r_T)$ and $g(r_1, \dots, r_{t-1}, \overline r_t, r_{t+1}, \dots, r_T)$ differ by significantly more than $\frac 2K$. To see this, consider the following example. Let $r_i = 1$ for even $i$ and $r_i = \frac 14$ for odd $i$. Then $g(r_1,\dots,r_T) = \frac 12$ for even $T$ and $g(r_1,\dots,r_T) = \frac 12 - \frac 1K$ for odd $T$. However, for all even $T$ we have $g(\frac 12,r_2,\dots,r_T) = g(\frac 12,r_2,\dots,r_{T+1}) = \min\{1, \frac 12 + \frac T2 \cdot \frac 1K\}$, showing that a change in the first variable leads to a drastic change in the $g$-values for larger $T$.

This example shows that our stochastic modeling of the process cannot be analyzed via the method of bounded differences. We cannot rule out that a different modeling admits an analysis via the method of bounded differences, but nevertheless this example suggests that Theorem~\ref{tprobboundedexp} might be a useful tool in the theory of randomized search heuristics.

Without going into details (and in particular, without defining the notion of a martingale), we note that both Theorem~\ref{tprobboundeddiff} and~\ref{tprobboundedexp} are special cases of the following martingale result, which is often attributed to Azuma~\cite{Azuma67} despite the fact that it was proposed already in Hoeffding~\cite{Hoeffding63}. Readers familiar with martingales may find it more natural to use this result rather than the previous two theorems in their work, however, it has to be said that not all researchers in algorithms theory are familiar with martingales. 
\begin{theorem}[Azuma-Hoeffding inequality]
  Let $X_0, X_1, \dots, X_n$ be a martingale. Let $c_1, \dots, c_n > 0$ with $|X_i - X_{i-1}| \le c_i$ for all $i \in [1..n]$. Then for any $\lambda \ge 0$, \[\Pr[X_n - X_0 \ge \lambda] \le \exp\bigg(-\frac{\lambda^2}{2\sum_{i=1}^n c_i^2}\bigg).\]
\end{theorem}

This result has found several applications in the theory of randomized search heuristics, e.g., in~\cite{DoerrK15,Kotzing16,DoerrDY16gecco}. 

We observe that the theorem above is a direct extension of Theorem~\ref{tprobchernoffadditive} to martingales (note that the $c_i$ there are twice as large as here, which explains the different location of the $2$ in the bounds). In a similar vein, there are martingale versions of most other Chernoff bounds presented in this work. We refer to McDiarmid~\cite{McDiarmid98} for more details.

\subsubsection{Tail Bounds for Maxima and Minima of Partial Sums}

We end this section with a gem already contained in Hoeffding. It builds on the following elementary observation: If $X_0, X_1, \dots, X_n$ is a martingale, then $Y_0, Y_1, \dots, Y_n$ defined as follows also forms a martingale. Let $\lambda \in \R$. Let $i \in [0..n]$ minimal with $X_i \ge \lambda$, if such an $X_i$ exists, and $i = n+1$ otherwise. Let $Y_j = X_j$ for $j \le i$ and $Y_j = X_i$ for $j > i$. Then $Y_n \ge \lambda$ if and only if $\max_{i \in [1..n]} X_i \ge \lambda$. Since $Y_0, \dots, Y_n$ is a martingale with martingales differences bounded as least as well as for $X_0,\dots,X_n$ (and also all other variation measures at least as good as for $X_0,\dots, X_n$), all large deviation bounds provable for the martingale $X_0,\dots,X_n$ via the Bernstein method are valid also for $Y_0, \dots, Y_n$, that is, for $\max_{i \in [1..n]} X_i$. Since we did not introduce martingales here, we omit the details and only state some implications of this observation. The reader finds more details in~\cite[end of Section~2]{Hoeffding63} and~\cite[end of Section~3.5]{McDiarmid98}. It seems that both authors do not see this extension as very important (see also the comment at the end of Section~2 in~\cite{McDiarmid98}). We feel that this might be different for randomized search heuristics. For example, to prove that a randomized search heuristic has at least some optimization time $T$, we need to show that the distance of each of the first $T-1$ solutions from the optimum is positive, that is, that the minimum of these differences is positive.

\begin{theorem}[Tail bounds for maxima and minima]
Let $X_1, \dots, X_n$ be independent random variables. For all $i \in [1..n]$, let $S_i = \sum_{j=1}^i X_j$. Assume that one of the results in Section~\ref{secprobchernoffindependent} yields the tail bound $\Pr[S_n \ge E[S_n] + \lambda] \le p$. Then we also have 
\begin{equation}
\Pr[\exists i \in [1..n] : S_i \ge E[S_i] + \lambda] \le p. \label{eqprobmax}
\end{equation}
In an analogous manner, each tail bound $\Pr[S_n \le E[S_n] - \lambda] \le p$ derivable from Section~\ref{secprobchernoffindependent} can be strengthened to $\Pr[\exists i \in [1..n] : S_i \le E[S_i] - \lambda] \le p$.
\end{theorem}

Note that if the $X_i$ in the theorem are non-negative, then trivially equation~\eqref{eqprobmax} implies the uniform bound
\begin{equation}
\Pr[\exists i \in [1..n] : S_i \ge E[S_n] + \lambda] \le p. \label{eqprobmax2}
\end{equation}
Note also that the deviation parameter $\lambda$ does not scale with $i$. In particular, a bound like
$\Pr[\exists i \in [1..n] : S_i \ge (1+\delta) E[S_i]] \le p$ cannot be derived.

\subsection{Chernoff Bounds for Geometric Random Variables}\index{Chernoff inequality!geometric random variables}\label{secprobcgeom}

As visible from Lemma~\ref{lprobdomexamples} below, sums of independent geometric random variables occur frequently in the analysis of randomized search heuristics. Surprisingly, it was only in 2007 that a Chernoff-type bound was used to analyze such sums in the theory of randomized search heuristics~\cite{DoerrHK07} (for subsequent uses see, e.g.,~\cite{BaswanaBDFKN09,DoerrHK11,ZhouLLH12,DoerrJWZ13,DoerrD18}). Even more surprisingly, only recently Witt~\cite{Witt14} proved good tail bounds for sums of geometric random variables having significantly different success probabilities. Note that geometric random variables are unbounded. Hence the Chernoff bounds presented so far cannot be applied directly. 

We start this subsection with simple Chernoff bounds for sums of identically distributed geometric random variables as these can be derived from Chernoff bounds for sums of independent $0,1$ random variables discussed so far. We remark that a sum $X$ of $n$ independent geometric distributions with success probability $p>0$ is closely related to the \emph{negative binomial distribution}\index{negative binomial distribution} $\text{NB}(n,1-p)$ with parameters $n$ and $1-p$: We have $X \sim \text{NB}(n,1-p) + n$.

\begin{theorem}
  Let $X_1, \ldots, X_n$ be independent geometric random variables with common success probability $p>0$. Let $X:=\sum_{i=1}^n X_i$ and $\mu := E[X] = \frac np$.
  \begin{enumerate}
  \item For all $\delta \ge 0$, 
  \begin{align}
  \Pr[X \ge (1+\delta)\mu]  \le \exp\left(-\frac{\delta^2}{2}\frac{n-1}{1+\delta}\right) \le \exp\left(-\frac 14 \min\{\delta^2,\delta\} (n-1)\right).\label{eqprobcgeomgleichU}
  \end{align}
  \item For all $0 \le \delta \le 1$, 
  \begin{align}
  \Pr[X \le (1-\delta) \mu] &\le (1-\delta)^n \left(\frac{(1-\delta)(\mu-n)}{(1-\delta) \mu - n }\right)^{(1-\delta) \mu - n} \label{eqprobduvel1}\\
  &\le (1-\delta)^n  \exp(\delta n)\label{eqprobduvel2}\\
  &\le \exp\left(- \frac{\delta^2 n}{2- \frac 43 \delta}\right),\label{eqprobduvel3}
  \end{align}
  where the first bound reads as $p^n$ for $(1-\delta)\mu = n$ and as $0$ for $(1-\delta)\mu < n$. For $0 \le \delta < 1$ and $\lambda \ge 0$, we also have
  \begin{align}
  \Pr[X \le (1-\delta) \mu] &\le \exp\left(-\frac {2 \delta^2 p n}{(1-\delta)}\right),\label{eqprobCGA1}\\
  \Pr[X \le \mu - \lambda] &\le \exp\left(-\frac{2p^3 \lambda^2}{n}\right).\label{eqprobCGA2}
  \end{align}
  \end{enumerate}
\end{theorem}

Bounds~\eqref{eqprobCGA1} and~\eqref{eqprobCGA2} are interesting only for relatively large values of~$p$. Since part~(a) has been proven in~\cite{BaswanaBDFKN09}, we only prove~(b). The main idea in both cases is exploiting the natural relation between a sum of independent identically distributed geometric random variables and a sequence of Bernoulli events.
   
\begin{proof}
  Let $Z_1, Z_2, \dots$ be independent binary random variables with $\Pr[Z_i = 1] = p$ for all $i \in \N$. Let $n \le K \le \frac np$. Let $Y_K = \sum_{i=1}^K Z_i$. Then $X \le K$ if and only if $Y_K \ge n$. Consequently, by Theorem~\ref{tprobCMU},
  \begin{align*}
  \Pr[X \le K] & = \Pr[Y_K \ge n]\\
  & = \Pr\left[Y_K \ge \left(1 + \left(\frac{n}{Kp}-1\right)\right) E[Y_K]\right]\\
  & \le \left(\frac{Kp}{n}\right)^n \left(\frac{K - Kp}{K - n}\right)^{K-n}\\
  & \le \left(\frac{Kp}{n}\right)^n \exp(n - Kp)\\
  &\le \exp\left(\frac{p \lambda^2}{2K + \frac 23 \lambda}\right),
  \end{align*}
  where we used the shorthand $\lambda := \mu - K$ for the absolute deviation. From Theorem~\ref{tprobchernoffadditive01}, we derive
  \begin{align*}
  \Pr[X \le K] & = \Pr[Y_K \ge n]\\
  & = \Pr[Y_K \ge E[Y_K] + (n-Kp)]\\
  & \le \exp\left(-\frac{2(n-Kp)^2}{K}\right)\\
  & = \exp\left(-\frac{2p^2\lambda^2}{\mu-\lambda}\right) \le \exp\left(-\frac{2p^2\lambda^2}{\mu}\right).
  \end{align*}
  Replacing $K$ by $(1-\delta) \mu$ and $\lambda$ by $\delta \mu$ in these equations gives the claim.
\end{proof}

When the geometric random variables have different success probabilities, the following bounds can be employed.

\begin{theorem}\label{tprobcgeomungleich}
 Let $X_1, \ldots, X_n$ be independent geometric random variables with success probabilities $p_1, \dots, p_n>0$. Let $p_{\min} := \min\{p_i \mid i \in [1..n]\}$. Let $X:=\sum_{i=1}^n X_i$ and $\mu = E[X] = \sum_{i=1}^n \frac 1 {p_i}$. 
  \begin{enumerate}
  \item For all $\delta \ge 0$, 
  \begin{align}
  \Pr[X \ge (1+\delta)\mu] &\le \frac{1}{1+\delta} \, (1-p_{\min})^{\mu(\delta-\ln(1+\delta))}\label{eqprobcgeomUjanson1}\\
  &\le \exp(-p_{\min} \mu (\delta - \ln(1+\delta)))\label{eqprobcgeomUjanson2}\\ 
  &\le \left(1 + \frac{\delta \mu p_{\min}}{n}\right)^n \exp(-\delta\mu p_{\min})\label{eqprobcgeomUscheideler}\\
  &\le \exp\left(-\, \frac{(\delta \mu p_{\min})^2}{2n (1+\frac{\delta \mu p_{\min}}{n})} \right)\label{eqprobcgeomUweak}.
  \end{align}
  \item For all $0 \le \delta \le 1$,  
  \begin{align}
  \Pr[X \le (1-\delta) \mu] &\le (1-\delta)^{p_{\min} \mu} \exp(-\delta p_{\min} \mu)\label{eqprobcgeomLjanson}\\
  &\le \exp\left(- \frac{\delta^2 \mu \pmin}{2 - \frac 43 \delta}\right)\label{eqprobcgeomLmiddle}\\
  &\le \exp(- \tfrac 12 \delta^2 \mu \pmin)\label{eqprobcgeomLscheideler}.
  \end{align}
  \end{enumerate}
\end{theorem}
Estimates~\eqref{eqprobcgeomUjanson1} and~\eqref{eqprobcgeomUjanson2} are from~\cite{Janson17}, bound~\eqref{eqprobcgeomUscheideler} is from~\cite{Scheideler00}, and~\eqref{eqprobcgeomUweak} follows from the previous by standard estimates. This last bound, when applied to identically distributed random variables, is essentially the same as~\eqref{eqprobcgeomgleichU}.

For the lower tail bounds, \eqref{eqprobcgeomLjanson} from~\cite{Janson17} is identical to~\eqref{eqprobduvel2} for identically distributed variables. Hence~\eqref{eqprobduvel1} is the strongest estimate for identically distributed geometric random variables. Equation~\eqref{eqprobcgeomLjanson} gives~\eqref{eqprobcgeomLmiddle} via the same estimate that gives~\eqref{eqprobduvel3} from~\eqref{eqprobduvel2}. Estimate~\eqref{eqprobcgeomLscheideler} appeared already in~\cite{Scheideler00}.

Overall, it remains surprising that such useful bounds were proven only relatively late and have not yet appeared in a scientific journal.

The bounds of Theorem~\ref{tprobcgeomungleich} allow the geometric random variables to have different success probabilities, however, the tail probability depends only on the smallest of them. This is partially justified by the fact that the corresponding geometric random variable has the largest variance, and thus might be most detrimental to the desired strong concentration. If the success probabilities vary significantly, however, then this result gives overly pessimistic tail bounds, and the following result of Witt~\cite{Witt14} can lead to stronger estimates.

\begin{theorem}\label{tprobchernoffgeomwitt}
  Let $X_1, \dots, X_n$ be independent geometric random variables with success probabilities $p_1, \dots, p_n>0$. Let $X = \sum_{i=1}^n X_i$, $s = \sum_{i=1}^n (\frac 1 {p_i})^2$, and $p_{\min} := \min\{p_i \mid i \in [1..n]\}$. Then for all $\lambda \ge 0$, 
\begin{align}
  &\Pr[X \ge E[X] + \lambda] \le \exp\left(-\frac 14 \min\left\{\frac{\lambda^2}{s}, \lambda p_{\min}\right\}\right),\\
  &\Pr[X \le E[X] - \lambda] \le \exp\left(-\frac{\lambda^2}{2s}\right).
\end{align}  
\end{theorem}

In the analysis of randomized search heuristics, apparently we often encounter sums of independent geometrically distributed random variables $X_1, \dots, X_n$ with success probabilities $p_i$ proportional to $i$. For this case, the following result from~\cite[Lemma~4]{DoerrD18} gives stronger tail bounds than the previous result. See Section~\ref{secprobharmonic} for the definition of the harmonic number $H_n$.\index{harmonic number}
\begin{theorem}\label{tprobgeomharmonic}
  Let $X_1, \ldots, X_n$ be independent geometric random variables with success probabilities $p_1, \dots, p_n$. Assume that there is a number $C \le 1$ such that $p_i \ge C \frac in$ for all $i \in [1..n]$. Let $X = \sum_{i=1}^n X_i$. Then 
  \begin{align}
  &E[X] \le \tfrac 1C n H_n \le \tfrac 1C n (1 + \ln n),\\
  &\Pr[X \ge (1+\delta) \tfrac 1C n \ln n] \le n^{-\delta} \text{ for all } \delta \ge 0.
  \end{align}
\end{theorem}

As announced in Section~\ref{secprobdomruntime}, we now present a few examples where the existing literature only gives an upper bound for the expected runtime, but where a closer look at the proofs easily gives more details about the distribution, which in particular allows to obtain tail bounds for the runtime. We note that similar results previously (and before~\cite{Doerr18evocop}) have only been presented for the \oea optimizing the \leadingones test function~\cite{DoerrJWZ13} and for RLS optimizing the \onemax test function~\cite{Witt14}. Zhou et al.~\cite{ZhouLLH12} implicitly give several results of this type, however, the resulting runtime guarantees are not optimal due to the use of an inferior Chernoff bound for geometric random variables.

\begin{lemma}\label{lprobdomexamples}
\begin{enumerate}
	\item \label{itprobdomexom} The runtime $T$ of the \oea on the \onemax function is dominated by the independent sum $\sum_{i=1}^n \Geom(\frac{i}{en})$~\cite{DrosteJW02}. Hence $E[T] \le enH_n$ and $\Pr[T \ge (1+\delta) en \ln n] \le n^{-\delta}$ for all $\delta \ge 0$.
	\item \label{itprobdomexff} For any function $f: \{0,1\}^n \to \R$, the runtime $T$ of the \oea is dominated by $\Geom(n^{-n})$~\cite{DrosteJW02}. Hence $E[T] \le n^n$ and $\Pr[T \ge  \gamma n^n] \le (1-n^{-n})^{\gamma n^n} \le e^{-\gamma}$ for all $\gamma \ge 0$.
	\item \label{itprobdomexec} The runtime $T$ of the \oea finding Eulerian cycles in undirected graphs using perfect matchings in the adjacency lists as genotype and using the edge-based mutation operator is dominated by the independent sum $\sum_{i=1}^{m/3} \Geom(\frac{i}{2em})$~\cite{DoerrJ07}. Hence $E[T] \le 2emH_{m/3}$ and $\Pr[T \ge 2(1+\delta)em\ln \frac m3] \le (\frac m3)^{-\delta}$ for all $\delta \ge 0$.
	\item \label{itprobdomexso} The runtime of the \oea sorting an array of length $n$ by minimizing the number of inversions is dominated by the independent sum $\sum_{i=1}^{\binom{n}{2}} \Geom(\frac{3i}{4e\binom{n}{2}})$~\cite{ScharnowTW04}. Hence $E[T] \le \frac{4e}{3} \binom{n}{2} H_{\binom{n}{2}} \le \frac{2e}{3} n^2 (1+2\ln n)$ and $\Pr[T \ge (1+\delta) \frac{4e}{3} n^2 \ln n] \le \binom{n}{2}^{-\delta}$. Similarly, the runtime of the \oea using a tree-based representation for the sorting problem~\cite{DoerrH08} has a runtime satisfying $T \preceq \sum_{i=1}^{\binom{n}{2}} \Geom(\frac 1 {2e})$. Hence expected optimization time is $E[T] = 2e \binom{n}{2}$ and we have the tail bound $\Pr[T \ge (1+\delta) E[T]] \le \exp(- \delta^2 n / (2+2\delta))$. This example shows that a superior representation can not only improve the expected runtime, but also lead to significantly lighter tails (negative-exponential vs.\ inverse-polynomial).
	\item \label{itprobdomexss} The runtime of the multi-criteria \oea for the single-source shortest path problem in a graph $G$ can be described as follows. Let $\ell$ be such that there is a shortest path from the source to any vertex having at most $\ell$ edges. Then there are random variables $G_{ij}$, $i \in [1..\ell]$, $j\in [1..n-1]$, such that (i)~$G_{ij} \sim \Geom(\frac{1}{en^2})$ for all $i \in [1..\ell]$ and $j \in [1..n-1]$, (ii)~for all $j \in [1..n-1]$ the variables $G_{1j},\dots,G_{\ell j}$ are independent, and (iii)~$T$ is dominated by $\max\{\sum_{i=1}^{\ell} G_{ij} \mid j \in [1..n-1]\}$~\cite{ScharnowTW04,DoerrHK11}. Consequently, for $\delta = \max\{\frac{4 \ln(n-1)}{\ell-1},\sqrt{\frac{4 \ln(n-1)}{\ell-1}}\}$ and $T_0 := (1+\delta) \frac \ell p$, we have $E[T] \le (1+\frac{1}{\ln(n-1)}) T_0$ and $\Pr[T \ge (1+\eps) T_0] \le (n-1)^{-\eps}$ for all $\eps \ge 0$.
\end{enumerate}
\end{lemma}

\begin{proof}
We shall not show the domination statements as these can be easily derived from the original analyses cited in the theorem. Given the domination result, parts~\ref{itprobdomexom}, \ref{itprobdomexec}, and~\ref{itprobdomexso} follow immediately from Theorem~\ref{tprobgeomharmonic}. Part~\ref{itprobdomexff} follows directly from the law of the geometric distribution.
 
To prove part~\ref{itprobdomexss}, let $X_1, \dots, X_\ell$ be independent geometrically distributed random variables with parameter $p = \frac{1}{en^2}$. Let $X = \sum_{i=1}^\ell X_i$. Let $\delta = \max\{\frac{4 \ln(n-1)}{\ell-1},\sqrt{\frac{4 \ln(n-1)}{\ell-1}}\}$. Then, by~\eqref{eqprobcgeomgleichU}, $\Pr[X \ge (1+\delta) E[X]] \le \exp(-\frac 12 \frac{\delta^2}{1+\delta} (\ell-1)) \le \exp(-\frac 14 \min\{\delta^2,\delta\} (\ell-1)) \le \exp(-\frac 14 \frac{4 \ln(n-1)}{\ell-1} (\ell-1)) = \frac 1 {n-1}$. For all $\eps > 0$, again by~\eqref{eqprobcgeomgleichU}, we compute 
\begin{align*}
\Pr[X \ge (1+\eps)(1+\delta) E[X]] &\le \exp\left(-\frac 12 \frac{(\delta+\eps+\delta\eps)^2}{(1+\delta)(1+\eps)} (\ell-1)\right) \\
&\le \exp\left(-\frac 12 \frac{\delta^2(1+\eps)^2}{(1+\delta)(1+\eps)} (\ell-1)\right) \\
&\le \exp\left(-\frac 12 \frac{\delta^2}{1+\delta} (\ell-1)\right)^{1+\eps}  \le (n-1)^{-(1+\eps)}.
\end{align*}
Let $Y_1, \dots, Y_{n-1}$ be random variables with distribution equal to the one of $X$. We do not make any assumption on the correlation of the $Y_i$, in particular, they do not need to be independent. Let $Y = \max\{Y_i \mid i \in [1..n-1]\}$ and recall that the runtime $T$ is dominated by $Y$. Let $T_0 = (1+\delta)E[X] = (1+\delta)\frac{\ell}{p}$. Then $\Pr[Y \ge (1+\eps)T_0] \le (n-1) \Pr[X \ge (1+\eps)T_0] \le (n-1)^{-\eps}$ by the union bound (Lemma~\ref{lprobunionbound}). By Corollary~\ref{corprobtaile}, 
\begin{align*}
E[Y] \le \left(1 + \frac 1 {\ln(n-1)}\right) T_0.
\end{align*}

\end{proof}

We note that not all classic proofs reveal details on the distribution. For results obtained via random walk arguments, e.g., the optimization of the short path function SPC$_n$~\cite{JansenW01}, monotone polynomials~\cite{WegenerW05}, or vertex covers on paths-like graphs~\cite{OlivetoHY09}, as well as for results proven via additive drift~\cite{HeY01}, the proofs often give little information about the runtime distribution (an exception is the analysis of the needle and the \onemax function in~\cite{GarnierKS99}).

For results obtained via the average weight decrease method~\cite{NeumannW07} or multiplicative drift analysis~\cite{DoerrG13algo}, the proof also does not give information on the runtime distribution. However, the probabilistic runtime bound of type $\Pr[T \ge T_0 + \lambda] \le (1-\delta)^\lambda$ obtained from these methods implies that the runtime is dominated by $T \preceq T_0 -1 + \Geom(1-\delta)$.

\subsection{Tail Bounds for the Binomial Distribution} 

For binomially distributed random variables, tail bounds exist which are slightly stronger than the bounds for general sums of independent $0,1$ random variables. The difference are small, but since they have been used in the analysis of randomized search heuristics, we briefly describe them here. 

In this section, let $X$ always be a binomially distributed random variable with parameters $n$ and $p$, that is, $X = \sum_{i=1}^n X_i$ with independent $X_i$ satisfying $\Pr[X_i = 1] = p$ and $\Pr[X_i = 0] = 1-p$. The following estimate seems well-known (e.g., it was used in~\cite{JansenJW05} without proof or reference). Gie{\ss}en and Witt~\cite[Lemma~3]{GiessenW17} give an elementary proof via estimates of binomial coefficients and the binomial identity. We find the proof below more intuitive.

\begin{lemma}\label{lprobbino}
  Let $X \sim \Bin(n,p)$. Let $k \in [0..n]$. Then \[\Pr[X \ge k] \le \binom{n}{k} p^k.\]
\end{lemma}

\begin{proof}
  For all $T \subseteq [1..n]$ with $|T|=k$ let $A_T$ be the event that $X_i = 1$ for all $i \in T$. Clearly, $\Pr[A_T] = p^k$. The event ``$X \ge k$'' is the union of the events $A_T$ with $T$ as above. Hence $\Pr[X \ge k] \le \sum_{T} \Pr[A_T] = \binom{n}{k} p^k$ by the union bound (Lemma~\ref{lprobunionbound}).
\end{proof}

When estimating the binomial coefficient by $\binom{n}{k} \le (\frac{en}{k})^k$, which often is an appropriate way to derive more understandable expressions, the above bound reverts to equation~\eqref{eqprobCMUstrongA}, a slightly weaker version of the classic multiplicative bound~\eqref{eqprobCMUstrong}. Since we are not aware of an application of Lemma~\ref{lprobbino} that does not estimate the binomial coefficient in this way, its main value might be its simplicity.

%
%
%

The following tail bound for the binomial distribution was shown by Klar~\cite{Klar00}, again with elementary arguments. In many cases, it is significantly stronger than Lemma~\ref{lprobbino}. However, again we do not see an example where this tail bound would have improved an existing analysis of a randomized search heuristics. 

\begin{lemma}\label{lprobklar}
  Let $X \sim \Bin(n,p)$ and $k \in [np..n]$. Then \[\Pr[X \ge k] \le \frac{(k+1)(1-p)}{k+1-(n+1)p} \Pr[X=k].\]
\end{lemma}

Note that, trivially, $\Pr[X=k] \le \Pr[X \ge k]$, so it is immediately clear that this estimate is quite tight (the gap is at most the factor $\frac{(k+1)(1-p)}{k+1-(n+1)p}$). With elementary arguments, Lemma~\ref{lprobklar} gives the slightly weaker estimate 
\begin{equation}
  \Pr[X \ge k] \le \frac{k-kp}{k-np} \,\Pr[X=k],
\end{equation}
which appeared also in~\cite[equation~(VI.3.4)]{Feller68}. For $p = \frac 1n$, the typical mutation rate in standard-bit mutation, Lemma~\ref{lprobklar} gives   
\begin{equation}
  \Pr[X \ge k] \le \frac{k+1}{k} \,\Pr[X=k].
\end{equation}
Writing Lemma~\ref{lprobbino} in the equivalent form $\Pr[X \ge k] \le (\frac{1}{1-p})^{n-k} \Pr[X=k]$ and noting that $(\frac{1}{1-p})^{n-k} \ge \exp(p(n-k))$, we see that in many cases Lemma~\ref{lprobklar} gives substantially better estimates.

Finally, we mention the following estimates for the probability function of the binomial distribution stemming from~\cite{Bollobas01}. By summing over all values for $k' \ge k$, upper and lower bounds for tail probabilities can be be derived.
\begin{theorem}\label{tprobpbino}
  Let $X \sim \Bin(n,p)$ with $np \ge 1$. Let $h > 0$ such that $k = np+h \in \N$. Let $q = 1-p$.
  \begin{enumerate}
  \item If $hqn/3 \ge 1$, then\\ $\Pr[X = k] < \frac{1}{\sqrt{2\pi p q n}} \exp(-\frac{h^2}{2pqn} + \frac{h}{qn} + \frac{h^3}{p^2 n^2})$.
  \item If $k < n$, then\\ $\Pr[X = k] > \frac{1}{\sqrt{2\pi p q n}} \exp(-\frac{h^2}{2pqn} - \frac{h^3}{2q^2 n^2} - \frac{h^4}{3 p^3 n^3} - \frac{h}{2pn} - \frac{1}{12k} - \frac{1}{12(n-k)})$.
  \end{enumerate}
\end{theorem}

\newcommand{\etalchar}[1]{$^{#1}$}


}
\end{document}